\documentclass[12pt]{amsart}
\usepackage{latexsym,amssymb,amsmath,amsthm}
\usepackage{comment,graphicx,color}
\usepackage{morefloats}
\usepackage[section]{placeins}
\definecolor{red1}{rgb}{1,0.9,0.9}
\definecolor{blue1}{rgb}{0.9,0.9,1}
\definecolor{green1}{rgb}{0.9,1,0.9}
\definecolor{yellow1}{rgb}{1,1,0.9}
\definecolor{yellow2}{rgb}{1,1,0.8}
\newtheorem{thm}{Theorem}
\newtheorem{coro}[thm]{Corollary}
\newtheorem{lemma}[thm]{Lemma}

\theoremstyle{definition}

\let\paragraph\subsection

\title{Combinatorial manifolds are Hamiltonian}
\author{Oliver Knill}
\date{June 14, 2018}
\address{
        Department of Mathematics \\
        Harvard University \\
        Cambridge, MA, 02138
        }

\begin{document}
\maketitle

\begin{abstract}
Extending a theorem of Whitney of 1931 we prove that all connected 
$d$-graphs are Hamiltonian for $d \geq 1$. A $d$-graph is a type of combinatorial manifold 
which is inductively defined as a finite simple graph for which every unit sphere is a $(d-1)$-sphere.
A $d$-sphere is $d$-graph such that removing one vertex renders the graph contractible. A
graph is contractible if there exists a vertex for which its unit sphere and the graph without that
vertex are both contractible. These inductive definitions are primed with the assumptions that the 
empty graph $0$ is the $(-1)$-sphere and that the one-point graph $1$ is the smallest contractible graph.
The proof is constructive and shows that unlike for general graphs, the complexity of the 
construction of Hamiltonian cycles in d-graphs is polynomial in the number of vertices of the graph. 
\end{abstract}

\section{Introduction}

\paragraph{}
Hassler Whitney \cite{Whitney1931} proved in 1931:
"Given a planar graph composed of elementary triangles in which there
are no circuits of 1,2 or 3 edges other than these elementary triangles, there exists a circuit which
passes through every vertex of the graph". This theorem implies that "every $4$-connected maximally planar graph is
Hamiltonian". The later can be restarted as "every $2$-sphere is Hamiltonian". A $2$-sphere
is a finite simple graph in which every unit sphere $S(x)$ is a $1$-sphere, a circular graph of length $4$ or higher
and. A $2$-ball is a $2$-sphere with a removed vertex. The class of $2$-spheres
is exactly the class of $4$-connected maximal planar graphs. 
Bill Tutte \cite{Tutte1955} extended Whitney's theorem to all 4-connected planar graphs. An overview
over graph theory and history, see \cite{HararyGraphTheory,BM,BiggsLloydWilson,Kumar2016}. 

\paragraph{}
As Whitney has pointed out already, the interest in triangulated graphs was fueled by the problem of
coloring maximal planar graphs as this is the class of planar graphs which 
are hardest to color among all planar graphs. Whitney's paper actually was part of the thesis which 
he wrote under the guidance of George Birkhoff in the area of graph chromatology. 
The 2-sphere restatement of Whitney's theorem is graph theoretical.
The original statement is too, but it only becomes so after rephrasing ``planar" within 
graph theory using the Kuratowski theorem of 1930. 
What is interesting about the sphere formulation of Whitney's theorem is that it opens the 
door to generalize coloring and Hamiltonian graph problems to higher dimensions, where it is 
detached from embedding questions. 

\paragraph{}
The Whitney result can be extended not only to other dimensions, we can also
extend the class by using the language of simplicial complexes. 
The $1$-skeleton complex of a finite abstract simplicial complex is a graph. When is such a graph
Hamiltonian? The concept of finite abstract simplicial complex is due to Dehn and Heegard from 1907 
\cite{DehnHeegaard}. One can restrict the question to shellable complexes, a 
concept put forward in a combinatorial setting by Hadwiger and Mani from 1972 \cite{HadwigerMani1972}. 
The phenomenon of non-shellability of spheres has been discovered already in 1924 in a topological setting. 

\paragraph{}
The result can be extended to {\bf generalized d-graphs}, finite simple graphs $G$ for which the boundary 
alone is a $(d-1)$ graph without boundary and for which every inner vertex has a unit sphere which is a
$(d-1)$ sphere and such that for every inner vertex $x$ near the boundary there is a triangle $(x,a,b)$,
where $(a,b)$ is an edge on the boundary. This goes beyond d-graphs in that we don't need to have interior
points for example. Every $d$-graph with or without boundary is a generalized $d$-graph. 

\paragraph{} 
A simple example of a generalized $d$-graph is the $d$-simplex itself, in which we look at the $d-1$ dimensional
skeleton complex. It has no interior point and the boundary is naturally a $(d-1)$-sphere. 
A wheel graph in which all boundary edges are stellated is {\bf not} a generalized
d-graph. While the boundary as a $1$-graph, it has an interior point which is not 
accessible from the boundary. It is not Hamiltonian because the Hamiltonian boundary 
curve which is forced can not be detoured to the interior. But if one exterior triangle 
is removed, then the inner point becomes accessible and the graph becomes Hamiltonian. 
A three dimensional example is the stellated solid octahedron, where the octahedron 
is the interior unit sphere. But in that three dimensional case, not only the accessibility
fails, also some unit spheres are neither $2$-spheres, nor $2$-balls or $2$-simplices. 

\paragraph{}
How can the Hamiltonian property fail for manifold-like graphs? It can be due to lack of connectivity but
this mechanism does not explain non-existence of Hamiltonian cycles of higher dimensional Goldner-Harary graphs. 
The classical Goldner-Harary graph is a non-Hamiltonian graph which is the 1-skeleton graph of a shellable 
3-dimensional complex for which all unit spheres are contractible shellable 2-complexes. 
What happens in those examples even after Barycentric refinements is that the unit spheres lack Euclidean
properties. But there are other mechanisms which make the Hamiltonian property fail: in two dimensions,
we even can have non-Hamiltonian graphs like the stellated wheel graph which is non-Hamiltonian. 

\paragraph{}
The failure of being Hamiltonian is rare when
building up complexes randomly by successive simplex aggregation to the boundary. 
We observe that one culprit for the failing Hamiltonian property often is the loss of the manifold structure. This happens
for the Goldner-Harary graph $G$, obtained from shellable complexes first in dimension $3$. The two dimensional boundary 
complex has a Barycentric refinement $G_1$, where unit spheres are not balls or spheres
as needed for a manifold with boundary. 
For the Barycentric refinement $G_1$ of $G$ the boundary becomes Hamiltonian. 
Still, there remain defects from being Euclidean:
there are still unit spheres which are neither 2-spheres nor generalized $2$-balls. 

\paragraph{}
The question on how to define a $d$-dimensional ``sphere" combinatorially came up already
when Ludwig Schl\"afli generalized the Euler gem formula $v-e+f=2$ from 2 dimensional to 
higher dimensional spheres and proving that the Euler characteristic of a d-dimensional 
sphere is $v_0-v_1+v_2- \dots \pm v_d=1+(-1)^d$ \cite{Schlafli}. 
Without stating explicitly, Schl\"afli assumed in 1852 
that spheres are shellable simplicial complexes. He called a contractible space a "polyschematische Figur" and 
a discrete sphere a "sph\"arisches Polyschem". Also Euler implicitly assumed a shell decomposition in two dimensions. 
However, there are triangulations of Euclidean spheres which are no more shellable, and the first 
examples appeared already in 1924.  
Non-shellable spherical simplicial complexes are puzzles, where the individual cells are so entangled within 
each other that it is not possible to take out one. Related is the existence of $2$-dimensional 
structures like the Bing house or dunce hat which are homotopic to a point but which can not be 
shrunk to a point. The Bing house for example produces a non-shellable 3-ball. 

\paragraph{}
In the wake of the foundational crisis in mathematics in the early 
20th century, Herman Weyl \cite{Weyl1925} asked for ways to define spheres within finite mathematics.
He refers in his expos\'e to \cite{DehnHeegaard} and mentions that a combinatorial definition of ``Kugelschemata" 
in dimension $4$ and larger is not yet known. 
Defining a sphere as a triangulation of a Euclidean sphere is not possible, as the later uses the notion of Euclidean space
tapping into mathematics which a finitist would not accept. Also the notions of triangulations can be quite
complicated and can contain higher dimensional simplices than the space itself. Today, one seems to be less worried about the 
consistency of the foundations of mathematics but cares more about realizing mathematical 
structures faithfully in a computer without approximations or to see a result in the context of 
reverse mathematics trying to localize which axioms or axiom systems are needed to prove a phenomenon. In any case,
combinatorial mathematics can be seen as a shelter in which one could retreat in the unlikely event of an emerging
inconsistency in the ZFC axiom system. 

\paragraph{}
Thanks to notions of homotopy put forward by J.H.C. Whitehead \cite{Whitehead} which was adapted to the 
discrete later on, one can define a d-sphere as a simplicial complex which has 
locally a unit sphere which is a $(d-1)$ sphere and such that removing a point renders the complex contractible.
(We use this term here as a synonym with collapsible and
would use ``homotopic to the 1 point complex" to describe more general homotopies). 
Such notions have emerged in the discrete in the 1990ies in Fisk's theory \cite{Fisk1977a}, in discrete Morse theory 
of Forman \cite{forman95} or digital topology of Evako \cite{I94}, refined in \cite{CYY}. The 
definitions of $d$-spheres and $d$-graphs are done with the goal to be as simple as possible. 

\paragraph{}
The assumption of having a ``contractible graph after puncturing a sphere" restricts the class of spheres
slightly because of the existence of non-shellable complexes and the existence of non-contractible complexes 
which are only homotopic to a point. 
It turns out not to be a substantial problem as the Barycentric refinement of any 
simplicial complex is always a Whitney complex and the Whitney complex of a d-sphere is always shellable. 
The troubles with non-shellability are washed away with one Barycentric refinement. The troubles with 
non-Euclidean properties however (as we see here by linking it to a Hamiltonian property) can not be 
removed through Barycentric refinements. In any case, the graph theoretical definition of spheres 
therefore are almost no lack of generality from a topological point of view and Schl\"afli's 
shellability assumption was not such a big restriction of generality after all when seen in the Weyl program
to finitize Euclidean geometry. 

\paragraph{}
Graphs not only have the benefit of being intuitively more accessible than simplicial complexes which are sets of sets.
Having a simple and solid graph theoretical definition of what a $d$-sphere is, one can ask questions
which usually are asked only in graph theory. There are two theorems about $2$-spheres which are particularly elegant: 
the first is the 4-color theorem which is equivalent to the statement that 
{\it 2-spheres have chromatic number $3$ or $4$}, where 3-colorability is equivalent to the property of being
Eulerian. Second, there is the already mentioned Whitney theorem on Hamiltonian graphs stating that 
{\it 2-spheres are Hamiltonian}. The observation that the class of 2-spheres is the same as the 
class of maximally planar, 4-connected graphs is not difficult to verify \cite{knillgraphcoloring}. The fact that
looking at maximally planar 4-connected graphs covers all of the coloring of planar graphs has been pointed out
by Whitney already \cite{Whitney1931}. It must have been folklore wisdom already then, but not yet formulated
graph theoretically. 

\paragraph{}
$1$-spheres are cyclic graphs with $4$ or more nodes. They have chromatic number $2$ or $3$ and are 
Hamiltonian. Having Whitney covering the two dimensional case, it is natural to ask
what happens with the coloring statement in dimension $3$ and higher. 
In Fisk theory \cite{Fisk1977a,Fisk1977b}, the coloring problem is reduced to coloring problems in higher 
dimensions, leading to the conjecture that a $d$-sphere has chromatic number 
$d+1$ or $d+2$ \cite{knillgraphcoloring,knillgraphcoloring2}. Discrete projective planes can have chromatic 
number $5$ and indeed they are not embeddable in an Euclidean 3-space. For tori, the chromatic number is 3,4 or 5
\cite{AlbertsonStromquist}. 

\paragraph{}
The sphere coloring statement bounding the chromatic number of $d$-spheres 
between $d+1$ and $d+2$ still remains to be 
proven in dimensions larger than $2$ but it looks good: one only has to verify that it is 
possible to do successive edge refinements in the interior of a $(d+1)$-ball to render it minimally 
colorable by $d+2$ colors, coloring so the boundary sphere. The existence result is constructive also in
higher dimensions but it will lead in particular to a constructive proof of the
4 color theorem. The second line of question about Hamiltonian cycles is easier and covered in 
the current article. Unlike in 2 dimensions, where the existence of a Hamiltonian path is unrelated to coloring
statements in higher dimensions because the coloring argument in two dimensions relies on a
Jordan curve statement. 

\paragraph{}
The main reason why the main result proven here is true is that the
boundary of a $d$-graph with boundary is a union of connected $(d-1)$ graphs which by induction are Hamiltonian. 
We can use calculus to divide the discrete manifold up and expose the interior to the boundary. 
In a larger class of ``generalized d-graphs", we can make the reduction enhanced by removing 
``boundary discs" until we have no interior points near the boundary any more. 
The Hamiltonian path on the boundary can then be extended to the still remaining inner points near the boundary. 
In some sense, we can extend a Hamiltonian path on a $(d-1)$-graph with boundary to a ``fattened surface" and
achieve with cutting out holes that all the graph becomes a fattened surface. 
Any pure simplicial d-dimensional complex is Hamiltonian as long one can identify the boundary $\delta G$ as a $(d-1)$
dimensional complex without boundary and such that $\delta G$ allows access to the nearby interior points. 

\paragraph{}
In general, we see that the construction of Hamiltonian paths can be done quite effectively 
for $d$-graphs: take a $d$-graph with boundary with $n$ vertices,
use a function $f$ to chop it into two $d$-graphs $\{ f<0 \}$ and $\{f>0 \}$ with boundary, where each part 
contains about half the vertices. Now do the Hamiltonian construction in both parts and build 
a single bridge in the joining part $f=0$ which does not contain any vertices
to get a Hamiltonian path in the entire graph. Now repeat the construction on the smaller parts. 
This informal description needs to be made more precise as the cutting has to be smooth enough. We will 
bypass this difficulty by cutting away only small holes. This still gives a polynomial complexity in the 
number of vertices of the graph. Using the already known fact that in the case $d=2$, the complexity is
linear, we get linearity in all dimensions. 

\paragraph{}
In the proof we cut explicitly near a point in the interior producing a cavity.
If the graphs have become so small that it is no more possible to cut a hole into them, we can edge away from the boundary
until we have a generalized $d$-graph without interior. Now we can cover the graph with a Hamiltonian
path because the Hamiltonian path on the ``fattened boundary" covers everything.  Also the construction of the 
Hamiltonian path on that boundary is constructive. Start with a Hamiltonian path on the $(d-1)$ surface given by the theorem
for dimension $d-1$. A still constructive ``Swiss cheese" approach is then to cut small cavities away 
to reach the near interior points and then to join the various Hamiltonian paths on different connected
components using bridges.  

\section{Terminology}

\paragraph{}
A finite simple graph $G=(V,E)$ is {\bf Hamiltonian} if there exists a simple closed 
circuit which visits every vertex of $G$ exactly once. 
A {\bf finite abstract simplicial complex} is a finite set of sets which is closed under the operation 
of taking finite non-empty subsets. The {\bf $1$-skeleton graph} of a complex $G$ is the one dimensional sub-complex
of all sets in $G$ of cardinality $1$ or $2$. It naturally is a finite simple graph. 
A simplicial complex is called {\bf Hamiltonian} if 
its $1$-skeleton graph is Hamiltonian. We can think of the operation ``{\bf Skeleton}" either as a map from 
``simplicial complexes" to ``finite simple graphs" or as a projection from arbitrary complexes to 
$1$-dimensional complexes, the association with projection being justified by the identity
${\rm Skeleton} \circ {\rm Skeleton} = {\rm Skeleton}$. 

\paragraph{}
The {\bf Whitney complex} of a finite simple graph $(V,E)$ is the finite abstract simplicial complex $G$ 
in which the sets are the vertex sets of the complete sub-graphs of $(V,E)$. We often identify a graph 
with its Whitney complex. Similarly, we associate to a Whitney complex its skeleton graph. 
The two operations ``Whitney" and ``Skeleton" are not inverses of each other in general. 
The Whitney complex of the $1$-skeleton graph of a complex $G$ is often different from the complex, 
but Whitney complexes have the property that they 
are entirely encoded by the graph alone. The identity ${\rm Skeleton}({\rm Whitney})( (V,E) ) = (V,E)$ holds
for all finite simple graphs $(V,E)$. Even so in general ${\rm Whitney}({\rm Skeleton})(G)$ is not equal to $G$
the two constructs are natural operations between ``simple graphs" and ``simplicial complexes".

\paragraph{}
A graph is called {\bf pure} of dimension $d$, if its Whitney complex is pure: this means
that the {\bf facets} (=maximal simplices) in its Whitney complex $G$ all have the same dimension $d$. 
A finite simple graph is called {\bf shellable}, 
if its Whitney complex $G$ is a shellable complex: first formulated in a purely combinatorial setting in
\cite{BruggesserMani}, an abstract finite simplicial complex is called {\bf shellable} if 
there exists an ordering $x_1, \dots, x_n$ of its facets such that for every $k \geq 1$, the intersection 
$H_k$ of $x_k$ with $\bigcup_{j=1}^{k-1} x_j$ defines a shellable $(d-1)$ complex. 
By definition, a shellable complex of dimension $d$ produces a {\bf $d$-connected}
graph meaning that the 1-skeleton graph remains connected after removing $d$ vertices or less. 

\paragraph{}
A finite simple graph $G$ is called a {\bf d-sphere} if it is either the $(-1)$-sphere 
$0=(\emptyset,\emptyset)$ or if every unit sphere $S(x)$ is a $(d-1)$-sphere and the graph $G-x$ without $x$ 
is contractible for some vertex $x$. A graph $G$ is {\bf contractible} if it is either the one-point graph $1=K_1$ 
or if there exists a vertex $x$ such that $G-x$ and $S(x)$ are both contractible. 
A {\bf $d$-ball} is a graph $H=G-x$ which is obtained from a $d$-sphere $G$ by removing a single vertex $x$.
A {\bf $d$-graph} is a finite simple graph for which every unit sphere is a $(d-1)$ sphere. If every unit sphere
is either a $(d-1)$-sphere or a $(d-1)$ ball, then we deal with a {\bf $d$-graph with boundary}.

\paragraph{}
We use here a definition of spheres in higher dimensions \cite{knillgraphcoloring} which relates closely to 
notions put forward by to Evako \cite{Evako2013} who also formulated discrete versions of 
Whitehead homotopy \cite{I94}. Both the discrete homotopy 
as well as the discrete sphere definitions have been refined over time. We use the term {\bf contractible} 
even so this is often called ``collapsible". For us, contractible and collapsible are synonyms. Examples like
the dunce hat or Bing house show that they are not the same than {\bf homotopic to the $1$-point graph}
which is the property that there is a graph $H$ such that $H$ is both contractible and equivalent to $G$ by 
contraction steps, where a vertex with all connections to $S(x)$ is allowed to be removed if $S(x)$ is 
contractible. 

\paragraph{}
The {\bf dual graph} of a $d$-graph has the maximal $d$-simplices (=facets) as vertices
and connects two different maximal simplices if they intersect in a $(d-1)$ simplex.
The dual graph of a $d$-graph is always zero or one-dimensional. 
For a connected $d$-graph of dimension $d \geq 2$ different from a simplex, it is always one-dimensional. The proof:
the dual graph does not contain a triangle because if it would, the entire triplet of $d$-simplices would belong
to a larger dimensional simplex.  \\
Examples: 1) the dual graph of a wheel graph with boundary $C_n$ is $C_n$.  \\
2) the dual graph of the 16 cell, which is a 3-sphere 
has as a dual graph the {\bf tesseract}.  \\
3) The dual graph of the octahedron graph is the cube graph. 
Both the cube graph as well as the tessaract are not $d$-graphs as the unit spheres are not spheres. 
Also, the dual graph of the cube contains triangles because more than 2 maximal simplices intersect. 

\paragraph{}
While nice triangulations of Euclidean spheres produce shellable complexes, there exist already for $d=3$ 
non-shellable triangulations of spheres \cite{Lickorish91}. Examples of non-shellable balls
are the Furch ball of 1924, the Newman ball of 1926 or the Rudin's ball \cite{Rudin1958}. 
Ziegler's ball $G$ with 10 vertices and 21 facets appears to be the smallest known so far. 
The Barycentric refinement $G_1$ of any of these examples $G$ however is a shellable complex.

\paragraph{}
A $(d-1)$ simplex is a {\bf boundary face} in a $d$-graph $G$ if it is contained 
in only one of the maximal $d$-simplices. We call a $d$-graph $G$ without boundary 
{\bf strongly Hamiltonian}, if there exists a Hamiltonian cycle which has an edge in 
each of facets, the maximal simplices. For $d \geq 3$ and $d$-graphs with boundary we ask more and 
say that a $d$-graph $G$ with boundary is {\bf strongly Hamiltonian}, if there is a Hamiltonian path 
which visits every $d$ simplex in an edge and every boundary $(d-1)$ simplex in an edge. 
The wheel graph shows that we can not yet enforce this strong Hamiltonian property for $2$-graphs with boundary
This is similar as for $d=1$ graphs with boundary, where we can not even get the Hamiltonian property. 

\paragraph{}
A cyclic graph is always strongly Hamiltonian because the graph 
itself is the Hamiltonian path. The wheel graphs are $2$-balls. They are Hamiltonian but not strongly Hamiltonian.
There is a  Hamiltonian path around the boundary which detour once to the center.
In dimensions $d=3$ and higher, we always have the strong Hamiltonian property for $d$-graphs. 
The {\bf boundary complex} of a $d$-graph $G$ is the $1$-complex of 
the $(d-1)$ boundary complex of $G$. In order to have a good induction proof, we need to extend the
class of $d$-graphs and the class of $d$-graphs with boundary a bit. This is done next. 

\paragraph{}
A {\bf generalized $d$-graph} is a graph which either is a $d$-graph with or without boundary
or a graph for which all spheres are $(d-1)$ balls, $(d-1)$ spheres or $(d-1)$ simplices and for which every 
interior vertex $x$ in distance $1$ to the boundary is part of a triangle $xab$, with an edge $(a,b)$
at the boundary. \\

The last condition disallows $d$-subgraphs which have only vertices in common with the boundary. An example
is the stellated wheel graph. This condition is only really necessary in the case $d=2$. In higher dimensions,
we get the Hamiltonian property also without it. The reason is that in accessible inner points can still be 
reached from the outside via a detour, but we don't bother with that for now and leave the condition in all 
dimensions. \\

{\bf Examples:} \\

1) The diamond graph $C_2(2)$ is an example of a generalized $2$-ball because the sphere condition is satisfied and
no interior point exist. \\
2) The stellated square is not a generalized $2$-ball. The inner sphere has no common edge with the boundary. 
After doing edge refinement in the interior of the stellated square, we can get examples of arbitrary 
large 2-balls or (after digging holes) get 2-graphs which are isolated $2$-subgraphs. \\
3) The Goldner-Harary graph has no interior but it fails the unit sphere condition. 
There are unit sphere which are neither $(d-1)$ balls nor $(d-1)$ simplices. \\
4) The $3$-dimensional ``Avici" graph $C_3(2)$ obtained by gluing two $K_4$ along a $K_3$ is a generalized $3$-ball
because removing the $2$-dimensional interior triangle and the two $3$-simplices produces a $2$-dimensional complex. \\
5) For any $d \geq 1$, the $d$-simplex is a generalized $d$-ball. 
After removing the interior $d$-dimensional facet, it has a 
Barycentric refinement which is a $(d-1)$ sphere. \\
6) The {\bf $d$-dimensional prism} is an example of a generalized $d$-ball. It has no interior It has a 
$(d-1)$ dimensional simplex roof and a $(d-1)$ dimensional simplex floor.
In $d$ dimensions, it is a union of $d$ simplices of dimension $d$. It is a generalized ball without interior point.
7) The stellated solid octahedron is not a generalized $3$-ball because some unit spheres are not balls, spheres 
or simplices.

\paragraph{}
Let us look a bit more at the threshold, where the Hamiltonian property fails:
we can verify directly by induction that a shellable graph of dimension $d \geq 2$ for which the dual graph 
is a path graph is Hamiltonian. It only amounts to show by induction that the strong Hamiltonian property can get extended 
when adding a new simplex. The just mentioned Avici graph is an example. More generally,
any {\bf cyclic polytope}, where $d$-simplices are attached to each other in a linear way such that their 
intersection is a $d-2$ simplex are examples of shellable complexes which are Hamiltonian. \\

\paragraph{}
When trying to extend the result from linear complexes to complexes, where the dual 
graph is a tree, we in general lose the Hamiltonian property. The Goldner-Hariri example is prototypical 
of this type. The reason is that we might run out of vertices which can use to lift the Hamiltonian path to all 
the newly attached vertices. In the Gardner-Harary graph, this happens, leading to a failure of
the Euclidean structure and Hamiltonian graph property, even after Barycentric refinement. 
We are fine in the case of a $d$-prism, as there, we have no interior
and the boundary sphere possesses a Hamiltonian path, because it is a $2$-sphere. 
An other prototype example is the case of an interior sphere which does not meet an other interior unit sphere
nor any boundary. The smallest such example is the stellated square or stellated octahedron. 

\paragraph{}
There are shellable $d$-complexes for any $d \geq 1$ for which the dual graph contains triangles but 
these complexes fail the property of being ``generalized d-graphs" in one way or the other. 
They are either not Whitney complexes like for the 1-skeleton complex $C_3$ of $K_3$ 
for which the dual complex is again isomorphic to $C_3$. 
An other example in one dimension is a claw graph or the cube graph, 
where three edges come together. 
Or then, we have like in windmill graph, three triangles hinged together at a
single edge, which is not Hamiltonian. 

\paragraph{}
Here are the incidence relations: $d$-spheres $\subset$ $d$-graphs $\subset$ generalized $d$ graphs $\subset$ 
Whitney complexes $\subset$ finite abstract simplicial complexes. Similar incidence relations hold in the 
case of boundary. 

\section{The theorem}

\paragraph{}
The main result is about d-graphs, a class of graphs which behave like ``combinatorial manifolds". 
We avoid the later term, as it used in different ways in the literature, 
also in non-combinatorial set-ups like in the context of PL manifolds. 
$d$-graphs are finite objects, finite simple graphs. 

\begin{thm}[Hamiltonian manifold theorem]
Every connected $d$-graph is strongly Hamiltonian for $d \geq 1$.
\end{thm}

\paragraph{}
The proof will use induction with respect to dimension $d$. 
We will prove at the same time the analogue statement for 
$d$-graphs with boundary. The case $d=2$ is in some sense the hardest and most subtle case,
which is covered by Whitney already (except for the strong Hamiltonian property, which however
can be replaced by Hamiltonian connectednedss). 
Our proof is new and proves the result for slightly more general {\bf generalized $d$-graphs}.  
In dimension $d=2$, ti does not cover all of Tutte's theorem but goes beyond in the 
sense that $2$-graphs do not have to be planar. 

\begin{thm}[Boundary version]
Every connected $d$ graph with boundary is Hamiltonian for $d>1$ and
strongly Hamiltonian for $d>2$. 
\end{thm}

\paragraph{}
The boundary version fails in dimension $d=1$ because connected $1$-graphs with boundary are 
path graphs which are are not Hamiltonian. The boundary version fails the strong Hamiltonian property
in dimension $d=2$ because a Hamiltonian path would have to stay on the boundary and
cover every edge. This already does not work for the wheel graphs. 

\paragraph{}
The problem to find Hamiltonian cycles in a graph is an NP complete problem. It is known that 
for 2-spheres, the computation of Hamiltonian paths is linear \cite{CN}. This generalizes. 
Our proof shows: 

\begin{coro}
The construction of Hamiltonian cycles in d-graphs is polynomial in the number of vertices.
\end{coro}

Actually, using the linear complexity result \cite{CN} in two dimensions, we get also in 
general that the construction of Hamiltonian cycles in $d$-graphs is linear in the number
of vertices: there exists a constant $C$ such that a Hamiltonian cycle can be computed in 
less than $C n$ computation steps, if $n$ is the number of vertices in the graph $G$. 
For recent improvements of the computation of Tutte paths, see \cite{SchmidSchmidt}.

\section{Proof} 

\paragraph{}
The proof goes by induction with respect to dimension $d \geq 1$ and is done for 
generalized $d$ graphs with boundary for $d \geq 2$. 
The statement is true for all connected $1$-graphs as these are cyclic graphs. Assume the claim
has been shown for all $(d-1)$-graphs without boundary. Take a generalized $d$-graph with boundary,
remove the interior parts until every vertex is close to the boundary, then build a Hamiltonian
path by taking the one on the boundary and extending it to the vertices close to the boundary. 

\paragraph{}
To cave out holes, we can use a function $f$ to 
cut away a small $d$-ball $B$ around a 
{\bf strong interior point}, an interior point which has only interior points nearby. In each
case, we chose the strong interior point in distance $2$ the already established boundary. 
Given a strong interior point $x$, we remove that point and add the new sphere boundary $S(x)$ 
to the boundary. The Hamiltonian path on that boundary $S(x)$ can make a detour to the center $x$ so 
that the unit ball $B(x)$ is still Hamiltonian. A prismatic bridge, (as explained below) 
combines this with the rest of the graph. 

\paragraph{}
After getting rid of all strong interior points, we look 
by either caving away points such that $G \setminus x$ is still a generalized $d$-graph. 
For these points there exists an interior point $y$ near the boundary
such that $S(y) \cap \delta G$ consists of one point only. 
By adding all unit spheres $S(x)$ to the boundary, removing $x$,
we achieve that every interior point which remains is now near the boundary, meaning
of distance $1$ to the boundary.  \\

(There is a difficulty of having non-accessible interior points
which only appears in dimension 2. This is a case already covered in the literature.
In two dimensions, some interior points $x$ can remain which
can not be reached from the boundary $\delta G$. This happens if the unit sphere of $x$ touches
the boundary everywhere in one point $y$ only. We just absorb such point $y$ to a neighboring hole.
The hole can grow more. The interior could get a triangle or aggregate more points as
long as the subgraph generated by these points remain Hamiltonian. It can be $K_3$, the
diamond graph or any path shellable complex or generalized $2$-graph. In the proof we only need
one possible extension.)

\paragraph{}
Finally, we reach out for the interior points which are in distance $1$ to the boundary after all 
these modifications. These points can be reached by making a small detour from the 
boundary. To organize this, we drill a small hole into $G$ by removing some edges in the $(d-1)$ dimensional 
boundary, exposing the interior point to the boundary and making it a boundary point but 
without changing the property of being a generalized $d$-graph. In the case $d=2$, this means
removing an edge in the $1$-dimensional boundary (and so a triangle). 
In the case $d=3$, this means the interior edge and the two triangles 
of a diamond graph producing a hole with boundary $C_4$ in the two dimensional boundary 
(this also removes three tetrahedra). In the case $d=4$, this means
removing the interior edges and the three tetrahedrons of a prism at the boundary, 
building a three dimensional ball cavity into the boundary 
which means to remove four $4$-simplices and exposes the center point to the boundary.  \\

Here is the explicit construction algorithm of a Hamiltonian path for a connected d-graph $G$:

\begin{itemize}
\item If a boundary $\delta G$ is available, identity the vertices of this subgraph $\delta G$.
      It can be the empty set. 
\item Find an interior point $x$ in distance $2$ from the boundary $\delta G$.
   Remove $x$ from $G$ and join the unit sphere $S(x)$ to the boundary $\delta G$. 
   Find a Hamiltonian path in the unit ball $B(x)$. Build a bridge from $S(x)$ to $\delta G$ 
   using a prism bridge. Now we have a new generalized $d$-ball with boundary. 
   Repeat this step 2) until no interior points of distance 2 to the boundary can be 
   found any more. 
\item If we should end up with an interior point $y$ in distance $1$ to the boundary
   which is ``boxed in" in the sense that its unit sphere $S(y)$ 
   has no edge in common with $\delta G$, we make the interior of one of the 
   neighboring holes larger to 
   absorb some point $z \in S(y)$, allowing $y$ to be reached. This case only can occur in 
   dimension $2$ because for $d \geq 2$, the intersection of two neighboring spheres in $\delta G$
   has positive dimension and so contains an edge. 
\item Now we have a graph for which every still remaining interior
   point can be reached from the boundary. Build Hamiltonian paths in each
   of the connected components of the boundary and connect them with bridges.
   This path covers now everything except the interior points in distance $1$ from the
   boundary. 
\item Find an interior point $x$ in distance 1 from the boundary. 
   Build a detour so have that interior point is included into the Hamiltonian path. 
   Continue with step 5) until no interior points are available any more. Now we have
   a Hamiltonian path. 
\end{itemize}

\section{Calculus} 

\paragraph{}
In this section, we describe the cutting procedure using a function. 
It is a basic property of $d$-graphs without boundary that the level surface $\{ f=c \}$ is always either empty or a 
$(d-1)$ graph without boundary. Also for a $d$-graph with boundary, the level surface $\{ f=c \}$ is either empty 
or a $(d-1)$ graph with or without boundary. These things were observed in \cite{KnillSard}. It makes use of the
definition that the {\bf level set} $\{ f=c \}$ is the set of all simplices in $G$ on which $f$ changes
sign, then build the graph in which these simplices are the vertices and where two are connected, if one
is contained in the other. What is achieved with this definition, that for $c$ different from the range
of $f(V)$ on the vertex set $V$, the level surfaces $\{ f=c \}$ are nice ``discrete manifolds", hence 
the name ``discrete Sard theorem". We deal here with generalized $d$-graphs, but apply the cutting procedure 
only in the case when $\{ f = c \}$ is in the interior of a smaller $d$-graph with boundary. In our
proof we will only apply the situation where $f=1$ on a single vertex $x$ and $f$ is negative 
on every other vertex and where the ball $B_2(x)$ is a $d$-graph with boundary 
(and not only a generalized graph with boundary). 

\paragraph{}
For a general locally injective function $f$, while $\{ f=c \}$ is nice, the sets 
$\{ f<c \}$ and $\{ f>c \}$ 
do not need any more to be $d$-graphs with boundary. 
If $f$ is equal to $1$ on a vertex $x$ and negative everywhere else, then $f>0$ is a single point. 
for some small but positive $\epsilon$, then $\{ f=c \}$ is a $(d-1)$ sphere but $\{ f>c \}$ is the $1$-point graph $K_1$. 
Actually, the vertex set of $\{ f<0 \}$ can be an arbitrary subset of the vertex set. We want to use functions to 
cut a $d$-graph into smaller pieces, establish the Hamiltonian property in both parts then join the two paths
using a single bridge. This ``divide and conquer" strategy which provides a computational way
to find Hamiltonian paths quickly for $d$-graphs. 

\paragraph{}
A level surface $\{ f = c \}$ of a function $f$ on the vertex set of a $d$-graph is called {\bf smooth} 
if for all vertices $x$, the induced graphs from $S^-(x) = \{ y \in S(x) \ ; | \; f(y) <c \}$ 
and $S^+(x) = \{ y \in S(x) \ ; | \; f(y) >c \}$ are either $(d-1)$-balls, $(d-1)$ spheres, a $k$-simplex or empty.
A connected $1$-graph $G$ is just a cyclic graph $C_n$ for $n \geq 4$. Every locally injective function
on a $1$-graph is smooth. Almost by definition, we see that a level surface $\{ f = c \}$ is smooth 
if and only if both $\{ f>c \}$ and $\{ f < c \}$ are $d$-graphs with boundary or then a $k$-simplex with $k \leq d$.
The same definition can be extended to $d$-graphs to $d$-graphs with boundary. The smoothness notion is
not really needed in the proof but it can help in process to use smooth surfaces.

\section{Cyclic polytopes}

\paragraph{}
In this section, we cover a class of generalized d-graphs which are Hamiltonian. Before using the
Swiss cheese strategy, we used to cut up the graph into lower dimensional layers, eventually reaching
cyclic polytopes (stripes). This strategy seems to work also to get a Hamiltonian paths, even more
effectively, but we have not shown that we can always cut things up like that. The idea was to 
cut up a graph $G$ into two $d$-graphs $G_1=\{ f>0 \},G_2 = \{ f<0 \}$, then 
find the Hamiltonian paths in each, then glue them together. 

\paragraph{}
When finding a Hamiltonian path using such subdivision, there will a moment when we can no more divide
up a $d$-graph into smaller $d$-graphs using smooth cuts. There are arbitrary long $2$-balls 
already which can not be cut into smaller balls. An example is a Birkhoff Diamond: start with a
wheel graph $W_5$ with boundary $C_5$, then chose a boundary point then take the union with 
an other wheel graph centered there. This produces a $2$-ball with $2$ interior points and boundary $C_6$. 
It is a $2$-ball but it can not be cut into smaller paths which are $2$-balls. 

\paragraph{}
A {\bf cyclic polytope} is an example of a {\bf linear shellable complex}. This means that it is a 
shellable simplicial complex $G$ of dimension $d \geq 1$ for which the dual graph
is a finite path graph. Cyclic polytopes play an important role in combinatorial topology because
of McMullen's upper bound theorem telling that these polytopes are the ones which
maximize the volume (number of facets) with a fixed number of vertices. 
We have the following observation:

\begin{lemma} 
If $G$ is a $d$-graph which is strongly Hamiltonian, then adding a new vertex over one of its
facets is still Hamiltonian. It is strongly Hamiltonian on the newly added faces but might lose
the strong Hamiltonian property on other faces. 
\end{lemma} 

\begin{proof}
Just take an edge in each of the faces $F_j$ and "lift" it up to the new vertex $x_j$. 
As the intersection of $F_j$ with $F_k$ has no edge, we don't run into the problem of 
double booking a wire to two extensions. 
\end{proof} 

\paragraph{}
This immediately implies that we can continue to make extensions on a new set $F_j$ of newly 
generated faces. 

\begin{coro}
Every path shellable complex is Hamiltonian. 
\end{coro} 

\paragraph{}
The reason why not all shellable simplicial complexes are Hamiltonian is that during the extension
the strong Hamiltonian property has been lost on other faces. 
The prototype examples where the Hamiltonian property fails are Goldner-Harary in any dimensions, 
or stellated cross polytopes, also in any dimension. 

\section{Bridges}

\paragraph{}
A locally injective function on the vertex set $V$ of a graph $G=(V,E)$ is also called a {\bf coloring}. 
The discrete Sard theorem \cite{KnillSard} assures that the level set $\{ f = c \}$ 
in a $d$-graph is a $(d-1)$-graph, as long as $c$ is different from the range of $f$. The level set is
defined as the subgraph of the Barycentric refinement of $G$ generated by the set of simplices in $G$
on which $f$ changes sign. If we think of the level surface as water we need to build bridges 
between the part $\{ f<c\}$ and the part $f>c$. In dimension $d=2$, the level surface is built 
by the edges and triangles on which $f$ changes sign. Each connectivity component is a circular graph. 

\paragraph{}
Using functions, we solve the problem of ``cutting a d-graph into smaller parts" $\{ f<c \}$ and $\{ f>c \}$, where
$f=c$ does not contain any vertices but is naturally associated to a $d-1$ complex. 
As long as $c$ is different from the values $f(V)$, we get two regions $A,B$ for which prismatic bridges
exist between the two parts $A$ and $B$. Such a bridge is given by a $d-1$ boundary simplex in $f<c$ and a $d-1$ boundary
simplex in $f>c$ bounding a $d$-ball in dimension $d>1$. Now, if a Hamiltonian path in $A$ visits the 
first $d-1$ face in $A$ and a Hamiltonian path in $B$ visits a d-1 face in $B$, then we can rewire
the prism so that the two Hamiltonian paths merge to a single Hamiltonian path and this path still has
the property that it visits every d-simplex as the simplices in the bridge $f=c$ are bound by 
$(d-1)$-dimensional faces which contain edges. 

\paragraph{}
Here is the lemma which assures that we can join Hamiltonian paths in $A = \{ f<c \} $ and 
$B = \{ f>c \}$. Assume the bridge connects the simplex $A'$ in $A$ with the simplex $B'$ in $B$. 

\begin{lemma}[Bridge lemma]
For every edge $e=(a,b)$ in a $(d-1)$-face $X$ in $A'$, there is a $(d-1)$ face $Y$ in $B'$ 
such that for every edge $f=(c,d)$ in $Y$, there is a quadrilateral containing $a,b,c,d$. 
\end{lemma}

\begin{proof}
Given a face $X$. Look at all edges which go from $X$ to $B$. This naturally defines
a simplicial complex where the edges are the vertices. It has a single maximal $(d-1)$-simplex.
This defines an injective map from $X$ to $B$. The image is a simplex $Y$ in $B$. 
Now pick any edge $y=(c,d)$ in $Y$.  We claim that there is a quadrangle containing $a,b,c,d$. \\

First of all, the distance between $a$ and $c$ (or $d$) is smaller or equal than $2$. 
The reason is that $c,d$ are in the union of the unit spheres of $a$ and $b$. 
That shows that $a$ is connected either to $c$ or $d$.  \\

Now, assume $a$ is connected to $c$. We have to show that $b$ is connected to $d$. 
If there were no direct connection from $b$ to $d$, then because of the maximal
distance $2$, $d$ has to be connected to $c$. But then the closure of $bcd$ is
there and $bd$ is in.
\end{proof}

\paragraph{}
The set of all $d$-simplices with edges in $X$ or $Y$ or connections between $X$ and $Y$ is the 
{\bf prism} generated by the faces $X$ and $Y$. The edges connecting $X$ and $Y$ actually are
the vertices of a $(d-1)$-simplex in the hypersurface $f=c$. The prism itself is not a $d$-ball
but it is a generalized $d$-ball. \\

\paragraph{}
In the case $d=2$, for example, the prism is a diamond graph made of $2$ triangles and some 
unit balls are simplices. 
In the case $d=3$, the prism is a cylindrical prism made of three tetrahedra, which represent the
three edges of a triangle of the level surface $\{ f=c \}$. 

\section{Shellability}

\paragraph{}
In this section, we point out that the Whitney complex of a $d$-sphere $G$ or 
$d$-ball is a shellable complex in all dimensions $d \geq 1$. It is a result which is 
not needed in the Swiss cheese proof of the Hamiltonian manifold theorem. This means that
there exists a sequence $x_k$ of $d$-simplices such that $\{ x_j \}_{j=1}^n$ 
generates the complex $G$ and $G_k=\bigcup_{j=1}^{k-1} X_j \cap X_k$ is a shellable $(d-1)$ complex for 
every $k$, if $X_k$ is the complex generated by $\{x_k\}$: it is the
smallest set of non-empty sets which contains all $x_k$ and is closed under the operation of 
taking finite non-empty subsets. 

\paragraph{}
The usual suspects of non-shellable triangulations of spheres are not Whitney complexes of graphs.
The smallest $d$-ball $G$ due to Ziegler for example produces a
5 dimensional Whitney complex from its 1-dimensional skeleton complex.
The Barycentric refinement $G_1$ of $G$ is then a Whitney complex and indeed it is shellable. \\
All Barycentric refinements of a simplicial complex whose
Euclidean realization gives a $d$-sphere or $d$-ball is shellable
because a Barycentric refinement is a Whitney complex.
But all these  non-shellable cases are washed away with Barycentric refinements.

\begin{lemma}
The Whitney complex of a $d$-sphere is a shellable complex.
\end{lemma}

\begin{proof}
We use induction with respect to dimension $d$. 
For a $(-1)$ sphere, which is the empty graph, this is the assumption. Assume the claim holds for $d-1$-spheres. 
Take a $d$ sphere. By definition, there is a vertex $v$ such that $S(v)$ is a $(d-1)$-sphere
and $G-v$ is contractible. As $S(x)$ is shellable, there is a sequence of simplices 
$y_1, \dots, y_m$ building up $S(x)$ in such a way that $y_1 \cup \dots \cup y_{k-1} \cap y_k$ is
shellable. Let $x_j$ be the cone extension of $y_j$ with the vertex $v$. 
Now, $x_1 \cup \dots \cup x_{k-1} \cap x_k$ is shellable as the result is just the cone extension
with $v$. 
\end{proof}

\paragraph{}
Similarly, we have: 

\begin{lemma}
The Whitney complex of a $d$-ball is shellable.
\end{lemma}
\begin{proof}
To verify the statement for $d$-balls, let us assume, the statement is proven 
for all $(d-1)$-balls. We make now a second induction with respect
to the number $n$ of simplices in the ball. For one simplex, there is
nothing to show. Given a $d$-ball $G$. By definition of contractibility, 
there is a vertex $x$ such that both $S(x)$ and $G-x$ are contractible. 
Now, $H=G-x$ is smaller and so shellable. As $S(x)$ is a smaller dimensional 
ball, also the unit sphere $S(x)$ is shellable. Let $y_1, \dots, y_m$ be the simplices in $S(x)$.
The complex $G$ is obtained by adding the simplices $x_j$ to $G-x$, where 
$x_j$ is the cone extension of $y_j$ with $v$. We can now get $G$ from $G-x$
by successively adding the simplices $x_1, \dots, x_m$.
\end{proof}

\paragraph{}
This of courses does not go over to higher genus $2$-complexes already. The simplicial complex
of a triangulation of a $2$-torus for example is not shellable because at some point 
we have to add a triangle $x_k$ which intersects $\bigcup_{j=1}^{k-1} x_j$ in a
$0$-dimensional or $1$-dimensional non-pure complex violating in both cases the
shellability condition. 

\paragraph{}
Every cone extension $H=G+x$ of a $d$-sphere is a $(d+1)$-ball: the reason is that this is the 
removal of a vertex from a suspension of $G$. 
The graph $H$ has a boundary $\delta H=G$ which is the unit sphere $S(x)$.

\section{Questions}

\paragraph{}
A graph is called {\bf Hamiltonian connected} if for every
$a \in V, b \in V$, which need not to be distinct, there exists a Hamiltonian path from $a$ to $b$.
Are all $d$-spheres Hamiltonian connected? This is known for $d=2$ \cite{Thomassen83}. 
We see that we can force a Hamiltonian cycle to go
through a particular edge $e$. The reason is induction and the fact that we can force
$e$ to be in a prism, then force the boundary simplices. 

\paragraph{}
In two dimensions, there are already stronger results like Hamiltonian connectedness
which assure that we can force a Hamiltonian path to go through any edge. 
Is there a simple argument showing, that under the existence
of a Hamiltonian path, there is one which intersects every simplex in an edge? 
There is a simple argument showing that for any 2-sphere and any edge $e=(a,b)$, we can find 
a Hamiltonian path going trough that edge: remove all edges going from $a$ to some $y$ in $S(a)$
except for two adjacent edges. We have still a Hamiltonian path. It has to go through $e$ now
and this path is still a Hamiltonian path for the original sphere. 

\paragraph{}
We don't know the structure of the set of all Hamiltonian cycles in a $d$-sphere. 
Is the class of Hamiltonian cycles of a $d$-sphere we can define two paths $C,D$ to be
connected if there exists a local deformation from one to the other. 
Local means that we can only change both graphs in a prismatic bridge which 
connects two $d-1$ simplices. It is natural to ask whether this space of Hamiltonian
paths is connected. 

\paragraph{}
Edge refinements preserve the Hamiltonian property in $d$-graphs:
given an edge $(a,b)$ in $d$-dimensional complex: The intersection $S(x) \cap S(y)$ is 
either a $d-2$ sphere or a $d-2$ ball. A Hamiltonian path passing
through either $a$ or $b$, can make a detour of the central point $c$ splitting
$(a,b)$ and have the Hamiltonian property for the edge refined complex. \\
Can we find a similar construction for Barycentric refinements? Is there a constructive
way for $d$-graphs to get from a Hamiltonian path of $G$ to a Hamiltonian path in the 
Barycentric refinement?  
For general graphs, the Barycentric refinement procedure does not preserve the property of being
Hamiltonian.

\paragraph{}
Traditionally, coloring questions have been considered in parallel to the property 
of being Hamiltonian even so there are no general ways to get from a Hamiltonian
path a vertex coloring (one can get sometimes get face colorings of planar maps from
a Hamiltonian path on the boundary). Anyway, here is again the question for $d$-graphs: 
is the chromatic number of a $d$-sphere always 
equal to $d+1$ or $d+2$? For $d=2$ it is the 4-color theorem. 

\paragraph{}
A generalized {\bf Barnette question} asks whether the dual graph of a $d$-sphere with
chromatic number $d+1$ is Hamiltonian. The condition of having a minimal coloring 
replaces the Eulerian property in two dimension. For $d=2$, the problem is the classical 
Barnette problem which is still open. 

\section{Illustrations}

\begin{figure}
\scalebox{0.12}{\includegraphics{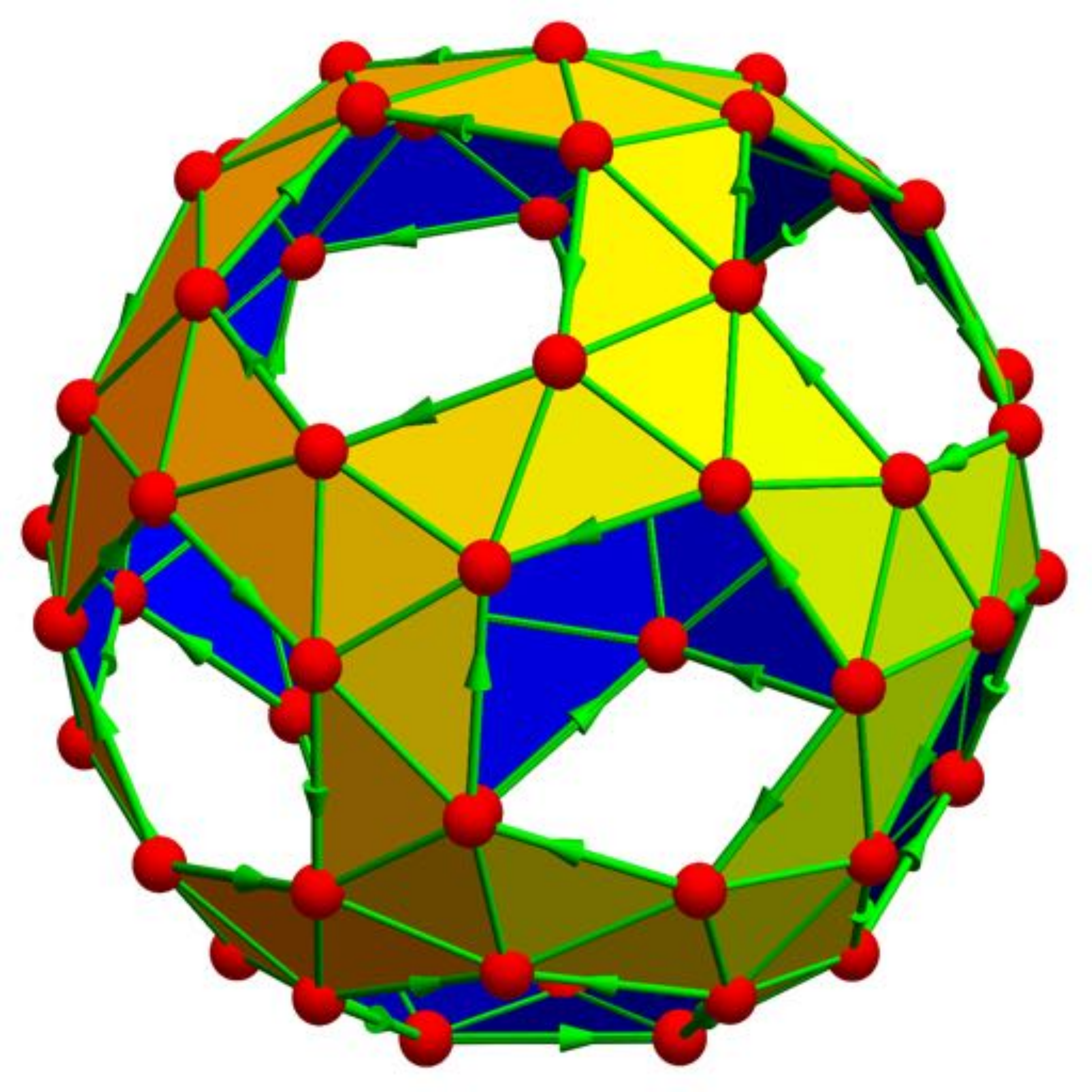}}
\scalebox{0.12}{\includegraphics{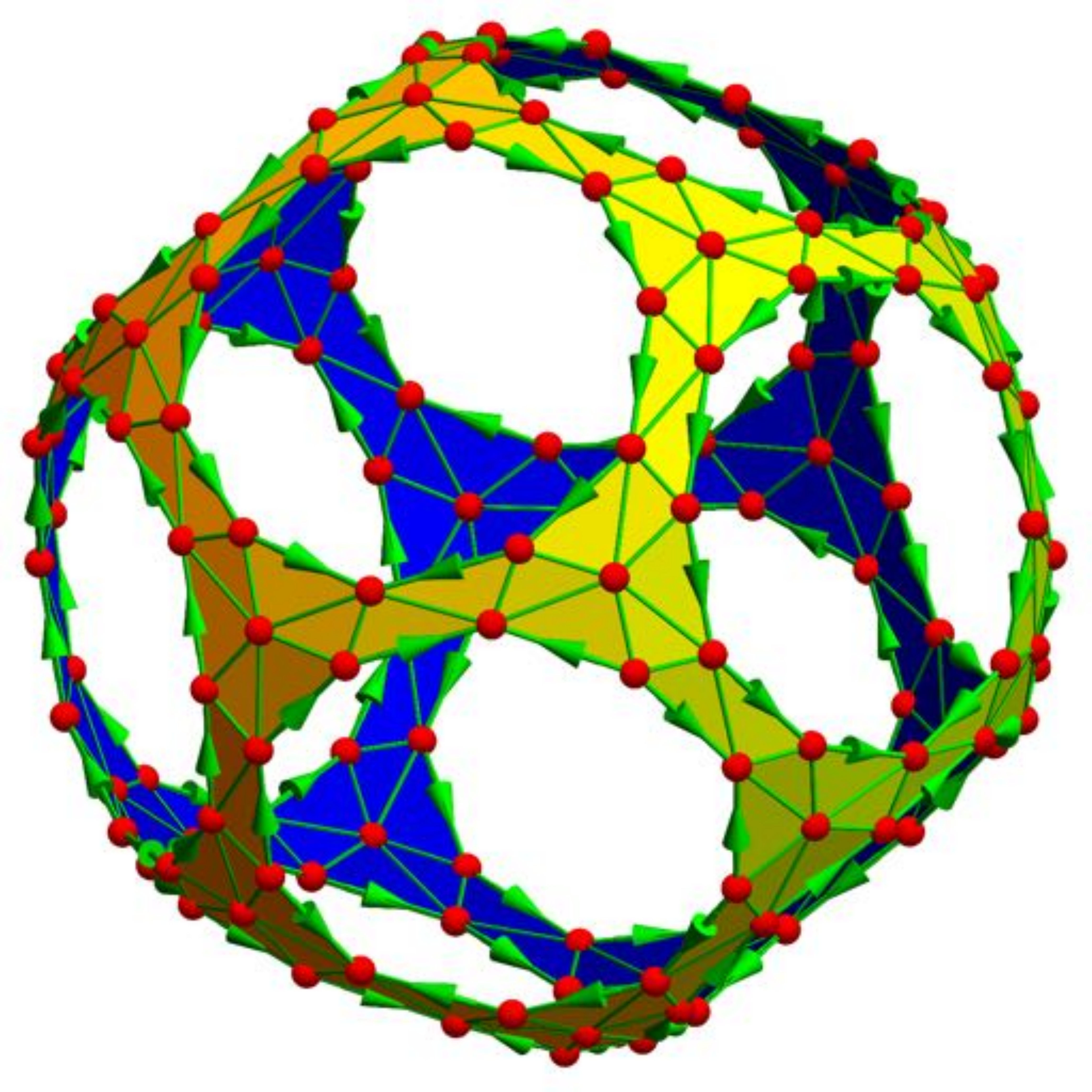}}
\caption{
\label{Stokes}
The first graph is a generalized $d$-graph without interior point. 
The second graph is a generalized $d$-graph with some interior
points. By deleting some edges we can expose them to the boundary.
}
\end{figure}

\begin{figure}
\scalebox{0.12}{\includegraphics{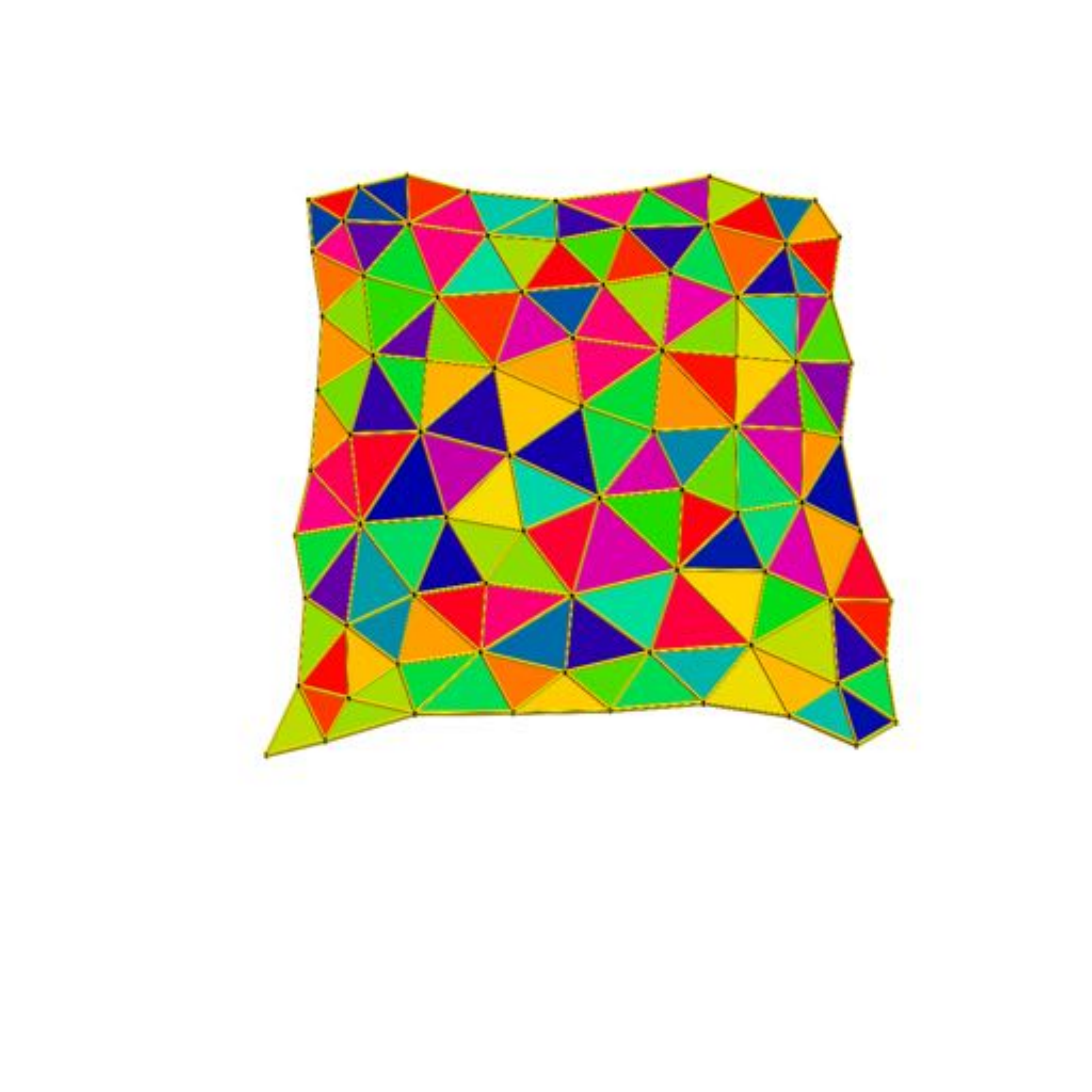}}
\scalebox{0.12}{\includegraphics{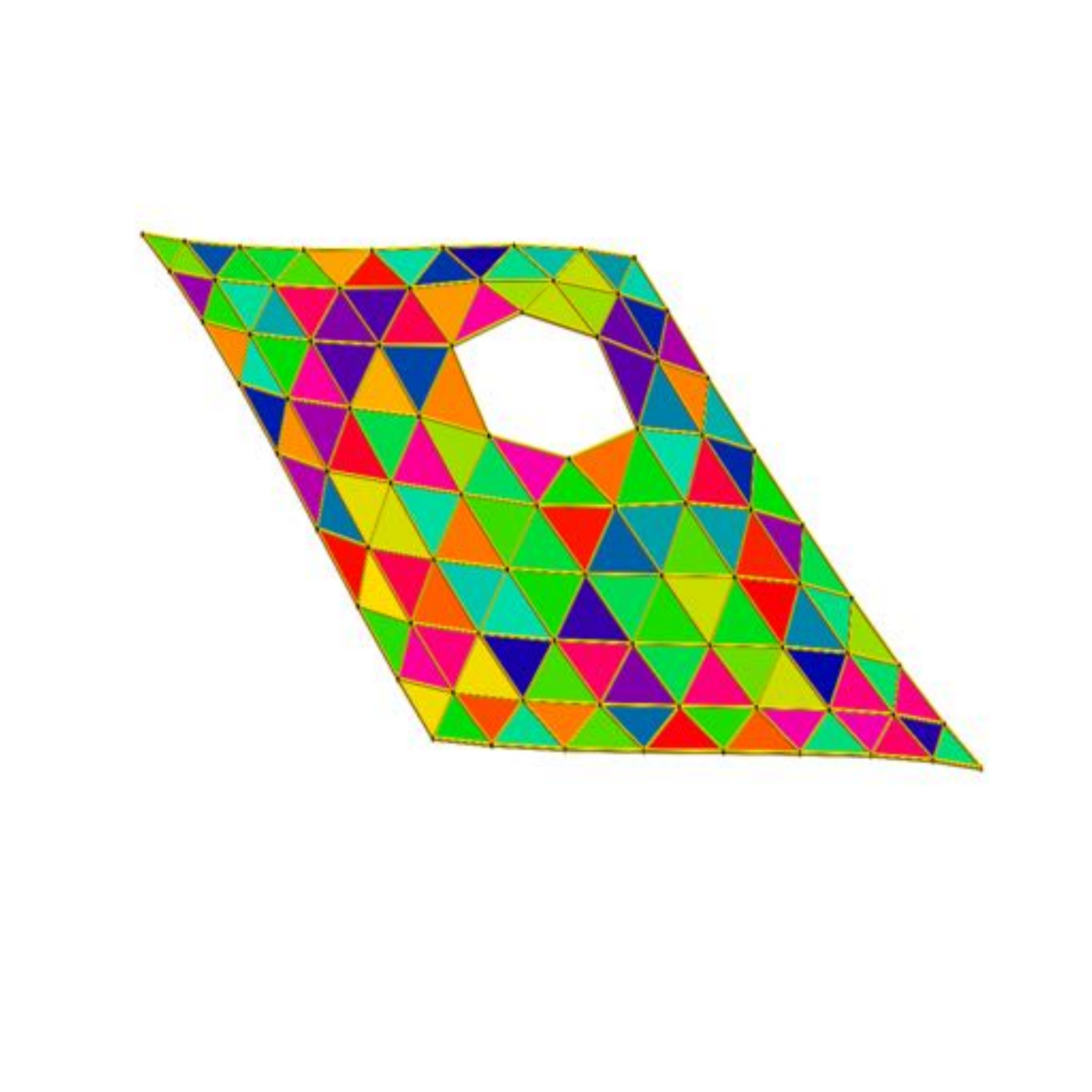}}
\caption{
\label{2Dim}
We see two generalized $2$-graphs with boundary. Removing the extreme
triangles which only intersect the rest on an edge produces $2$-graphs. 
These graphs are Hamiltonian. 
}
\end{figure}

\begin{figure}
\scalebox{0.12}{\includegraphics{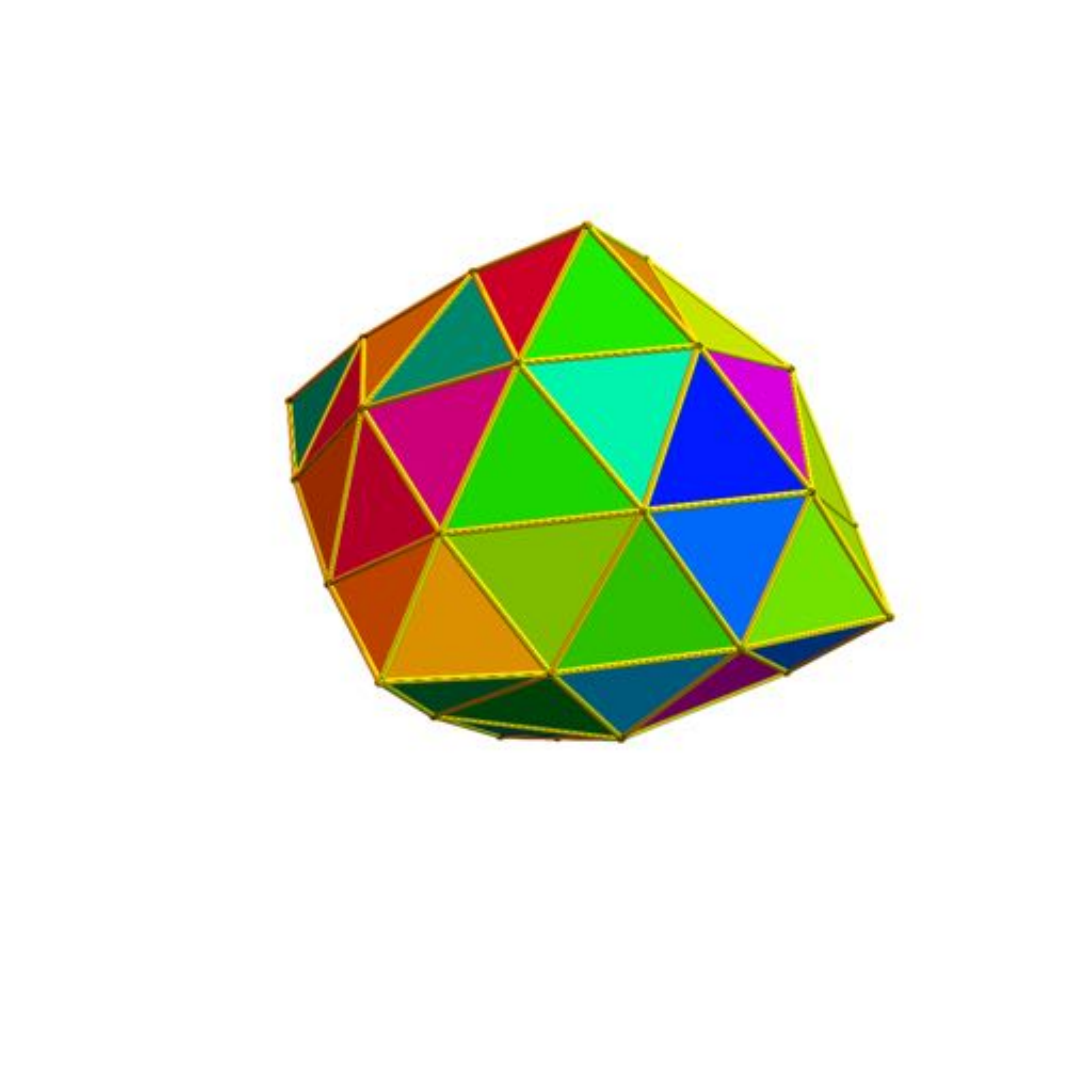}}
\scalebox{0.12}{\includegraphics{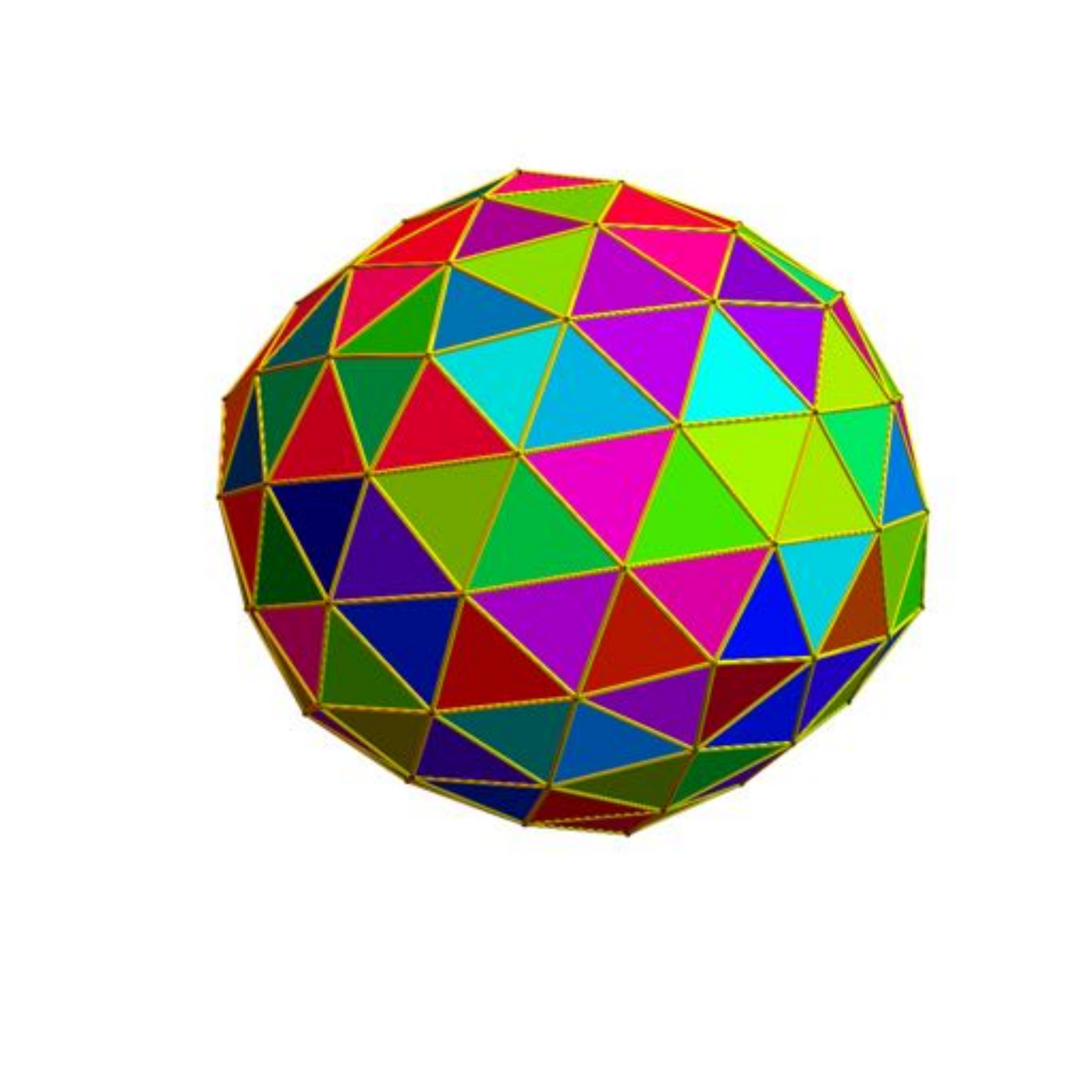}}
\caption{
\label{Bary}
Two $2$-spheres, tame refinements of the octahedron and the 
icosahedron. Both are best shown to be Hamiltonian by cutting
them up into smaller parts, covering the parts with Hamiltonian
paths, then joining them. 
}
\end{figure}

\begin{figure}
\scalebox{0.12}{\includegraphics{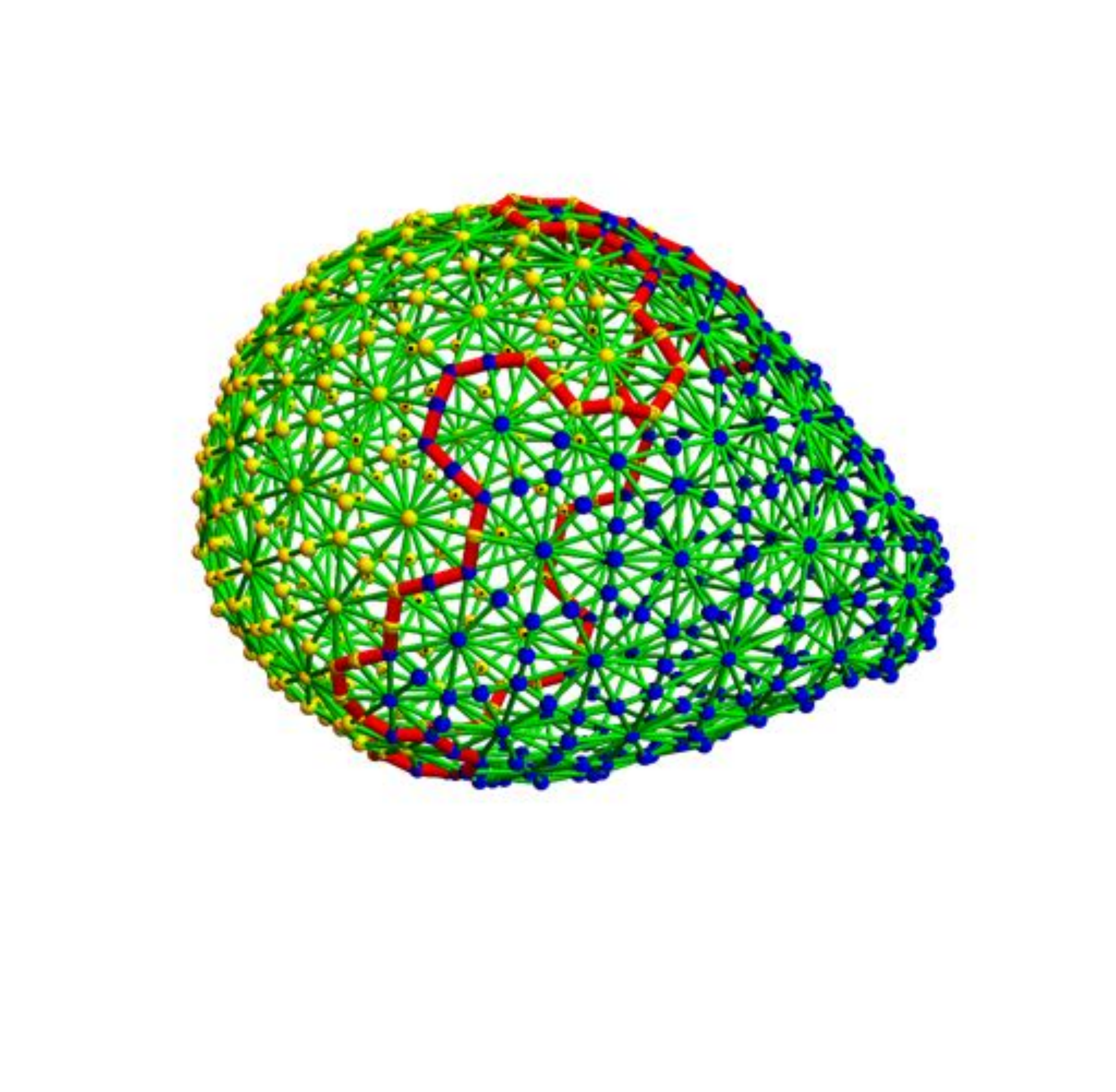}}
\scalebox{0.12}{\includegraphics{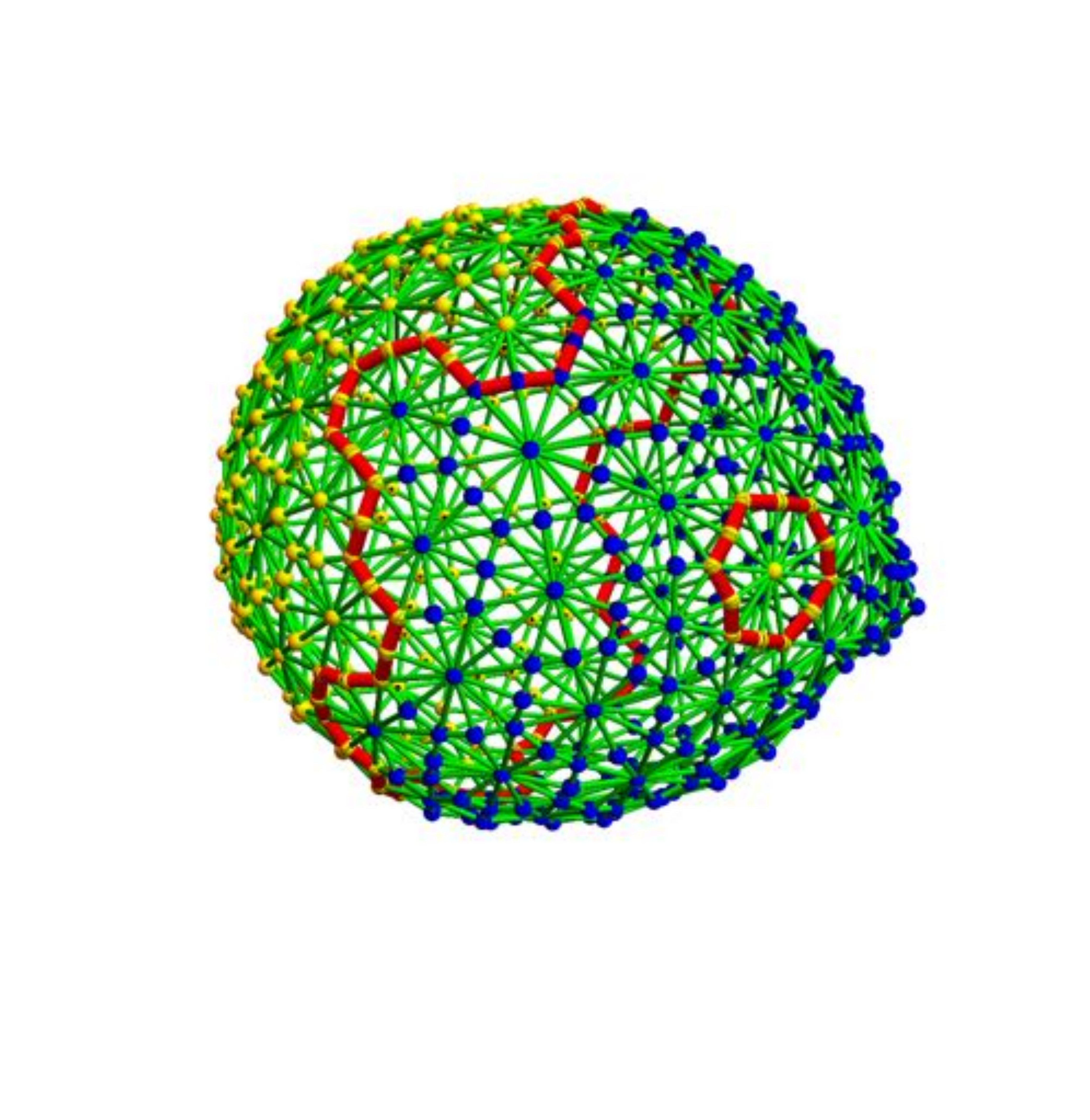}}
\caption{
\label{Stokes}
We look at the level curves $\{ f = c \}$ of some functions $f$ on 
a $2$-sphere $G$. To visualize the graph, we draw in each case
the Barycentric refinement, in which the curve is an actual curve. 
In this case there is first one single curve and then a curve with
two components. 
}
\end{figure}

\begin{figure}
\scalebox{0.24}{\includegraphics{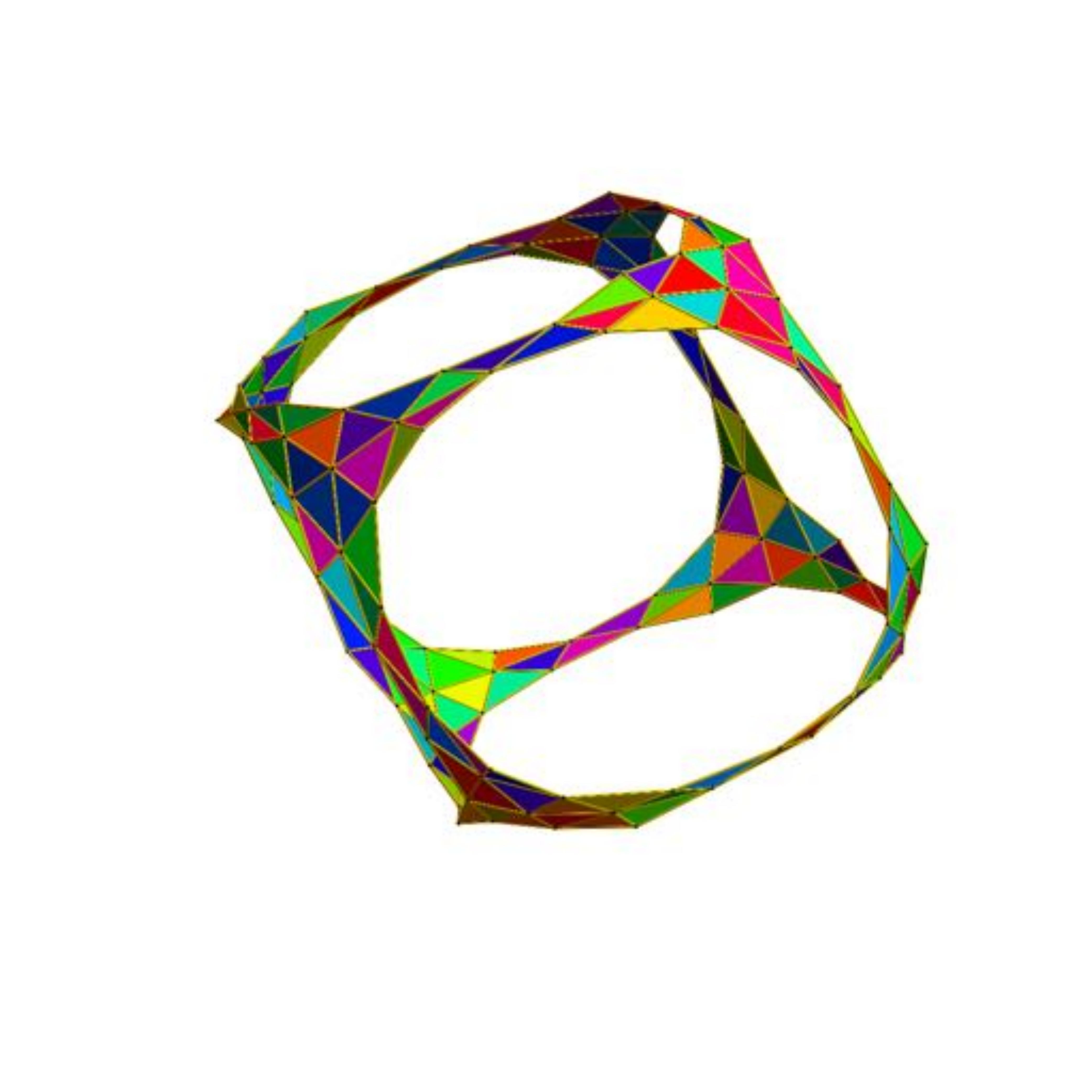}}
\caption{
\label{Stokes}
The picture shows a generalized 2-graph with boundary. In this case, an edge
refinement at every of the 8 corners will produce a
generalized 2-graph with boundary. 
}
\end{figure}

\begin{figure}
\scalebox{0.12}{\includegraphics{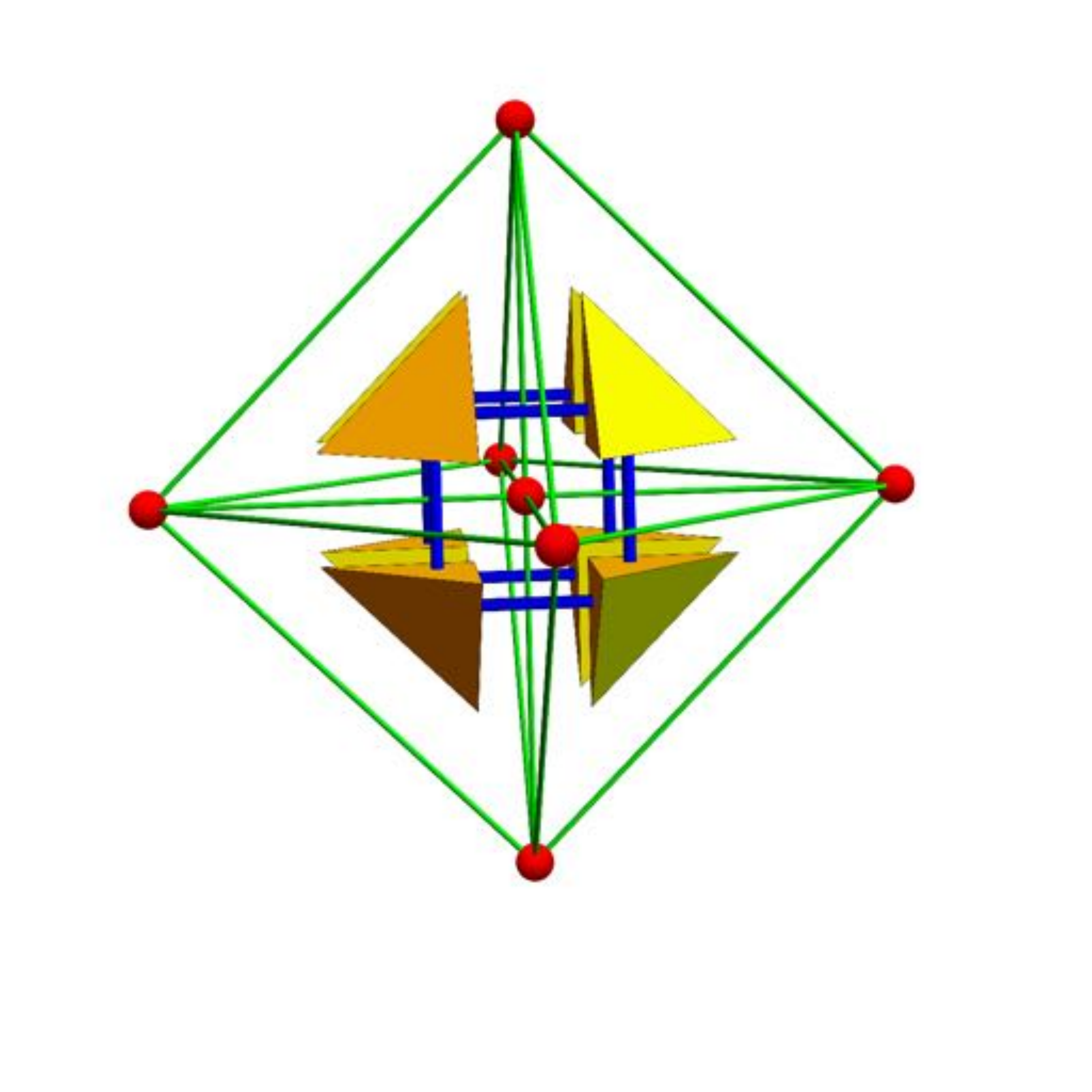}}
\scalebox{0.12}{\includegraphics{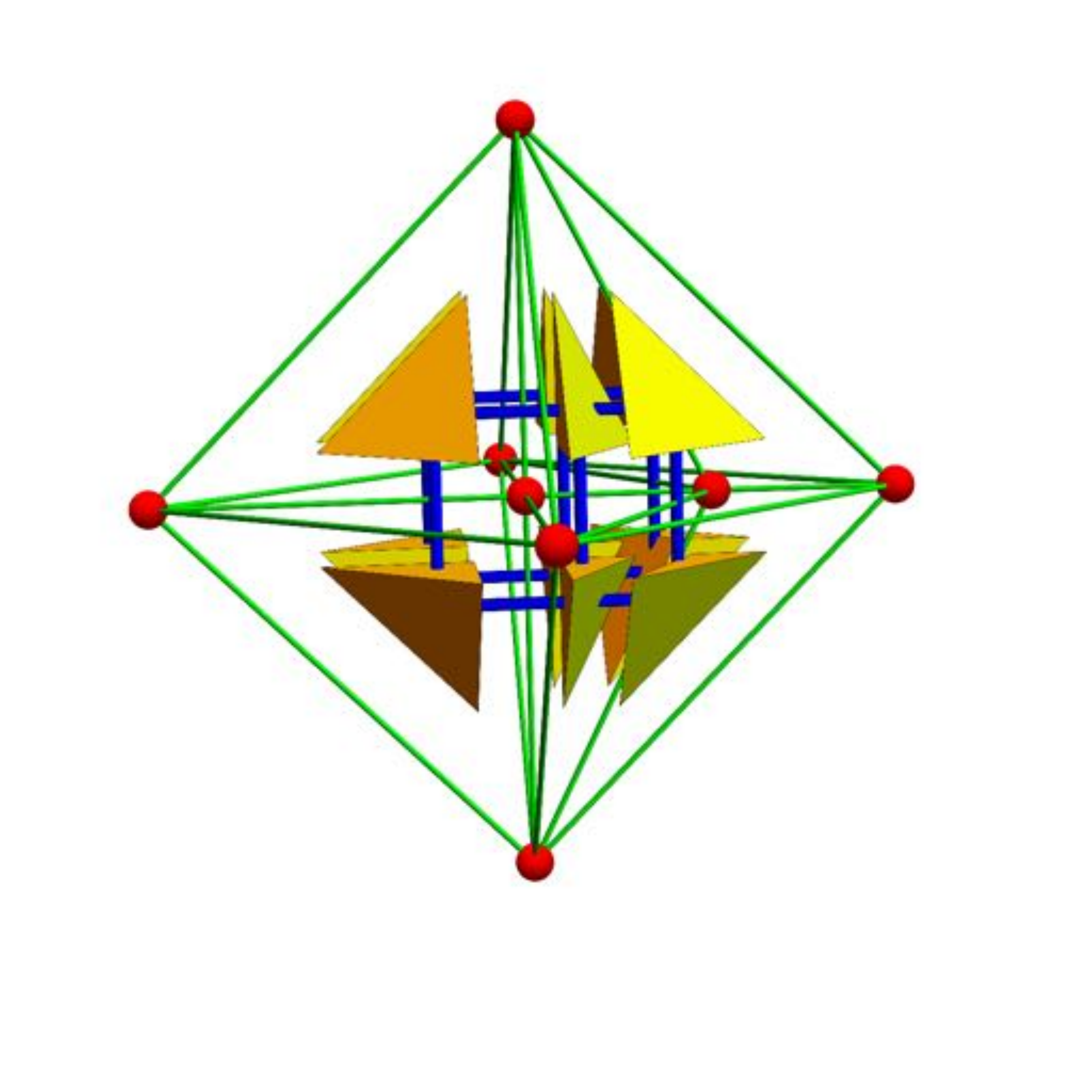}}
\caption{
\label{Edge refinement}
An edge refinement in the case of a 3-ball. It adds an other
interior vertex. Reversing the process is an edge collapse
near the boundary, which is the same than a homotopy step,
where a vertex is removed.The fact that edge refinements value
the Hamiltonian property fueled the research early on. 
}
\end{figure}

\begin{figure}
\scalebox{0.12}{\includegraphics{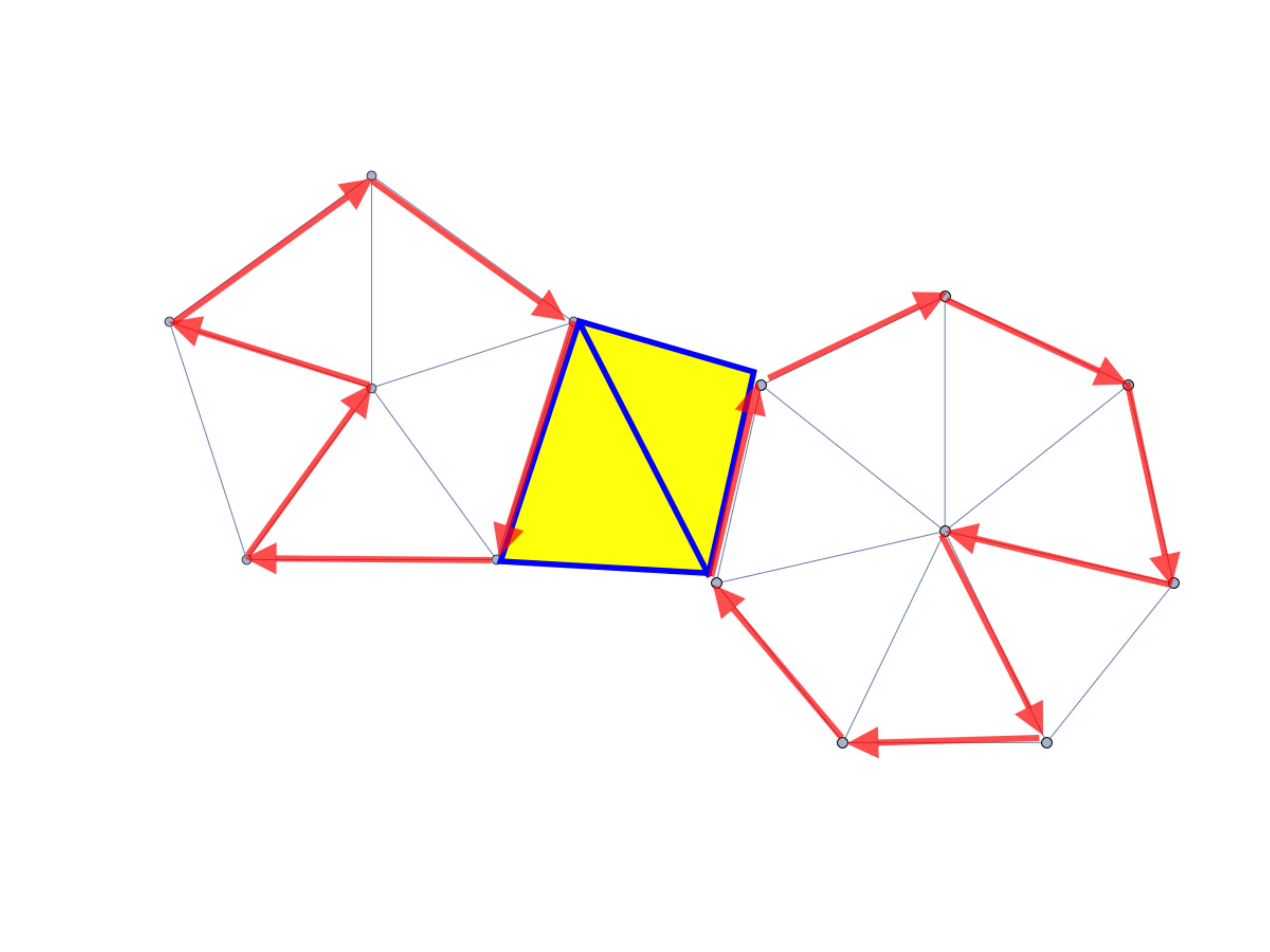}}
\scalebox{0.12}{\includegraphics{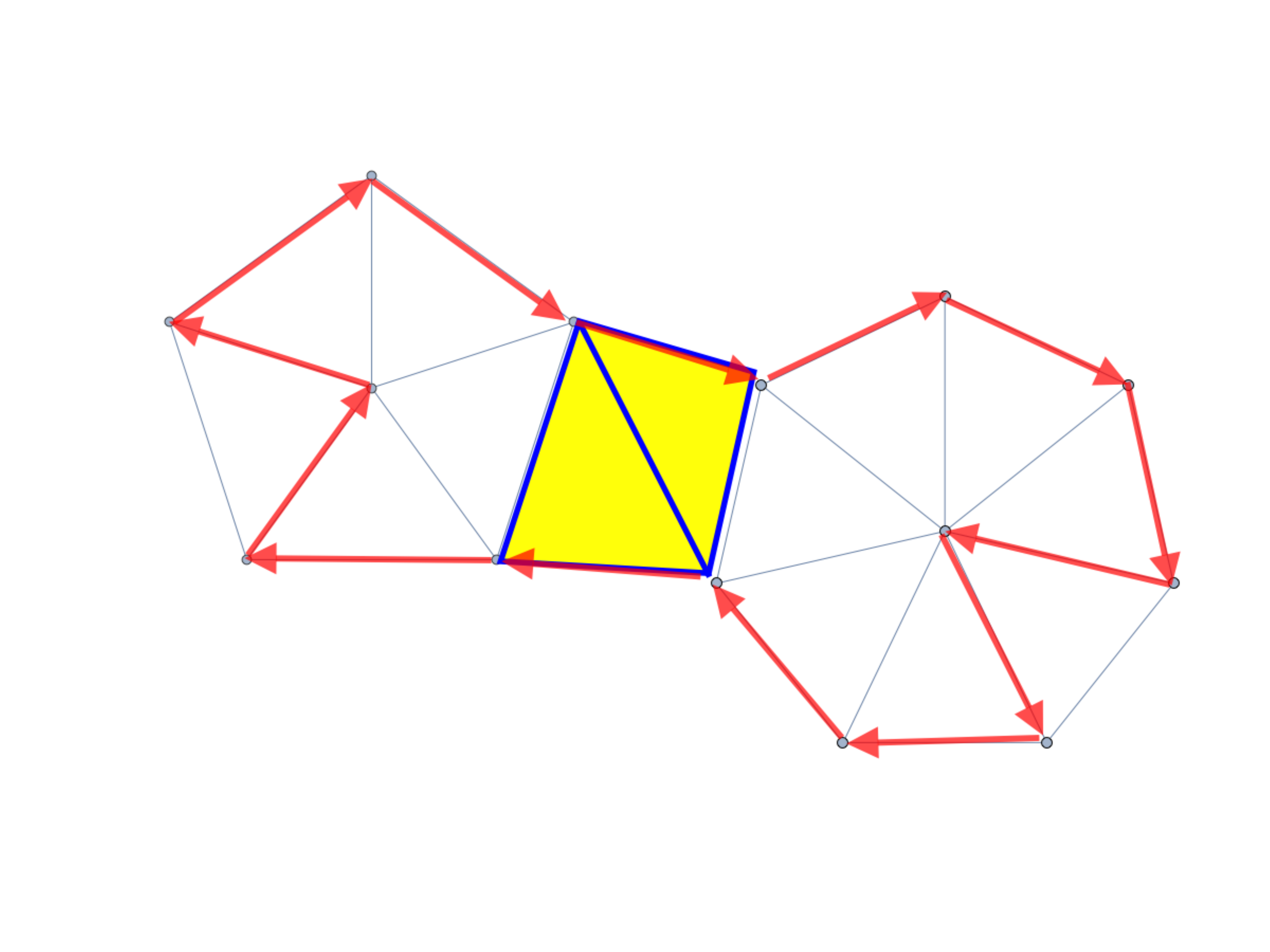}}
\caption{
\label{Prism constrution}
The prism construction can bridge two $f<c$ with $f>c$. 
The three edges in the prism represent three vertices in the
level curve $\{ f=c \}$. The two triangles are two edges in that
curve. 
}
\end{figure}

\begin{figure}
\scalebox{0.12}{\includegraphics{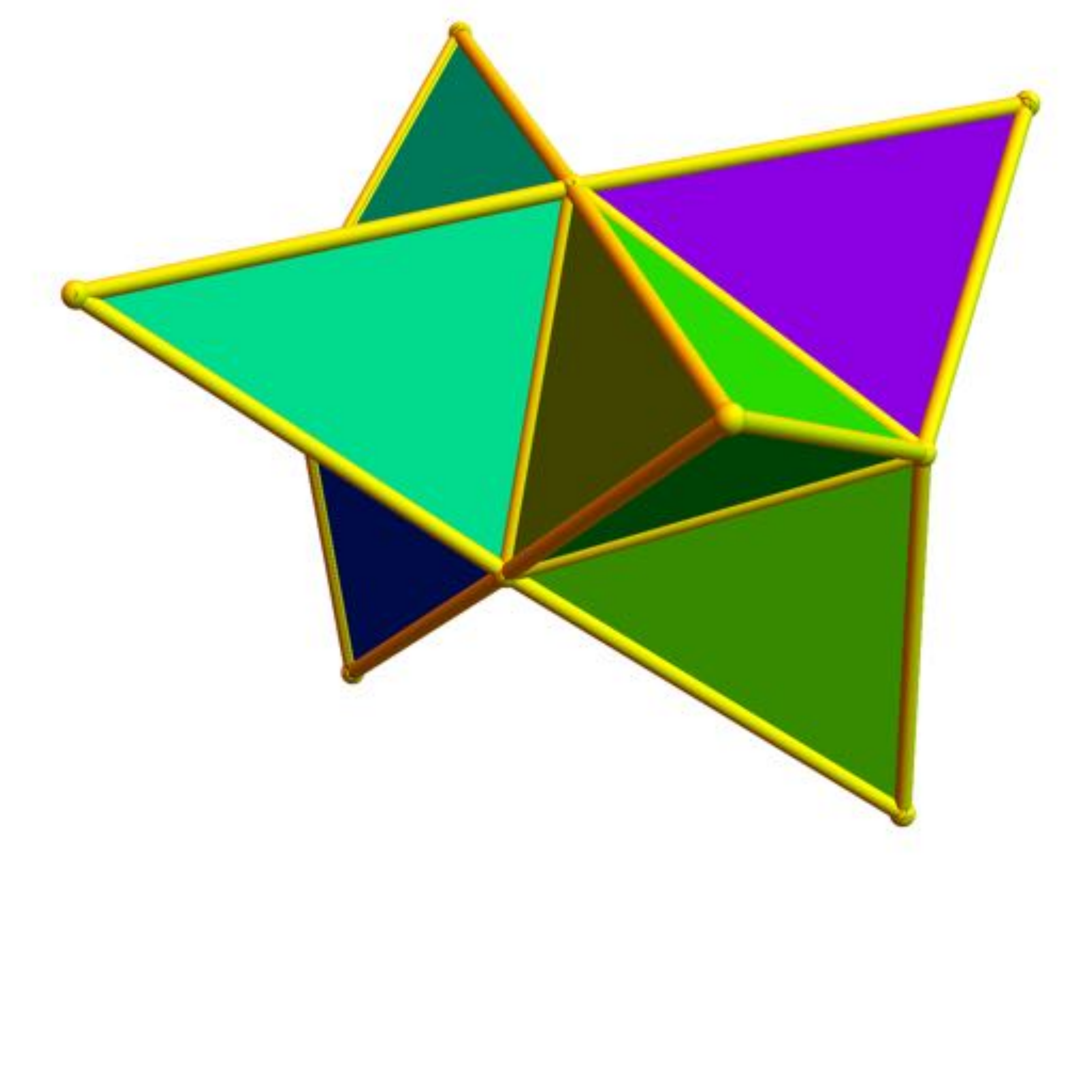}}
\scalebox{0.12}{\includegraphics{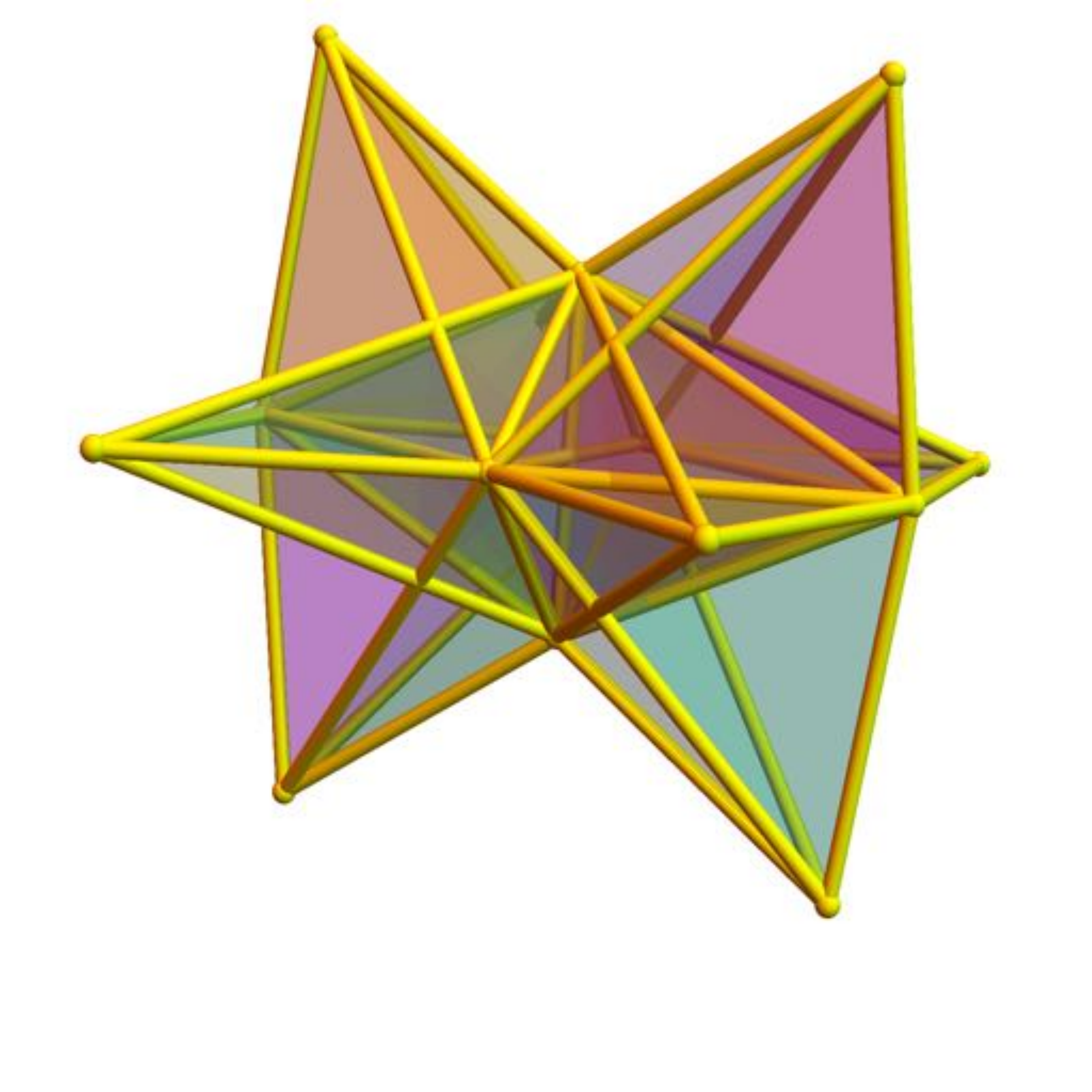}}
\caption{
\label{Goldner-Harary}
The Goldner-Harary graph seen to the left is the 1-skeleton graph of a shellable 3-dimensional 
simplicial complex whose dual graph is a tree. It is not Hamiltonian. 
It is an example of an Apollonian network. A $4$-dimensional version of the Goldner-Harary 
graph is seen to the right. We first
gluing two 4-simplices together along a $3$-
dimensional simplex, then stellate each of the eight $3$-dimensional 
boundary faces. The unit spheres in this graph are either Goldner-Harary graphs or 
tetrahedra. This graph is not Hamiltonian. It is not even a generalized $4$-graph.  
}
\end{figure}

\begin{figure}
\scalebox{0.12}{\includegraphics{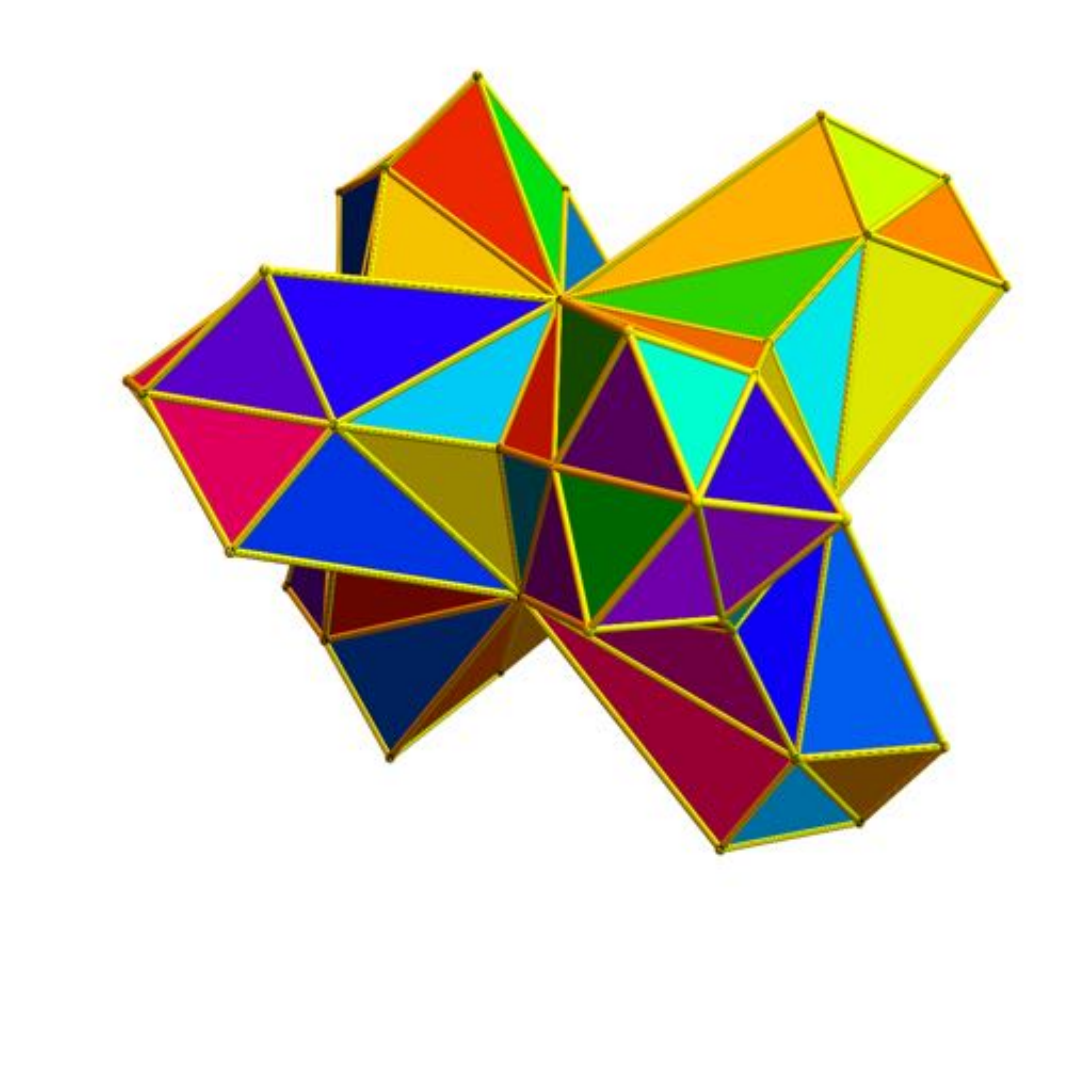}}
\scalebox{0.12}{\includegraphics{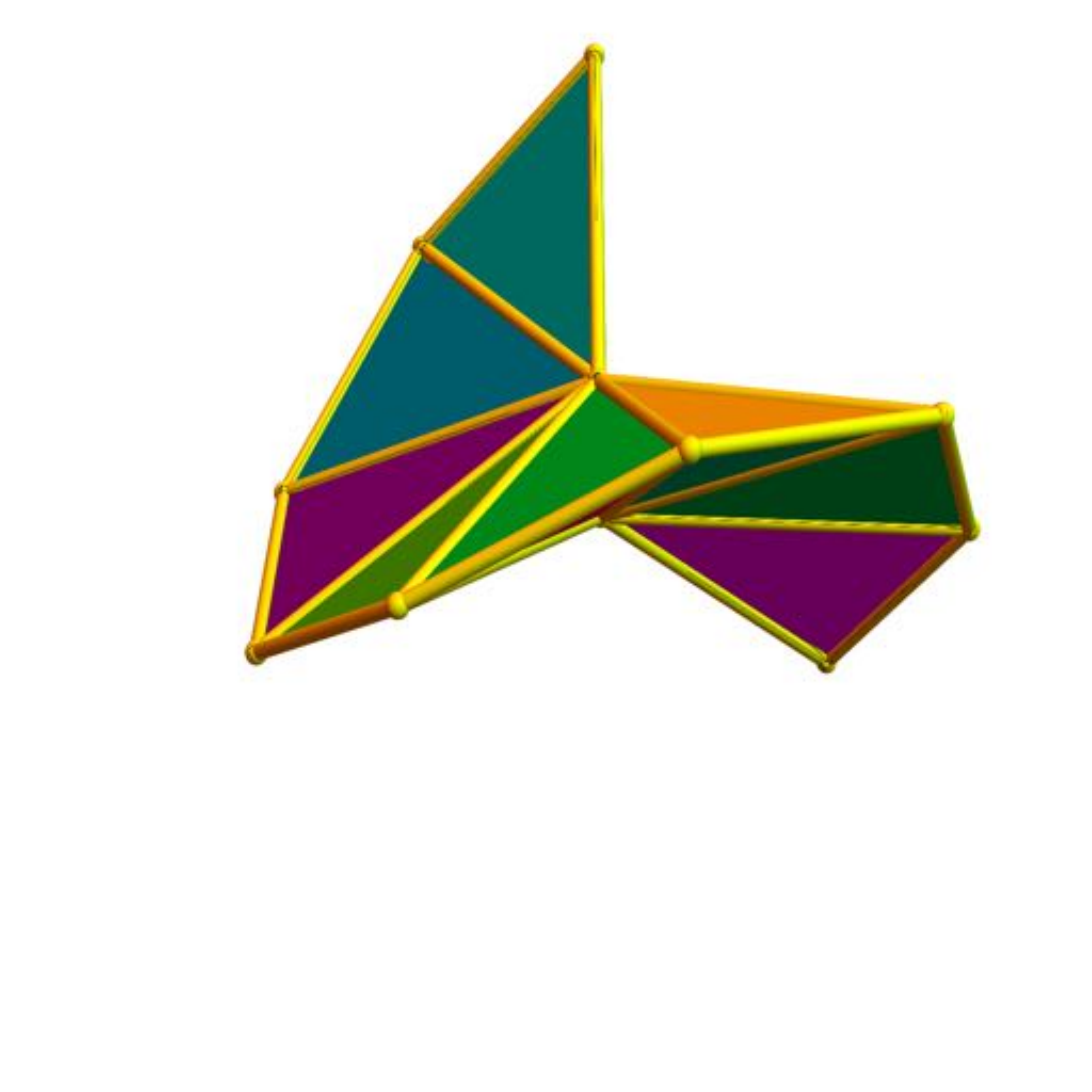}}
\caption{
\label{Goldner-Harary}
The Barycentric refinement of the Goldner-Harary graph $G$ is still not Hamiltonian. 
It looks more like a 3-ball, but there are still unit spheres $G$ which are not
balls, nor spheres, nor simplices. To the right, we see one of the
unit spheres. It  is the cyclic polytop $C_3(8)$ with the central hinge removed. 
}
\end{figure}

\begin{figure}
\scalebox{0.12}{\includegraphics{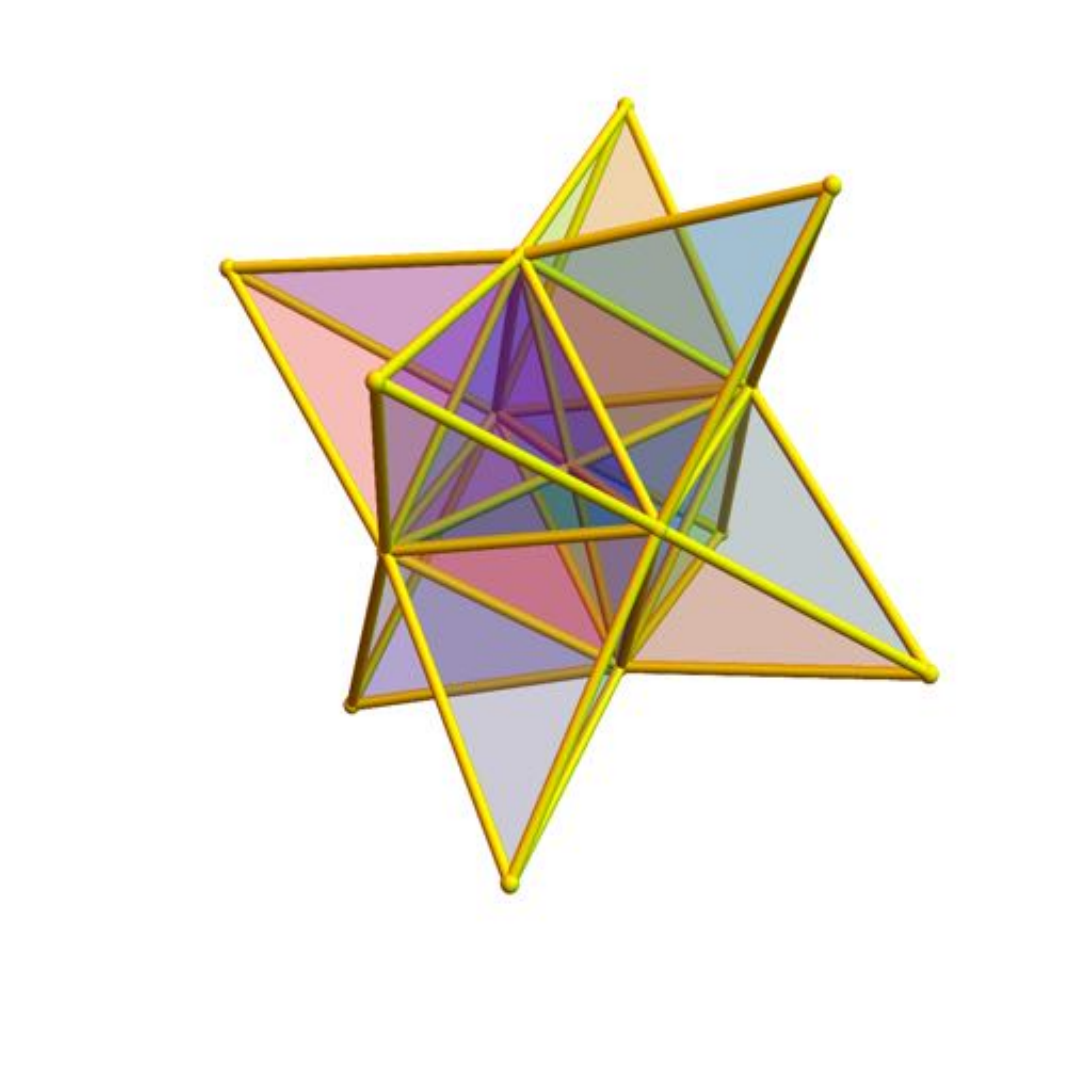}}
\scalebox{0.12}{\includegraphics{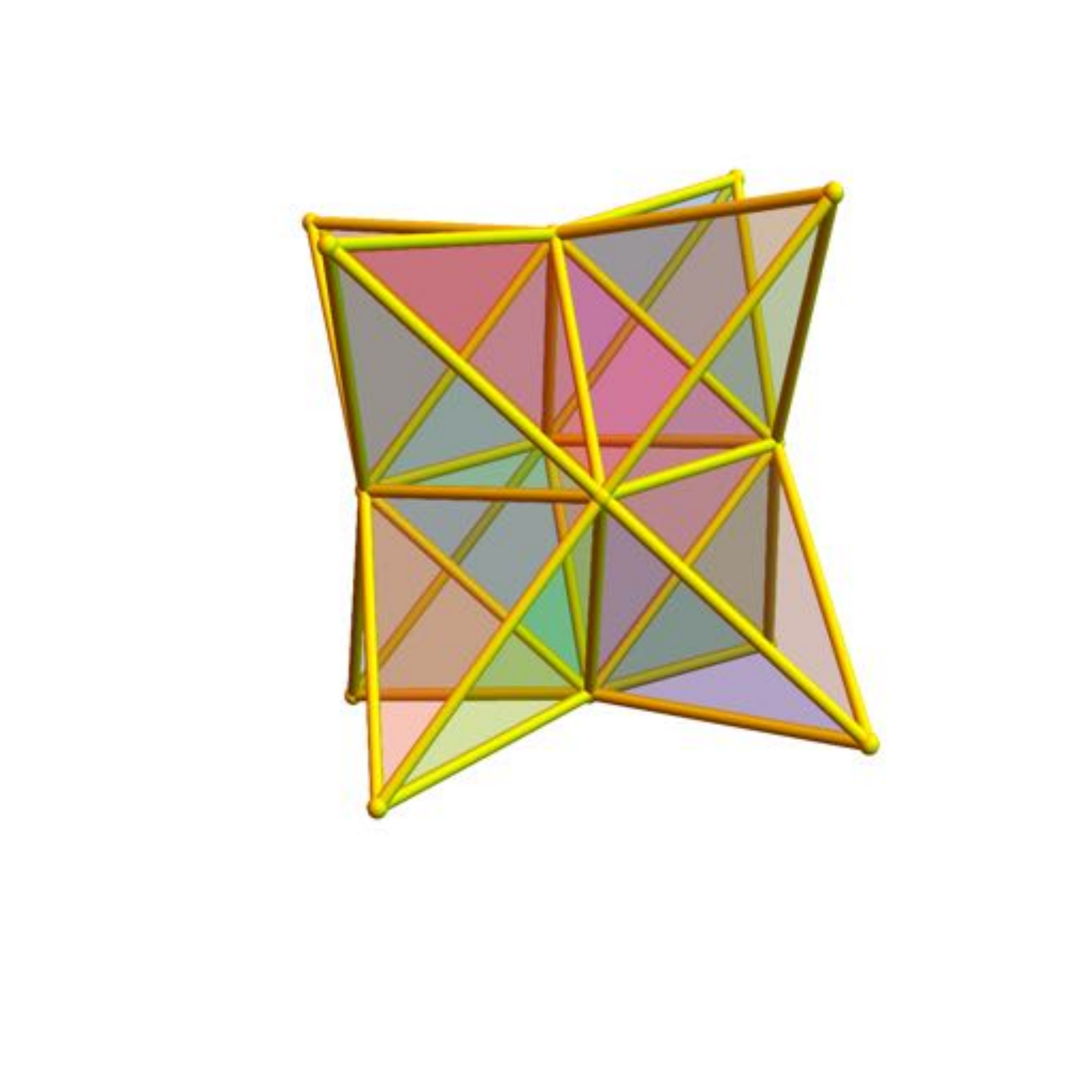}}
\caption{
\label{Goldner-Harary}
The stellated, filled octahedron, 
which is only 3-connected. It is not Hamiltonian and the Hamiltonian 
property fails for similar reasons than for the Goldner-Harary graph. It is not
a generalized $3$-graph because the inner unit sphere (the octahedron) has no edge
in common with the boundary. Also the local Euclidean property fails both in the
filled an non-filled case. This is the deciding factor against the Hamiltonian property
here as removing only one spike still does not make it Hamiltonian 
even so the inner vertex is now accessible from the boundary. 
}
\end{figure}

\begin{figure}
\scalebox{0.12}{\includegraphics{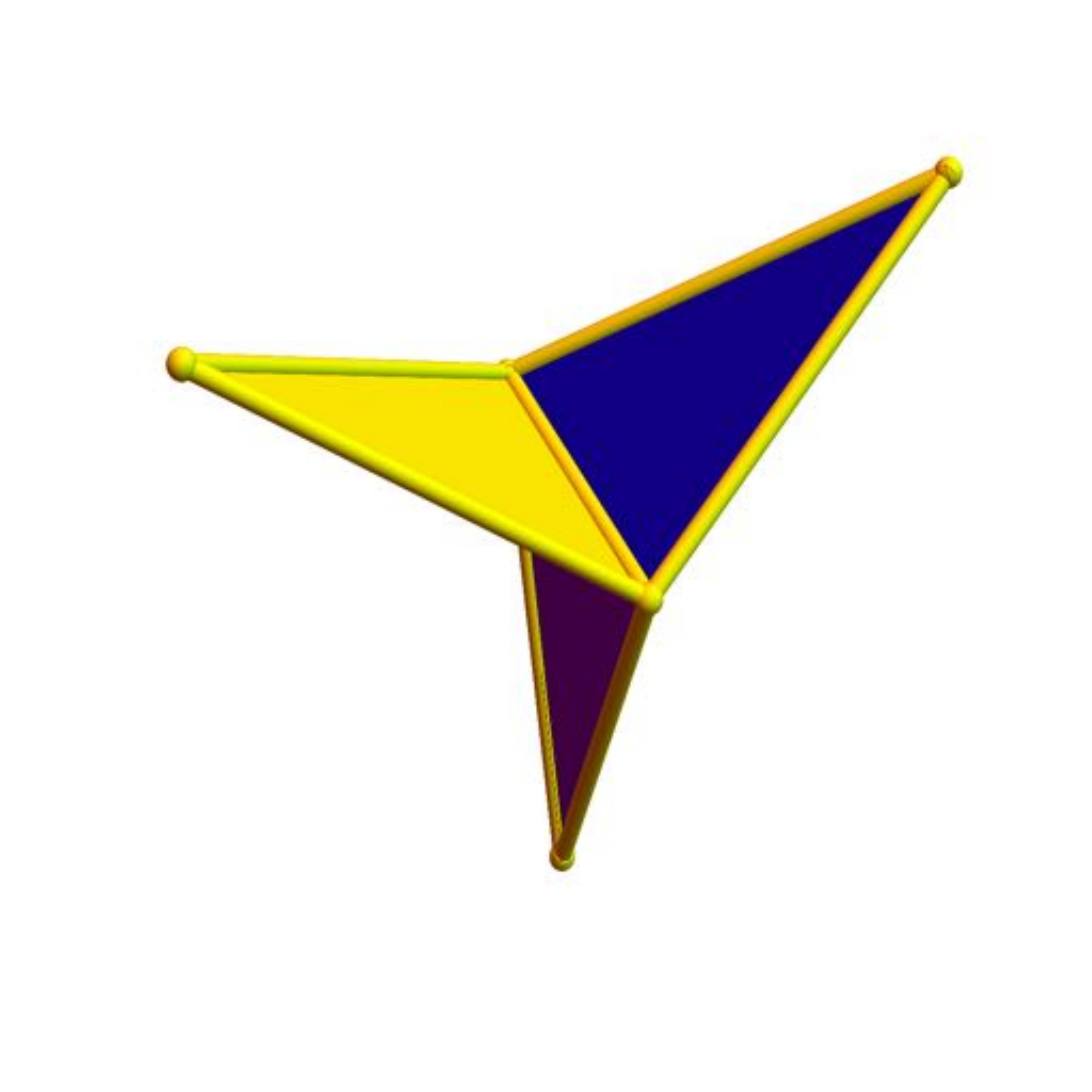}}
\scalebox{0.12}{\includegraphics{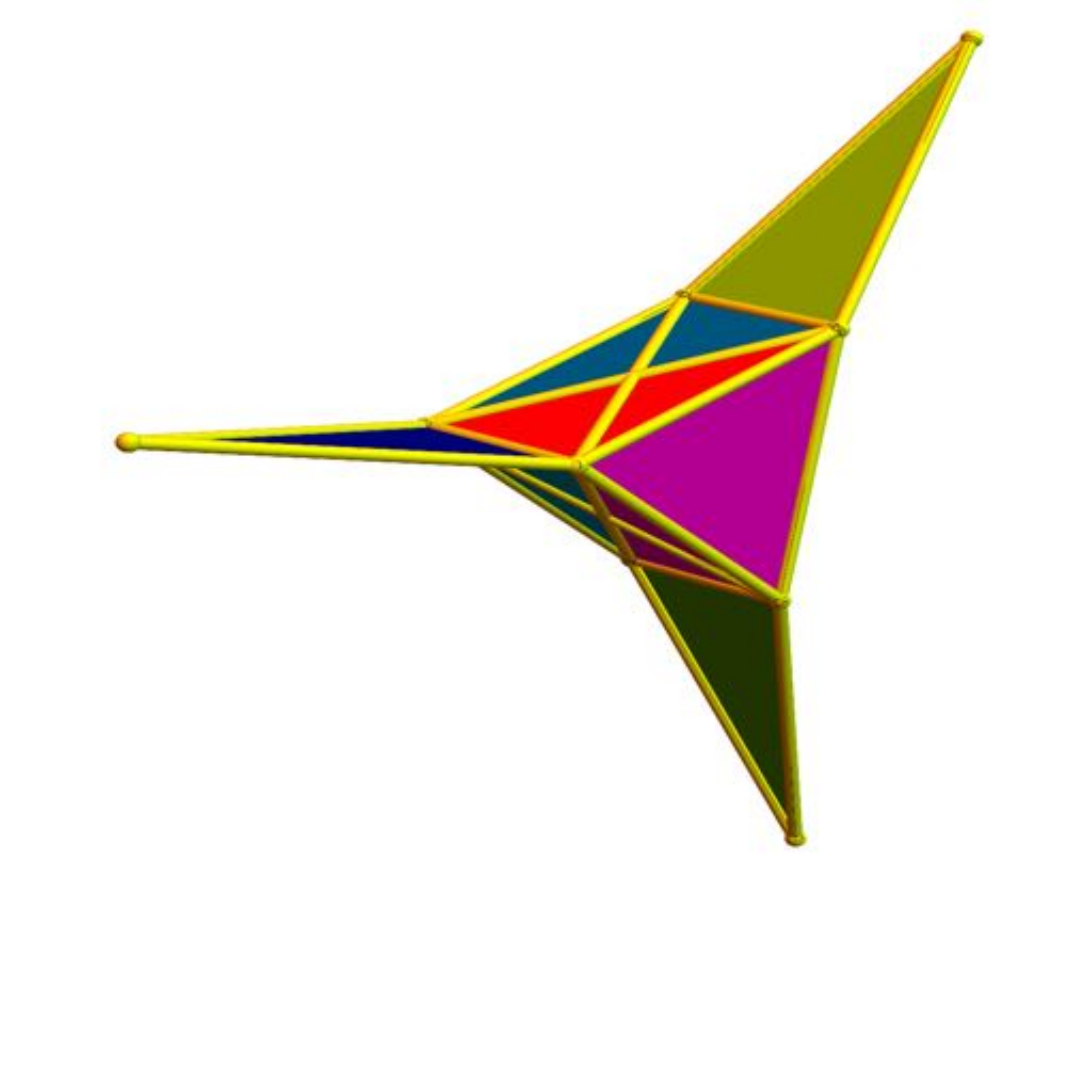}}
\caption{
\label{Blockisland}
The ``windmill graph" is a 2-dimensional shellable complex which 
is planar but not Hamiltonian. Algebraically, it is the join of $K_2$ with the 
$3$-point graph $3=1+1+1$. It consists of 3 triangles (blades) connected to
a common edge. Having non-spherical unit spheres is often a decisive factor 
against the Hamiltonian property. To the right we see an other non-Hamiltonian
graph, which is a fat windmill. 
}
\end{figure}

\begin{figure}
\scalebox{0.12}{\includegraphics{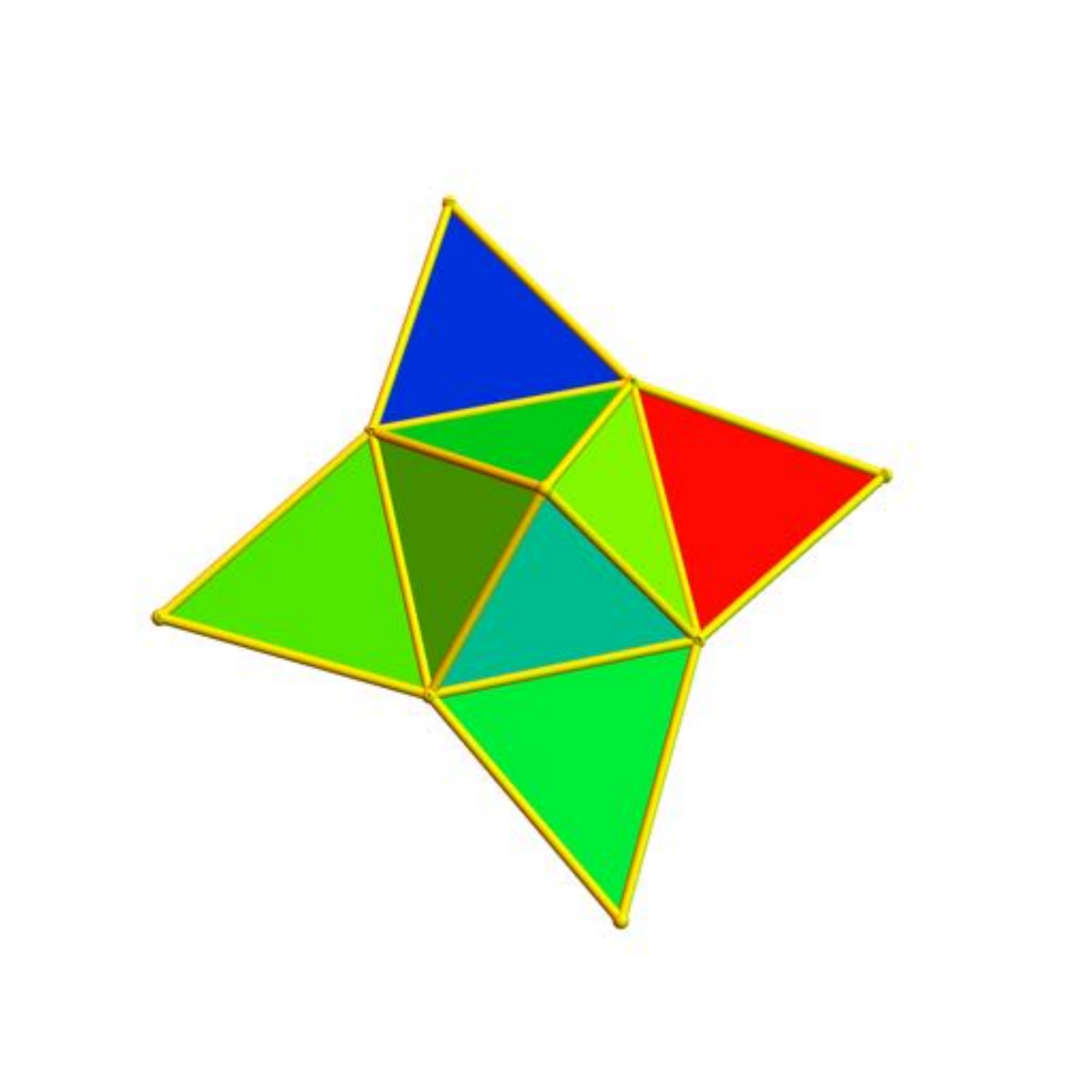}}
\scalebox{0.12}{\includegraphics{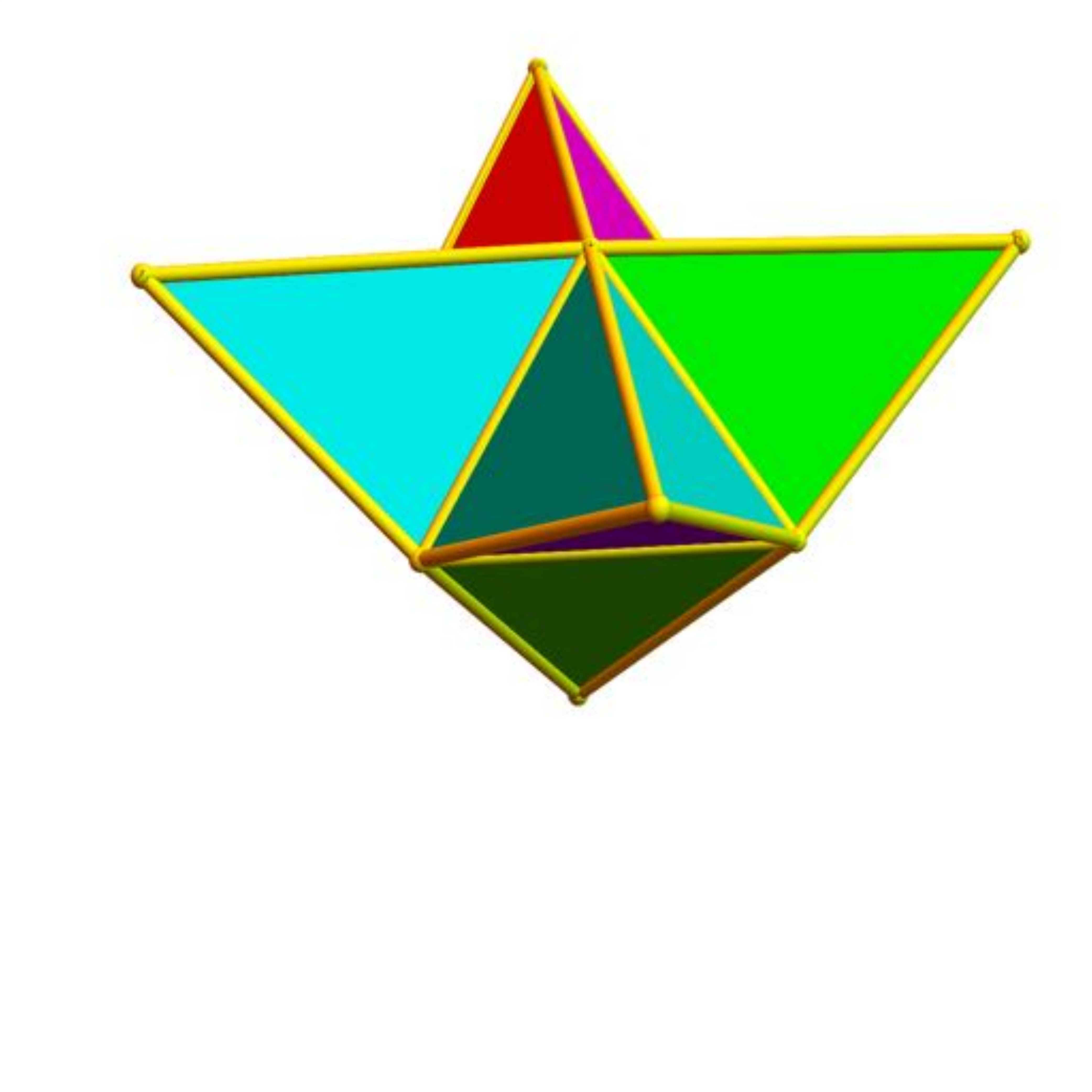}}
\caption{
\label{Stellated Square}
The stellated square to the left is a generalized 2-ball which
is not Hamiltonian. It is an example with an isolated 2-ball inside
which only touches the boundary in zero dimensional places.
When completing it to a sphere, it becomes Hamiltonian even so it is not 
a 2-sphere (there are unit spheres which are not circular graphs). 
}
\end{figure}

\begin{figure}
\scalebox{0.12}{\includegraphics{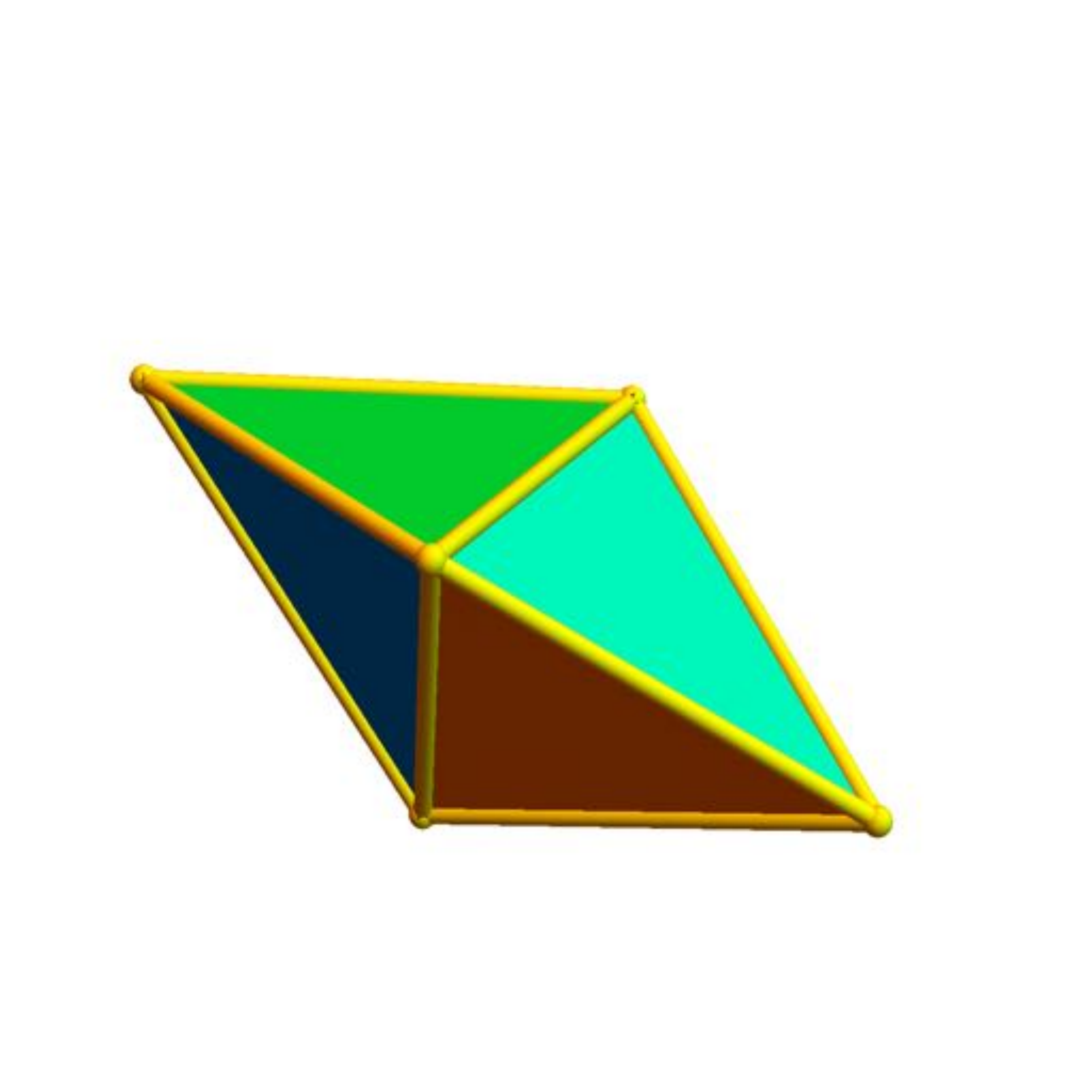}}
\scalebox{0.12}{\includegraphics{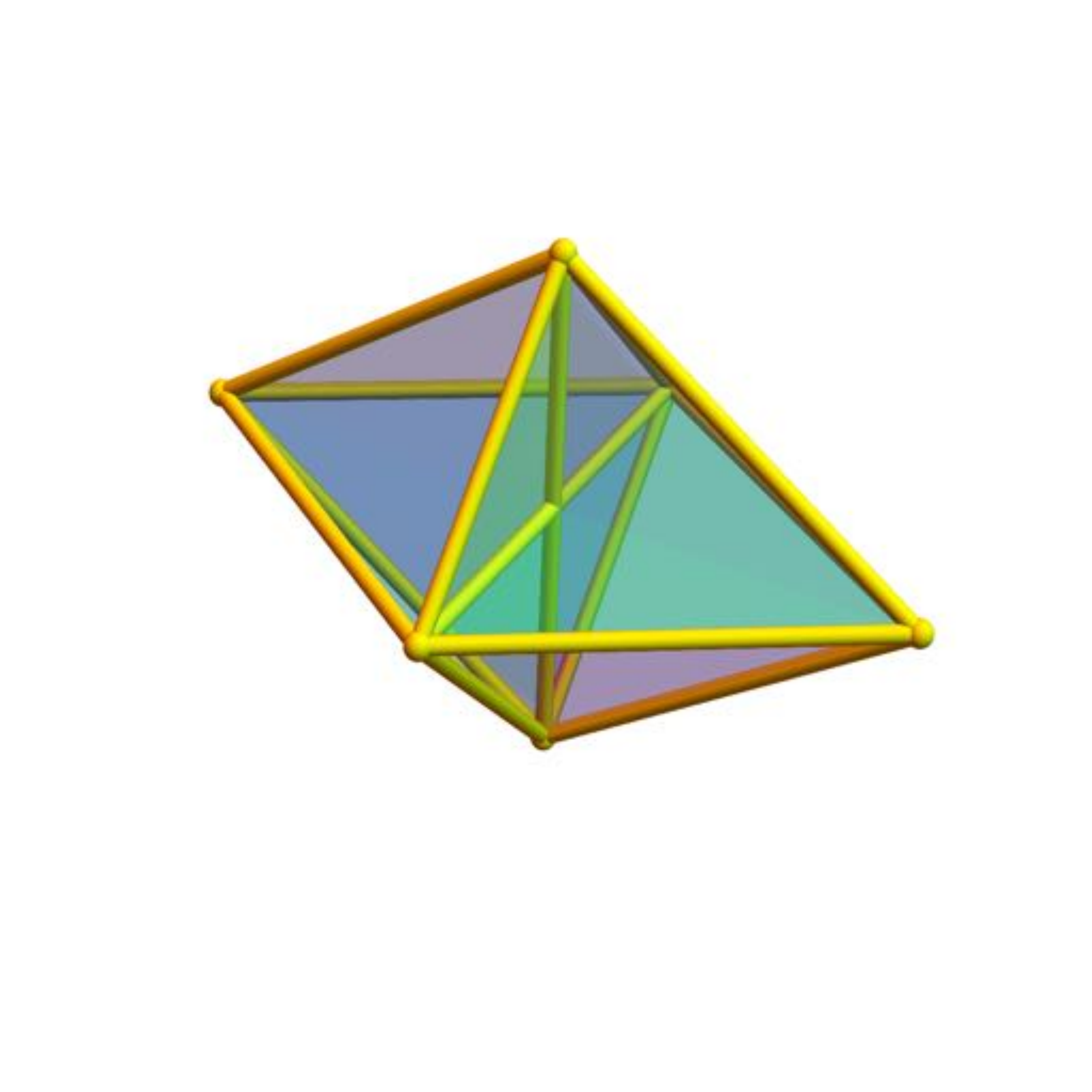}}
\caption{
\label{Avici}
The ``Avici graph" is a shellable 3-dimensional simplicial 
complex which is Hamiltonian. It is not a 3-ball. It is a generalized
$3$-graph. To the right we see the 4 dimensional Avici graph, which is a union of
two 4-simplices glued along a 3-simplex. It is a generalized $4$-graph
and also Hamiltonian. 
}
\end{figure}

\begin{figure}
\scalebox{0.15}{\includegraphics{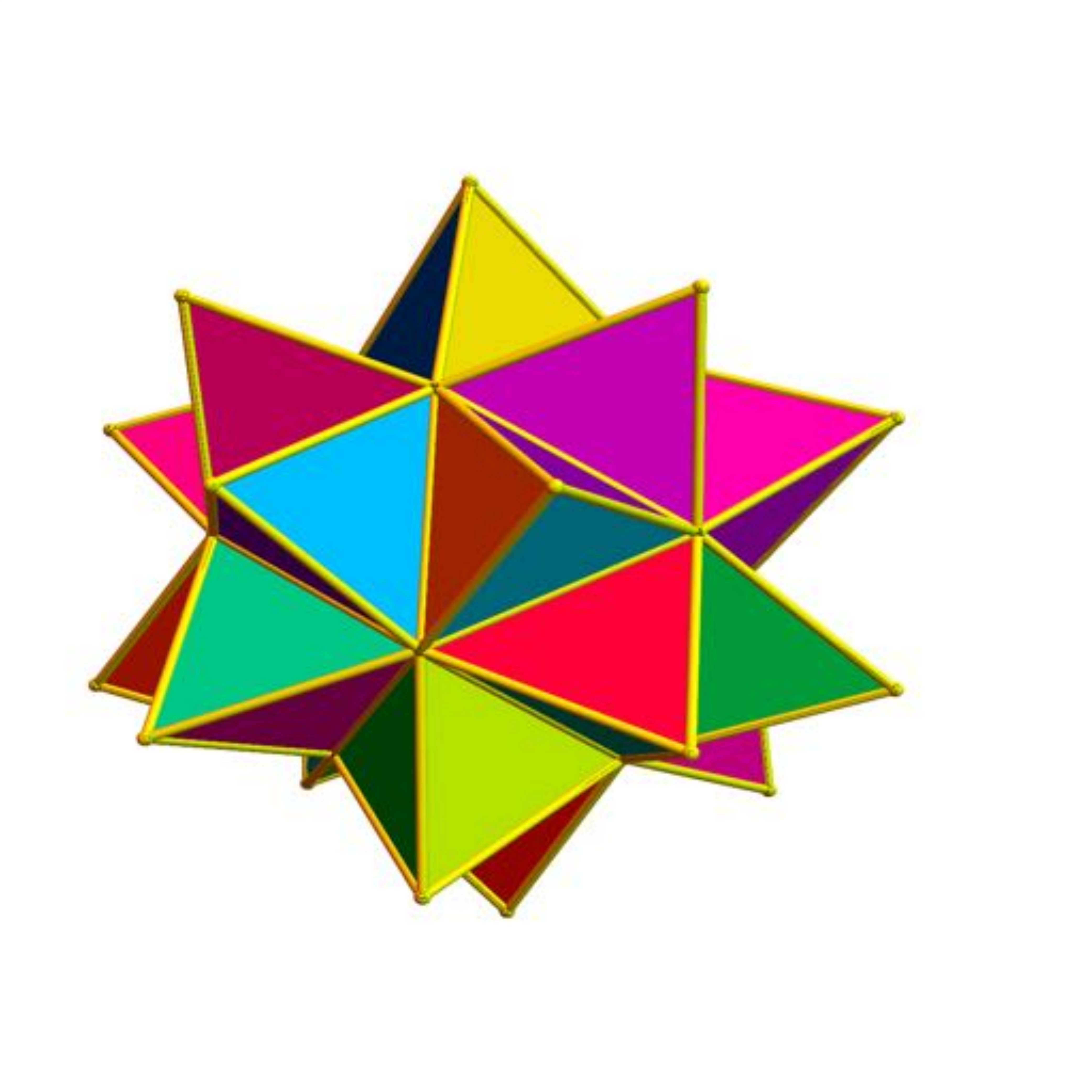}}
\caption{
\label{Keplerpoinsot2}
The stellated isosahedron is a shellable $3$-dimensional complex
which is not Hamiltonian. It is not a generalized 3-ball. 
}
\end{figure}

\begin{figure}
\scalebox{0.14}{\includegraphics{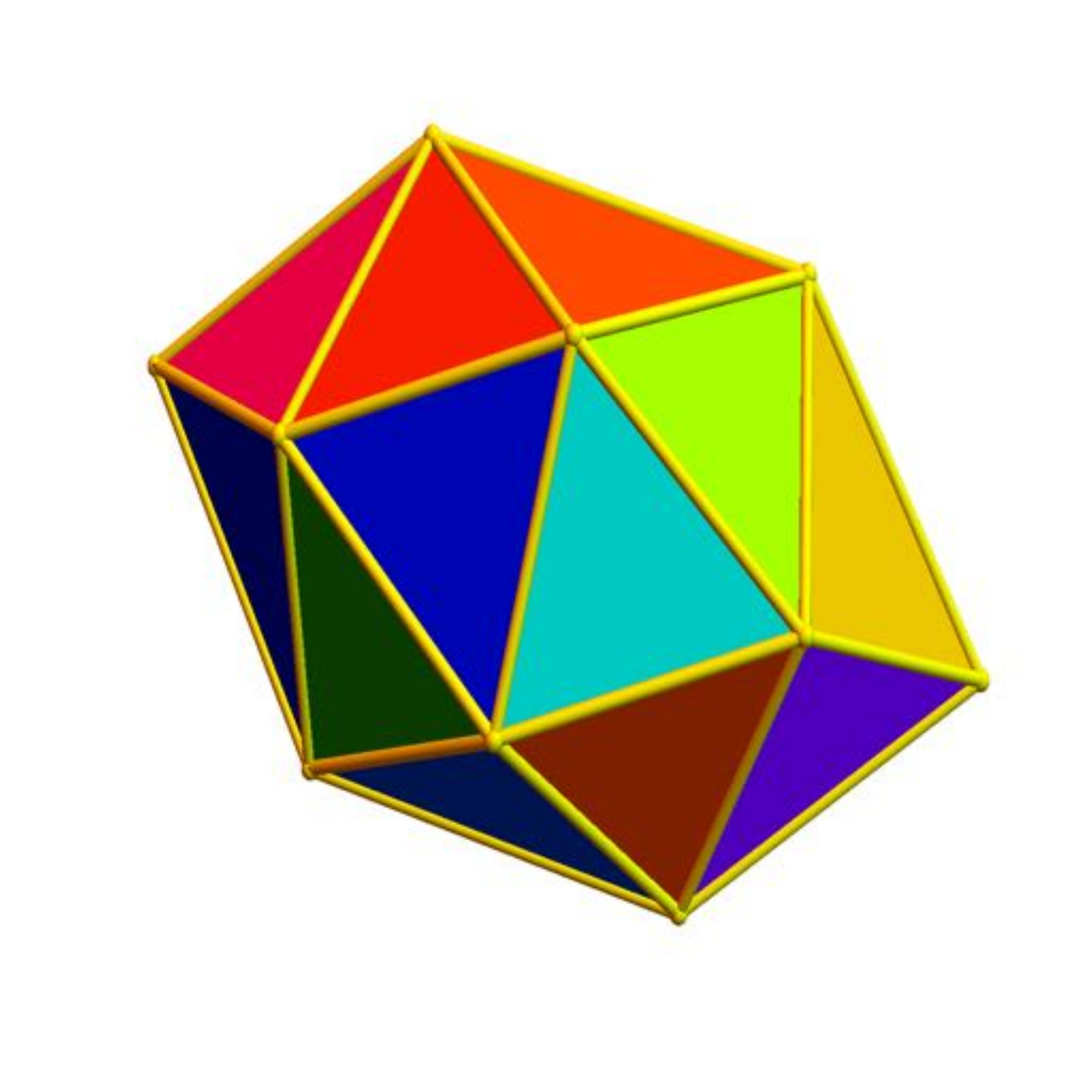}}
\scalebox{0.14}{\includegraphics{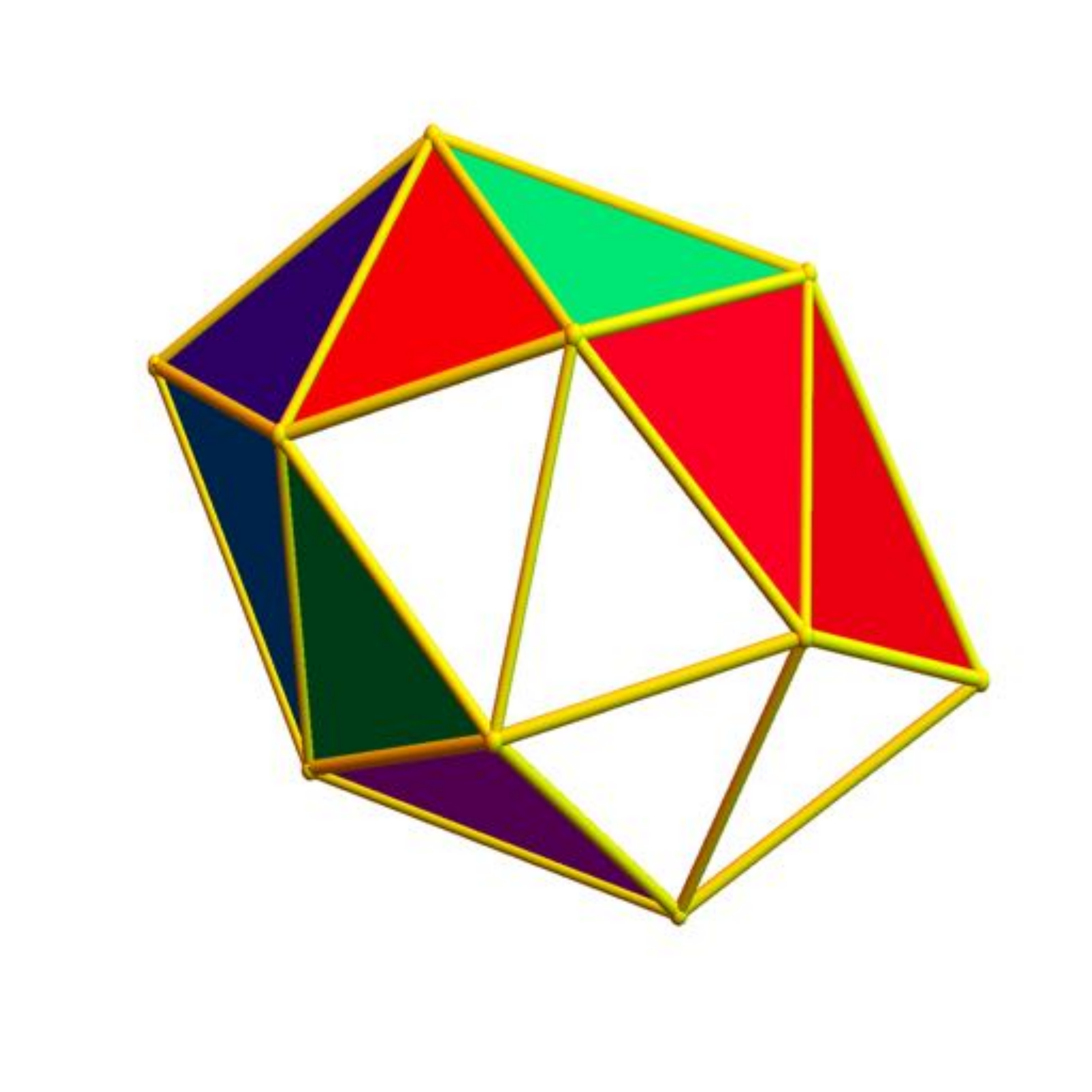}}
\caption{
\label{Birkhoffdiamond}
The Birkhoff diamond is a generalized 2-ball. 
We can remove 4 edges on the boundary to expose the 4 interior points to the boundary
producing a stellated diamond. We could also have used a function $f$ and 
join the Hamiltonian cycles for $f<0$ (which is a wheel graph) and $f>0$ which is a
diamond graph. Alternatively, we can remove four other edges (seen to the right).
A Hamiltonian cycle having edges in all triangles is obtained
by going around the boundary of the remaining triangles. 
}
\end{figure}

\begin{figure}
\scalebox{0.12}{\includegraphics{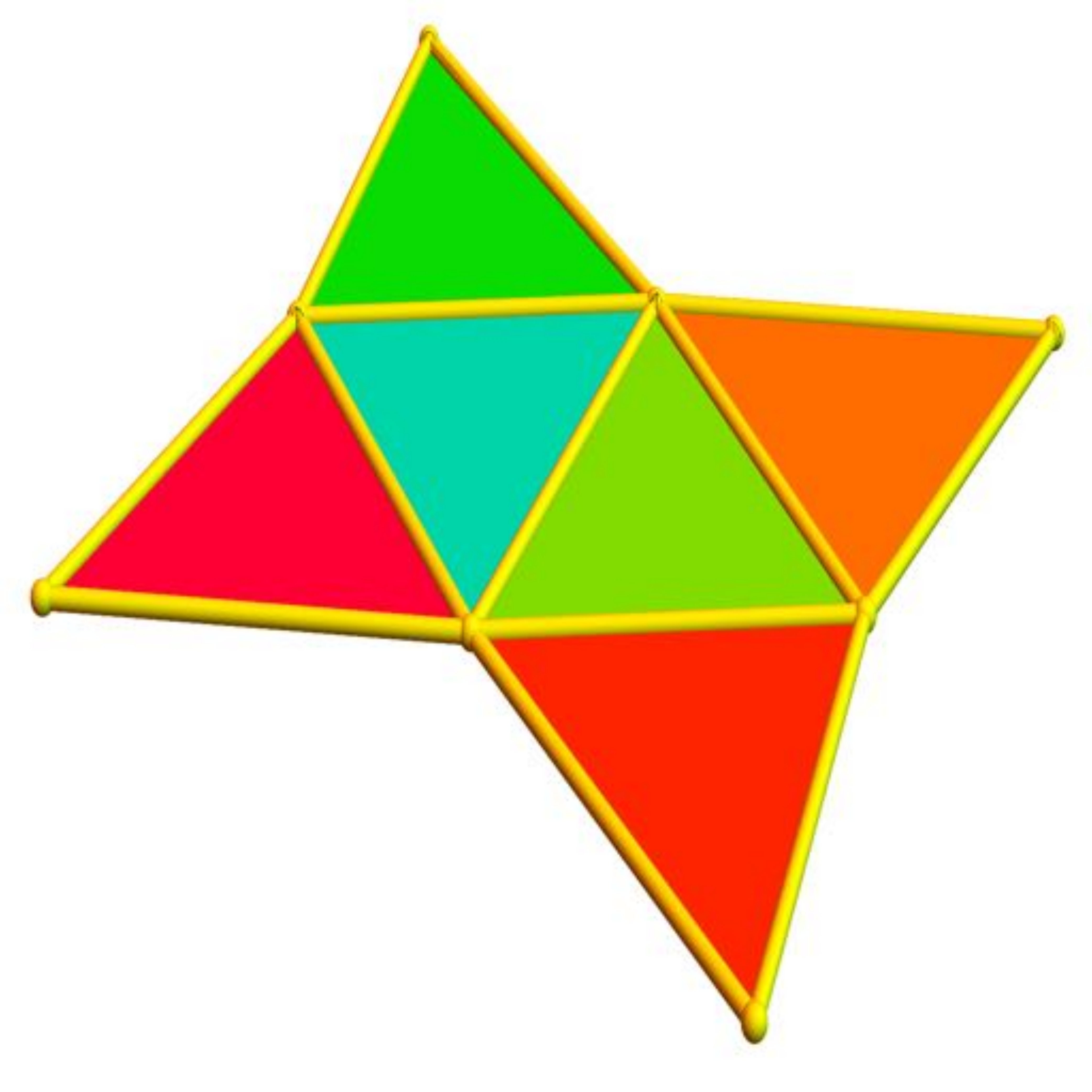}}
\scalebox{0.12}{\includegraphics{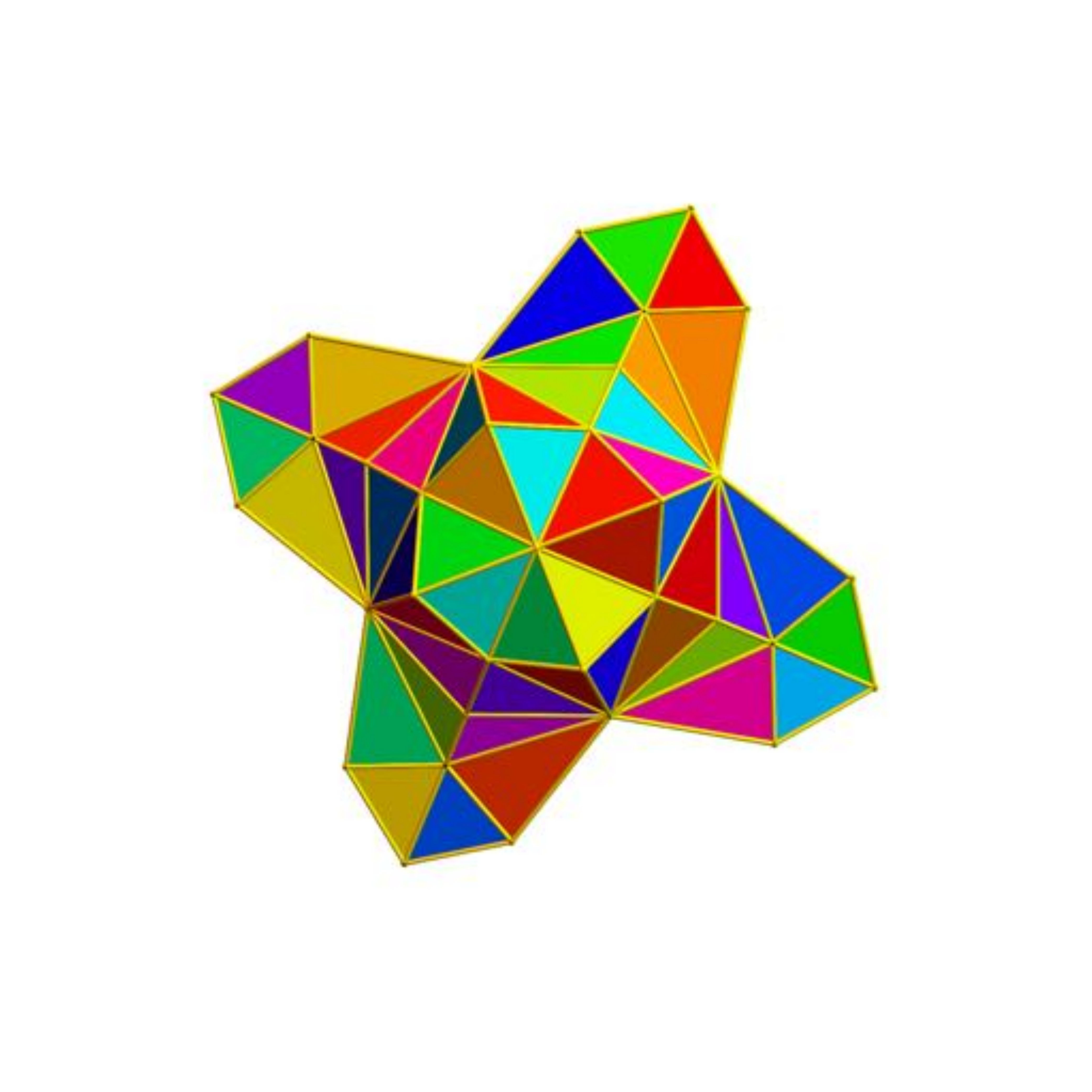}}
\caption{
\label{GoldnerHarary2}
The $2$-dimensional version of the Goldner-Harary graph is a
stellated diamond. It is Hamiltonian. 
To the right we see the Barycentric refinement $G_1$ of the stellated square $G$. 
While $G$ was not Hamiltonian, the refined $G_1$ is. The Barycentric refinement
of Goldner Harary was not Hamiltonian.
}
\end{figure}

\begin{figure}
\scalebox{0.12}{\includegraphics{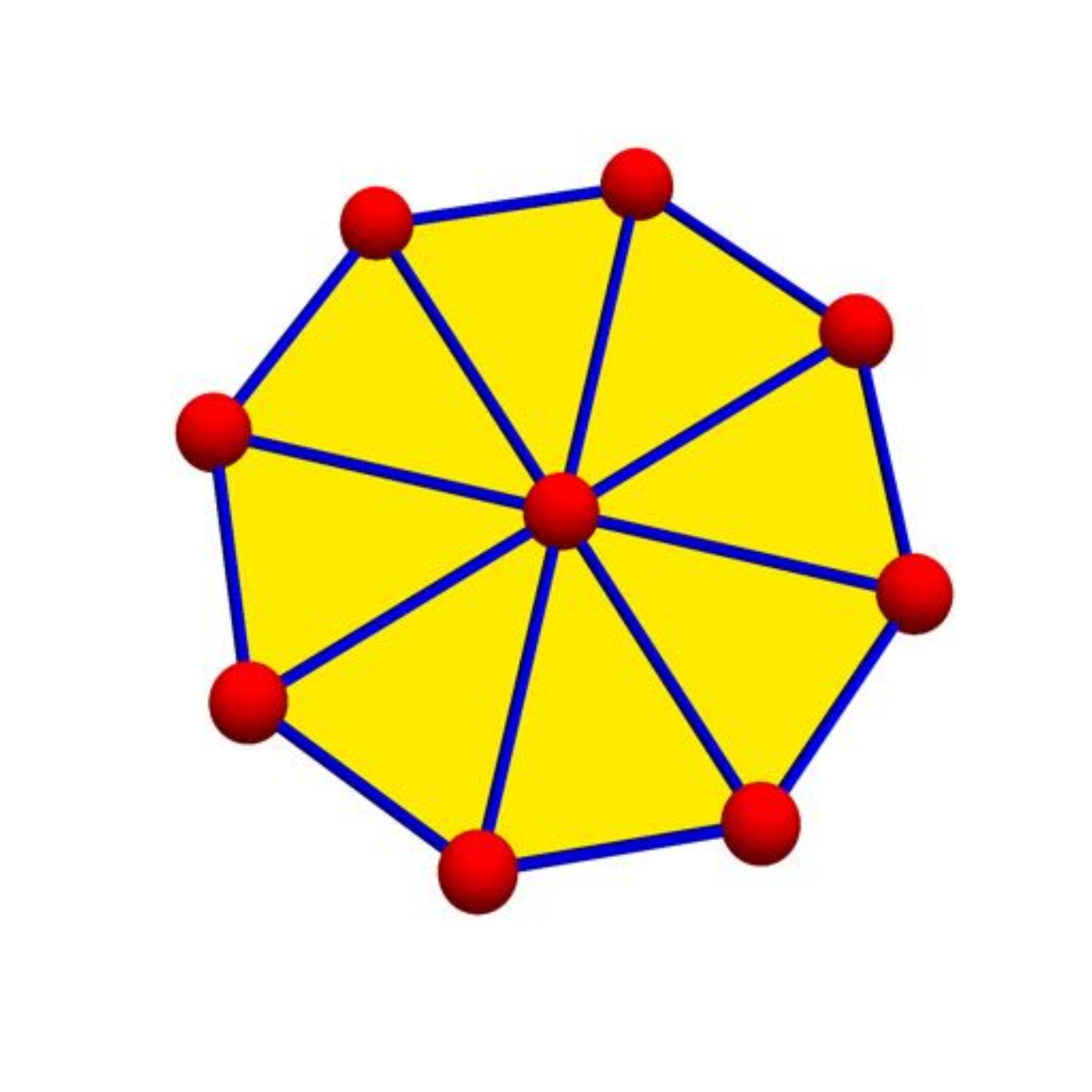}}
\scalebox{0.12}{\includegraphics{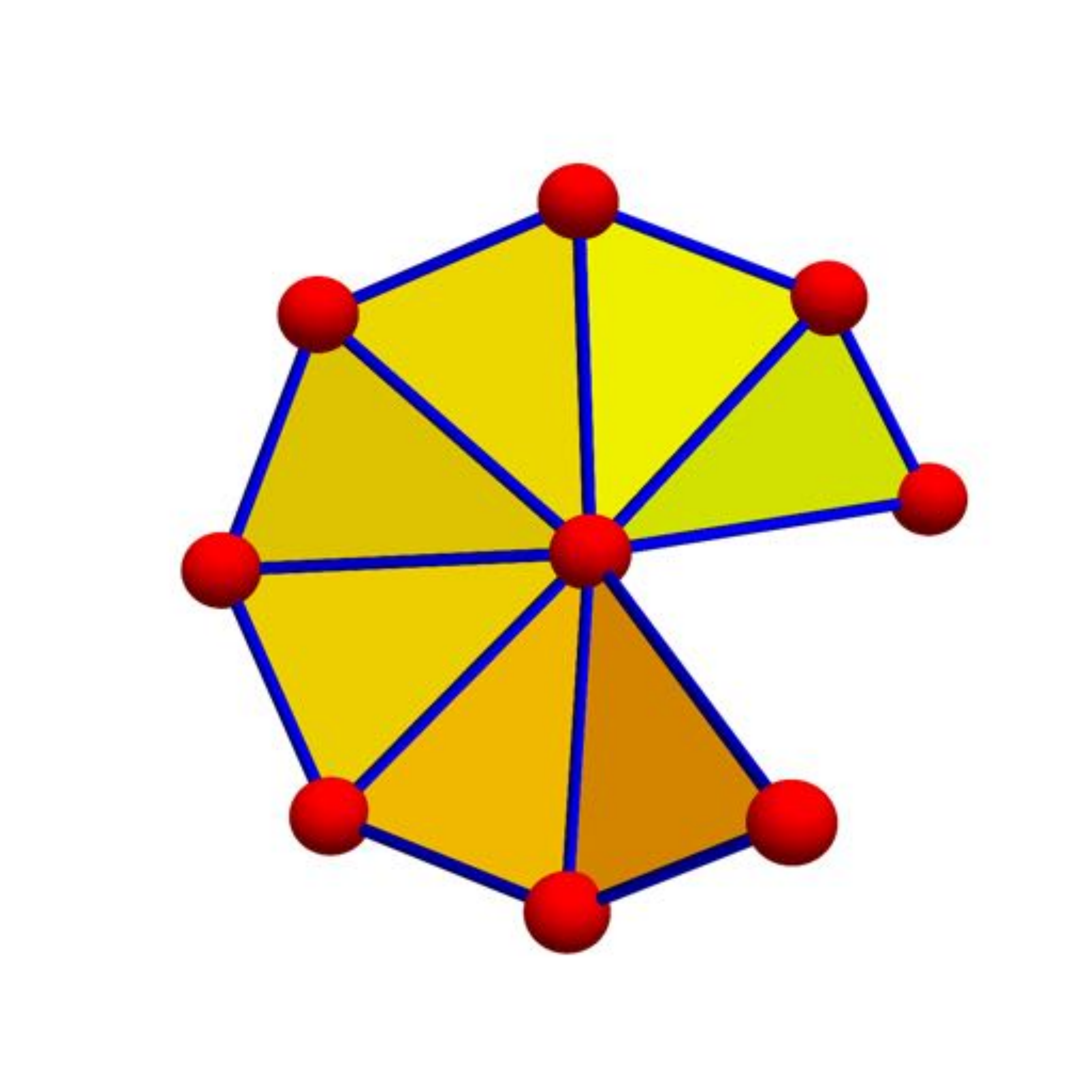}}
\caption{
\label{Packman}
A wheel graph is a 2-ball $G$. It has one interior point. Cutting away
one edge at the boundary produces a generalized 2-ball $H$. It is technically
no more a $2$-ball because two vertices have unit spheres $K_2$. After
removing the interior edges, it is a 1-sphere which is Hamiltonian. As this
boundary has the same vertices than $H$ and $G$, also $G$ is Hamiltonian. 
It is only dimension $d=2$ that we can not insist on the 
strong Hamiltonian property on the boundary. 
}
\end{figure}

\begin{figure}
\scalebox{0.12}{\includegraphics{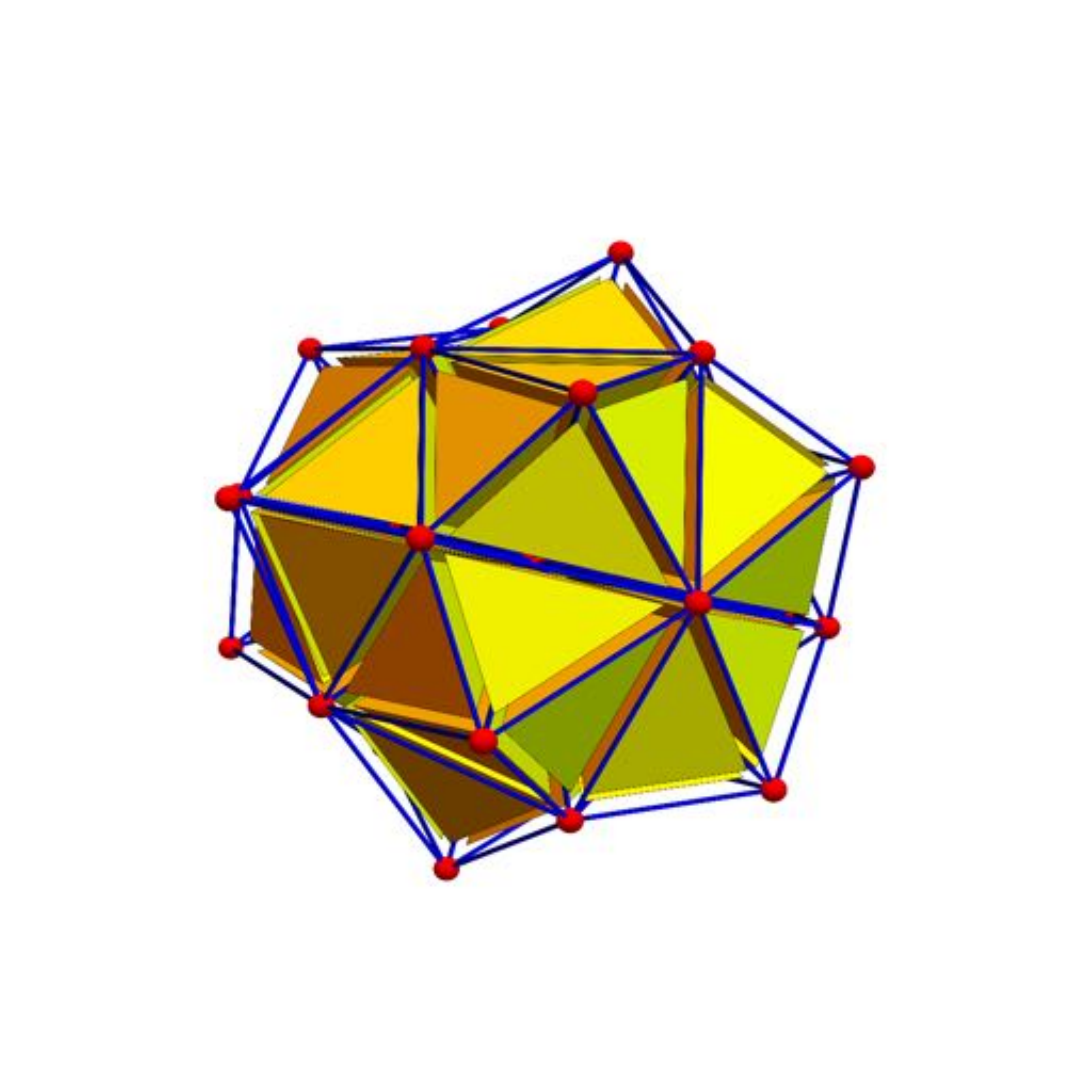}}
\scalebox{0.12}{\includegraphics{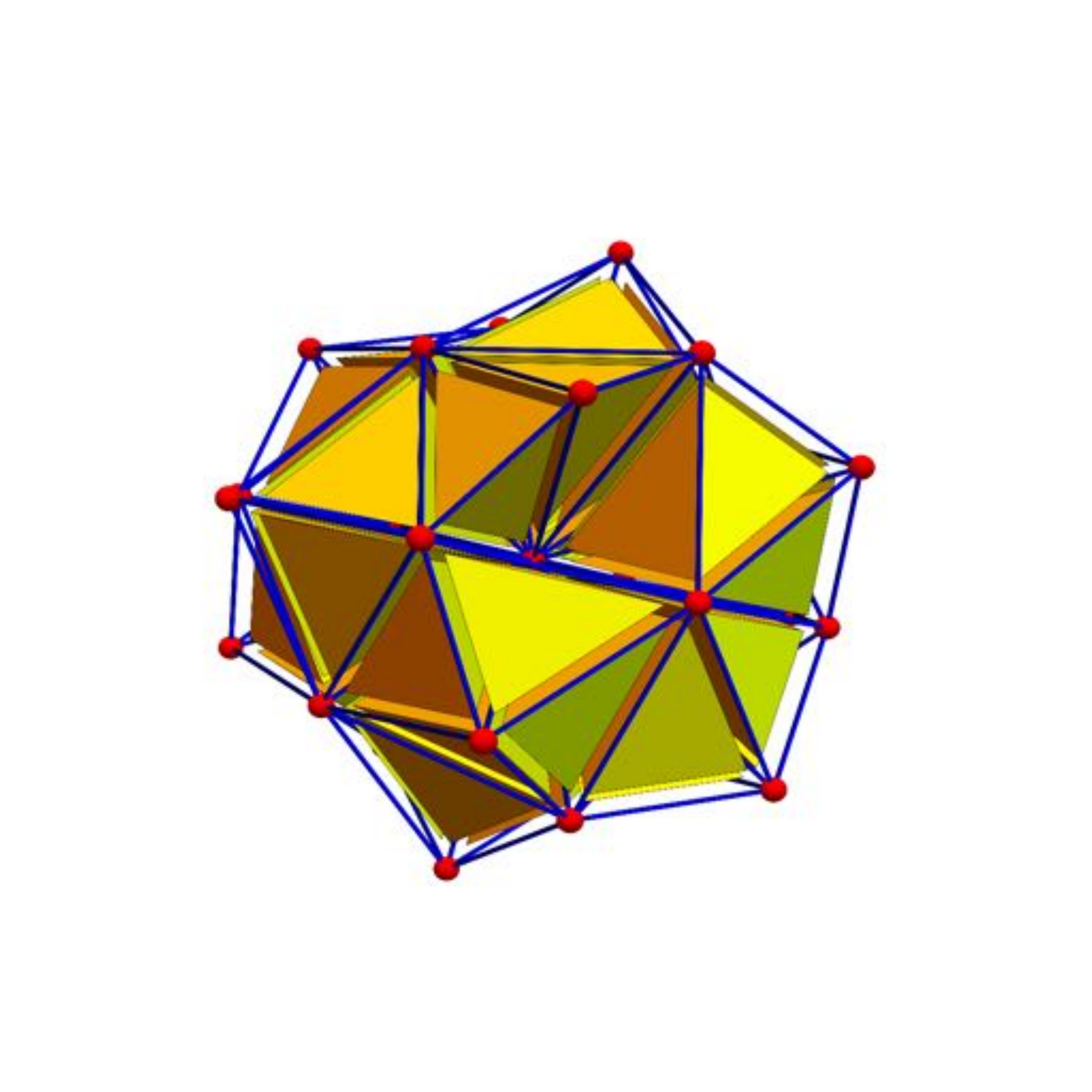}}
\caption{
\label{Cube}
The $3$-cube is a $3$-ball. 
It is the graph product $G_1 \times G_2 \times G_2 = P_2 \times P_2 \times P_2$, where
$P_2=K_2$ is the path graph of length $1$ (or complete $1$-dimensional graph). 
Its vertices are all triples $(x_1,x_2,x_3)$ where $x_i$ are simplices in $G_i$ and two are 
connected if one is contained in the other. To the right we see the generalized $3$-ball $H$ without
interior. A prismatic hole (Diamond graph) was drilled out at the boundary 
to expose the interior point to the boundary. After discarding
the interior edges we have a 3-ball without interior. 
It is Hamiltonian because it actually has become a 2-sphere. 
Because $H$ and $G$ have the same vertices, the Hamiltonian path in $H$ is
also a Hamiltonian path in $G$. 
}
\end{figure}

\begin{figure}
\scalebox{0.12}{\includegraphics{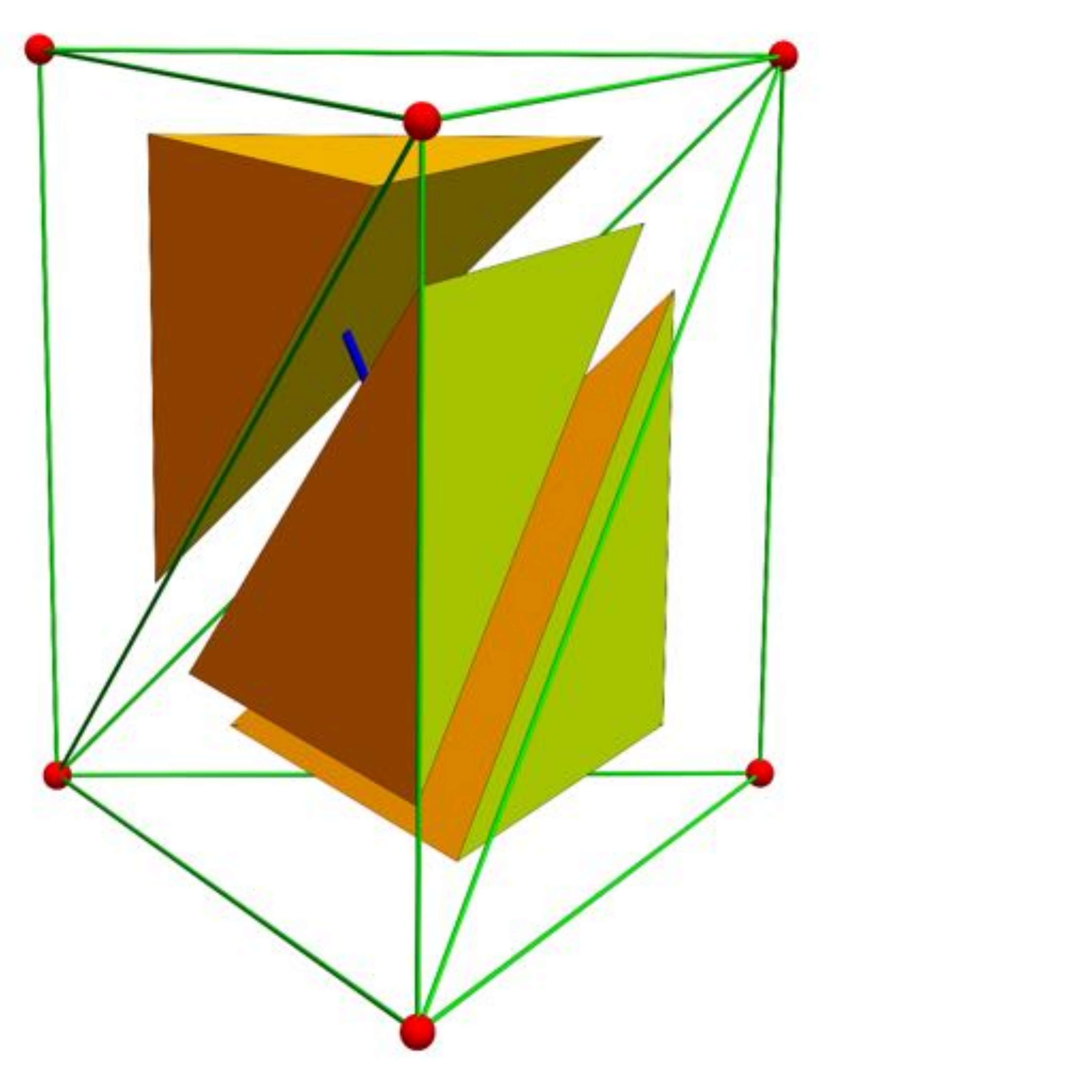}}
\scalebox{0.12}{\includegraphics{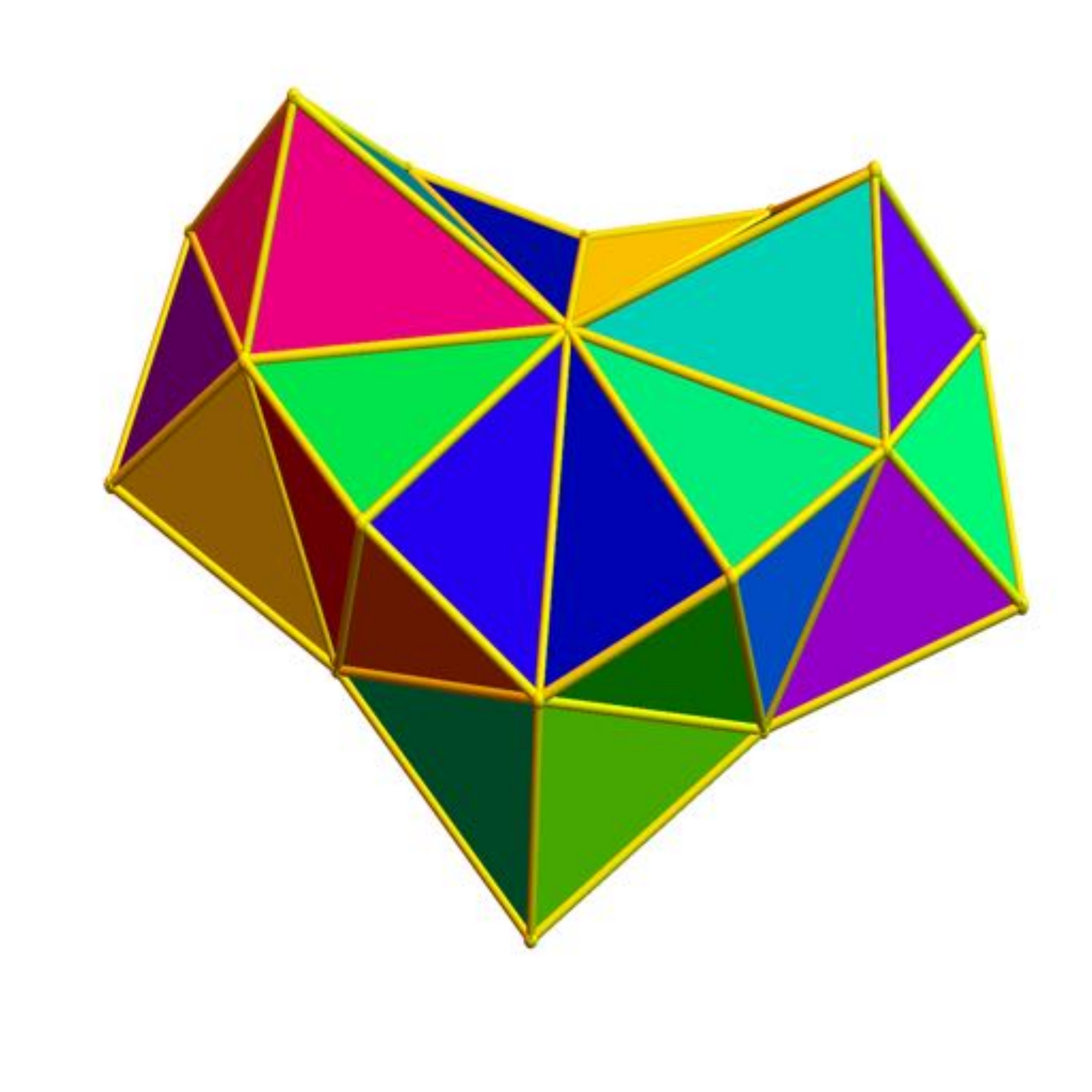}}
\caption{
\label{Prism}
The $3$-dimensional prism $G$. It is not a 3-ball because
some unit spheres at the boundary are triangles. It plays an important
role when briding different 3-graphs $f<0$ and $f>0$. To the right we 
see the Barycentric refinement of the boundary complex. It is a 2-sphere. 
The vertices of this graph are given by the $2$-skeleton complex of $G$
and two vertices are connected if one is contained in the other. 
}
\end{figure}

\begin{figure}
\scalebox{0.14}{\includegraphics{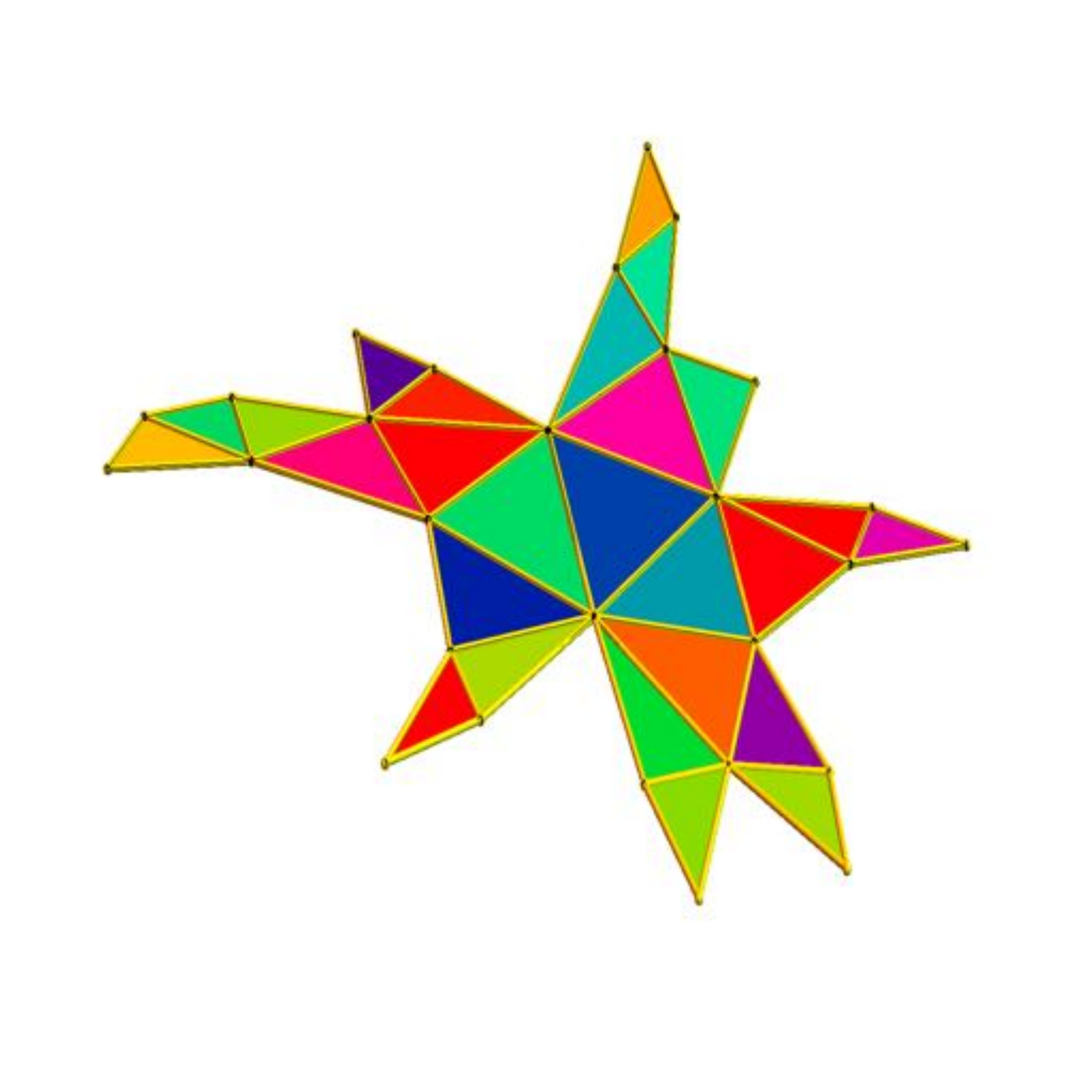}}
\scalebox{0.14}{\includegraphics{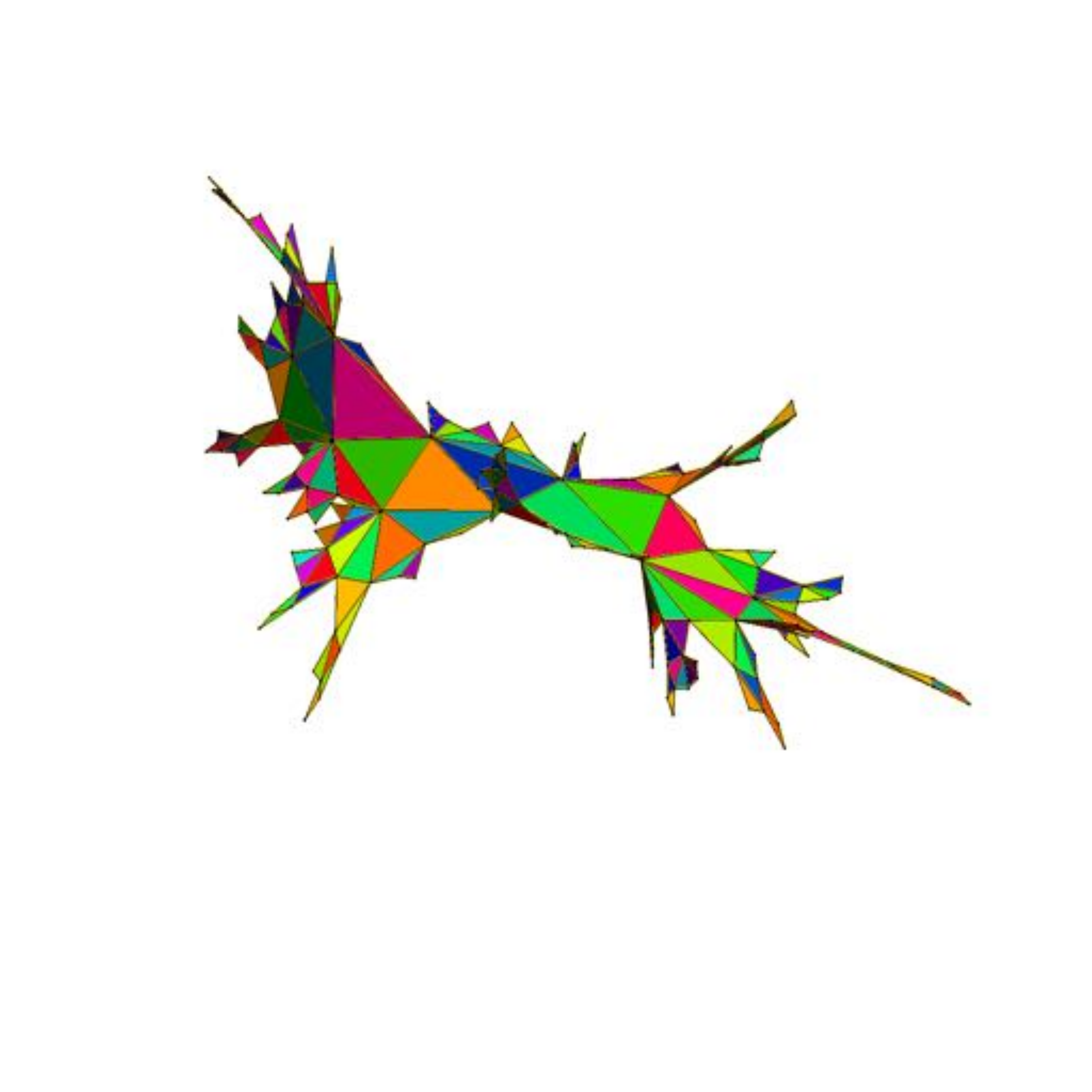}}
\caption{
\label{NoInterior}
Two shellable complexes which are not 2-balls. They are planar but not
4-connected so that Tutte does not apply. They are Hamiltonian.
The Hamiltonian path is the boundary curve. We could even add more edges 
building a 2-ball which still is Hamiltonian. 
}
\end{figure}

\begin{figure}
\scalebox{0.14}{\includegraphics{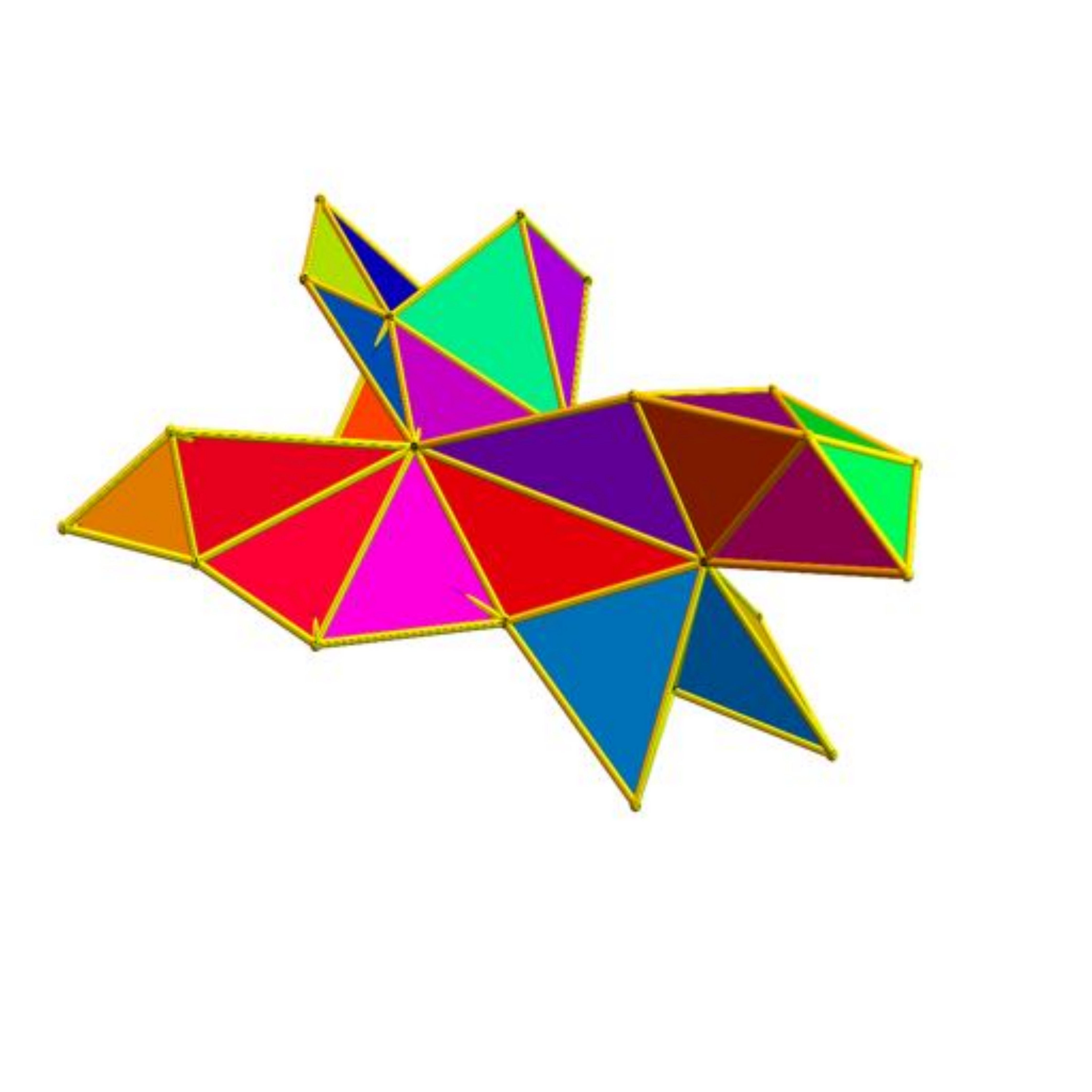}}
\scalebox{0.14}{\includegraphics{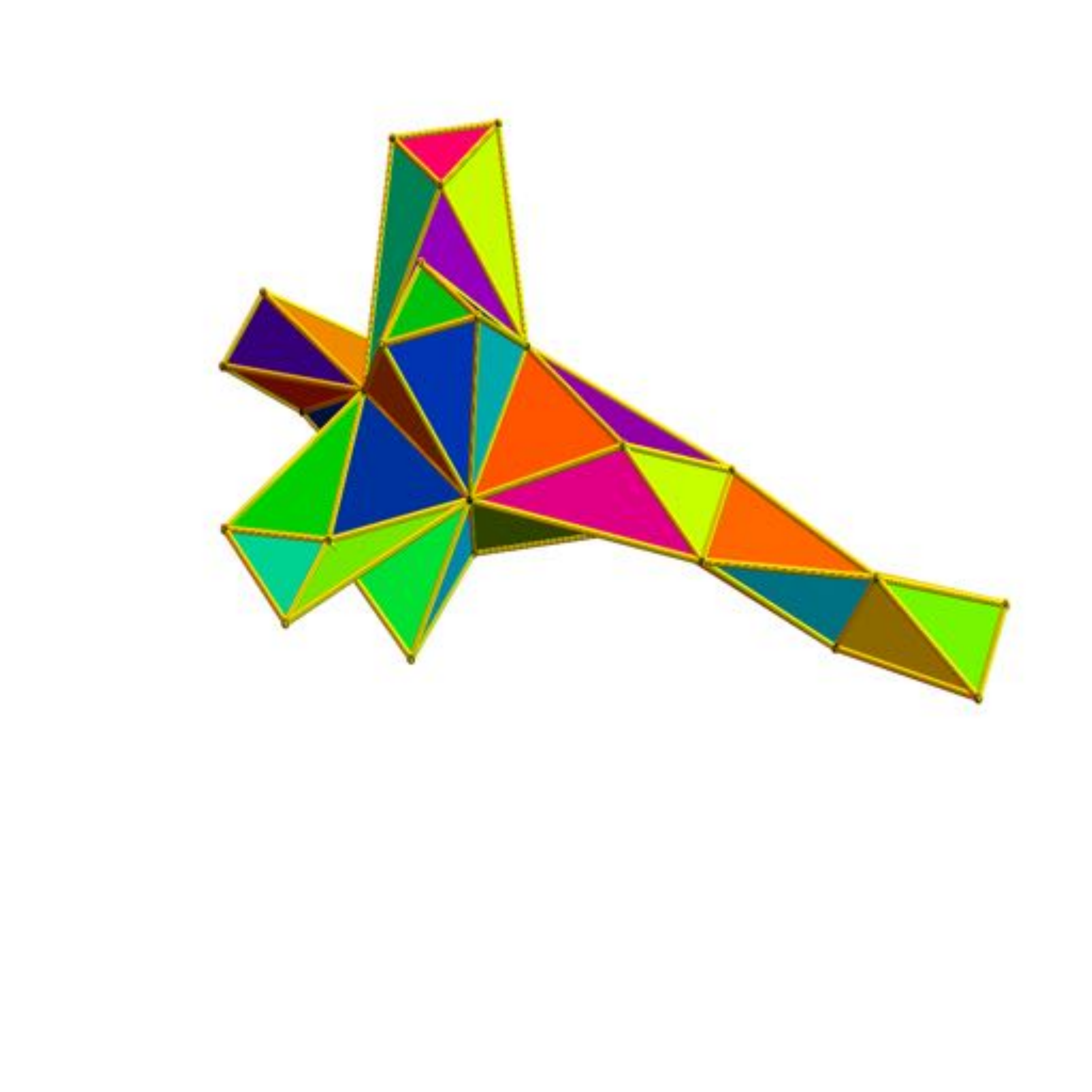}}
\caption{
\label{NoInterior}
Two shellable 3-complexes. The first one was obtained by gluing
successive Avici graphs to random places of the previous case. 
In the second case, we glue octahedra together. 
}
\end{figure}

\begin{figure}
\scalebox{0.14}{\includegraphics{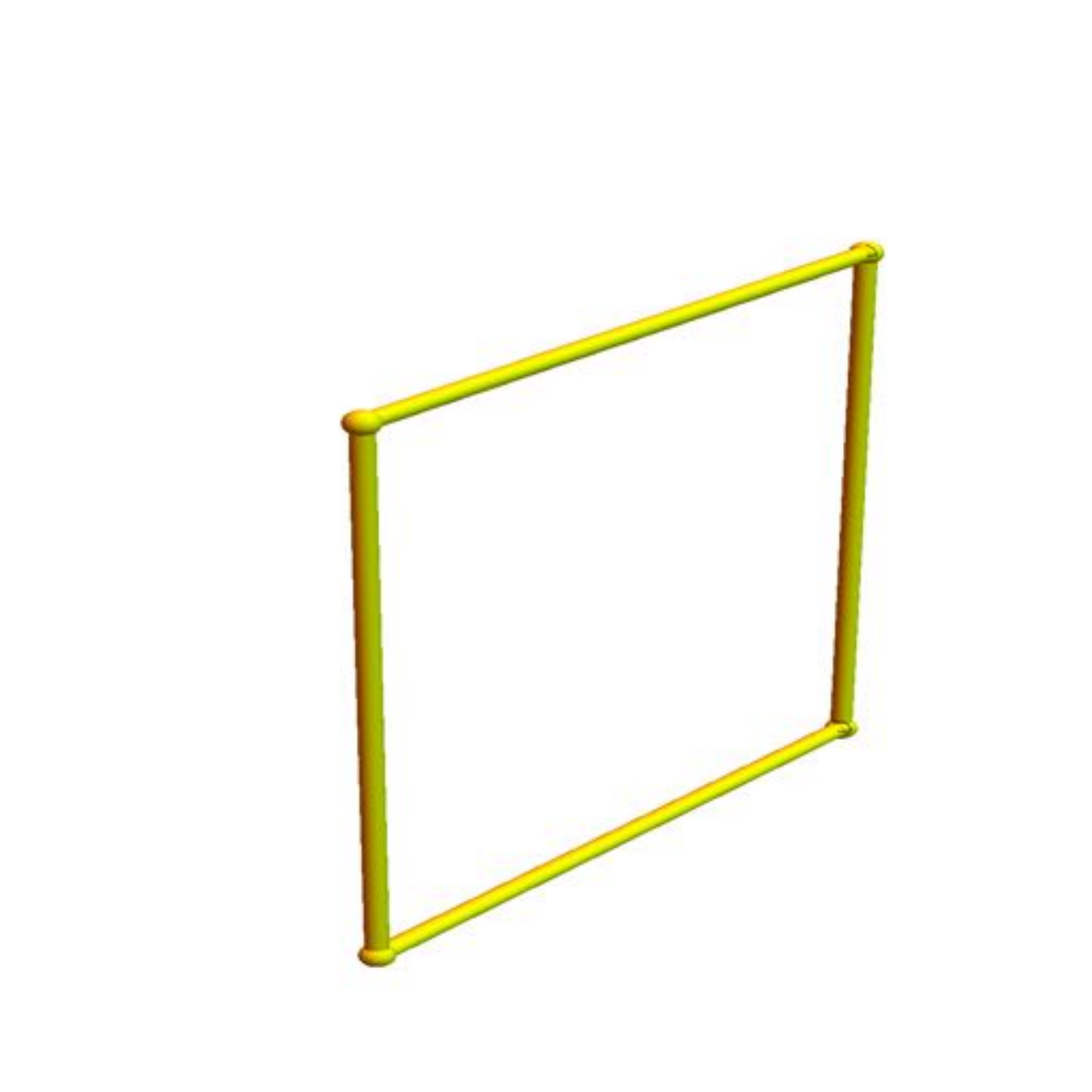}}
\scalebox{0.14}{\includegraphics{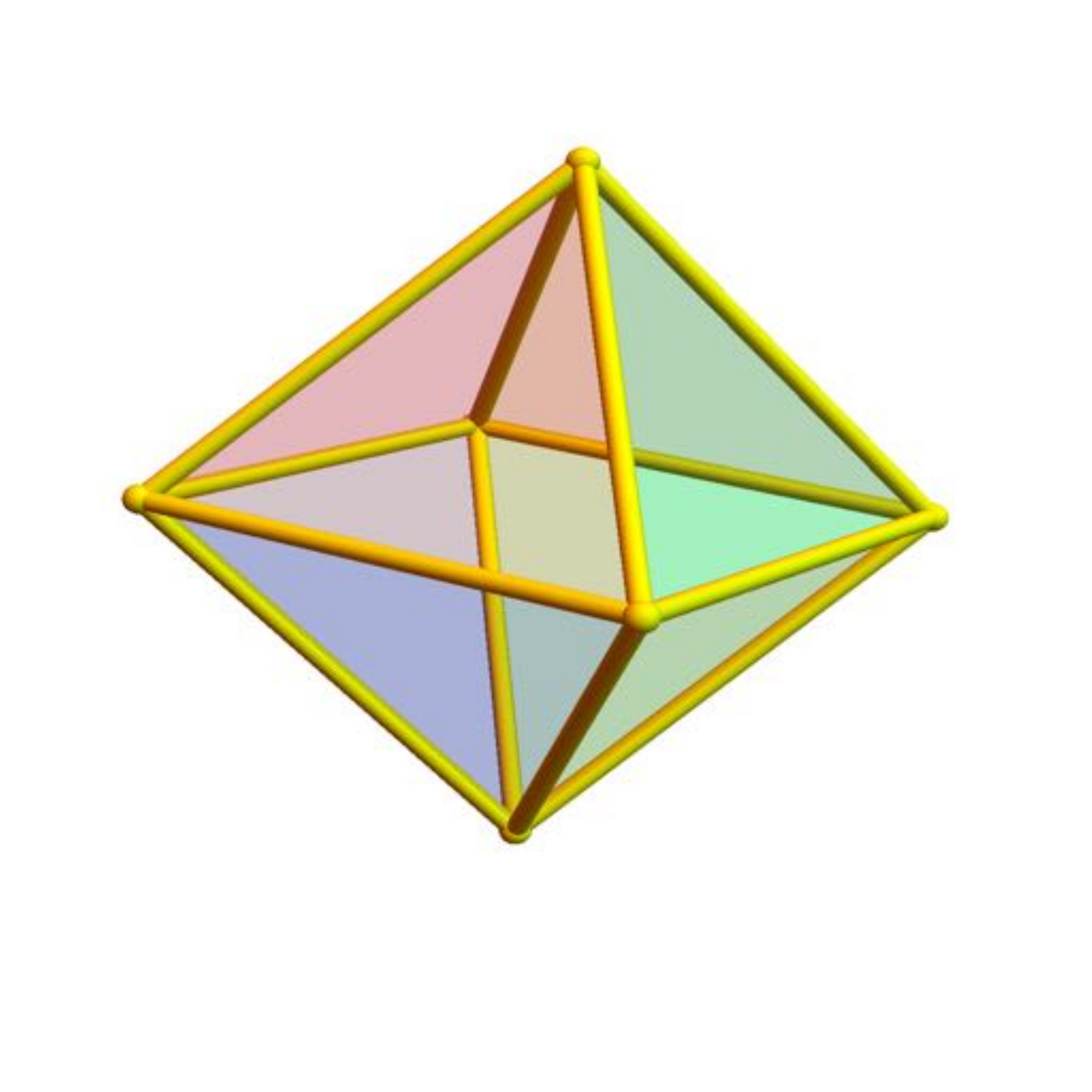}}
\scalebox{0.14}{\includegraphics{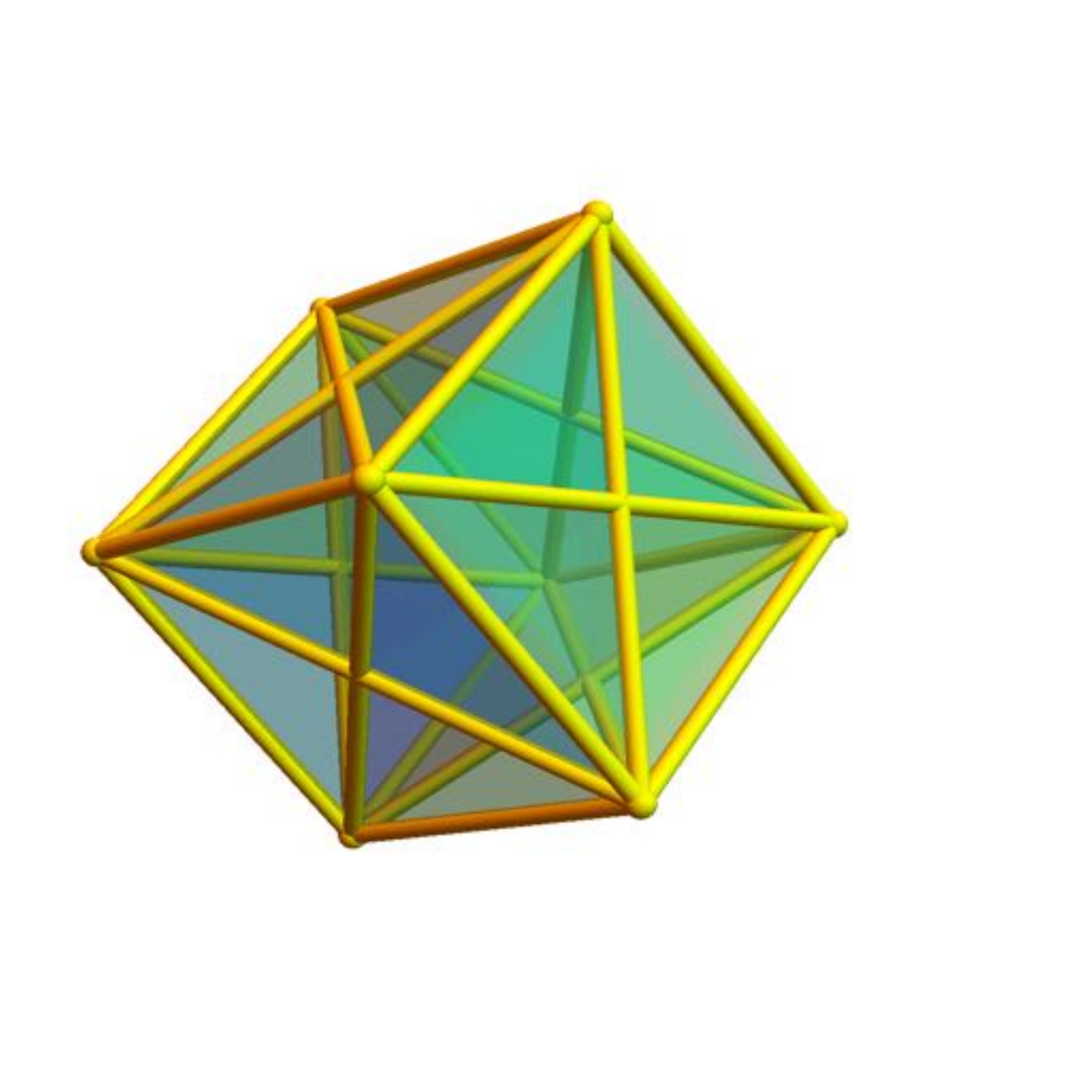}}
\scalebox{0.14}{\includegraphics{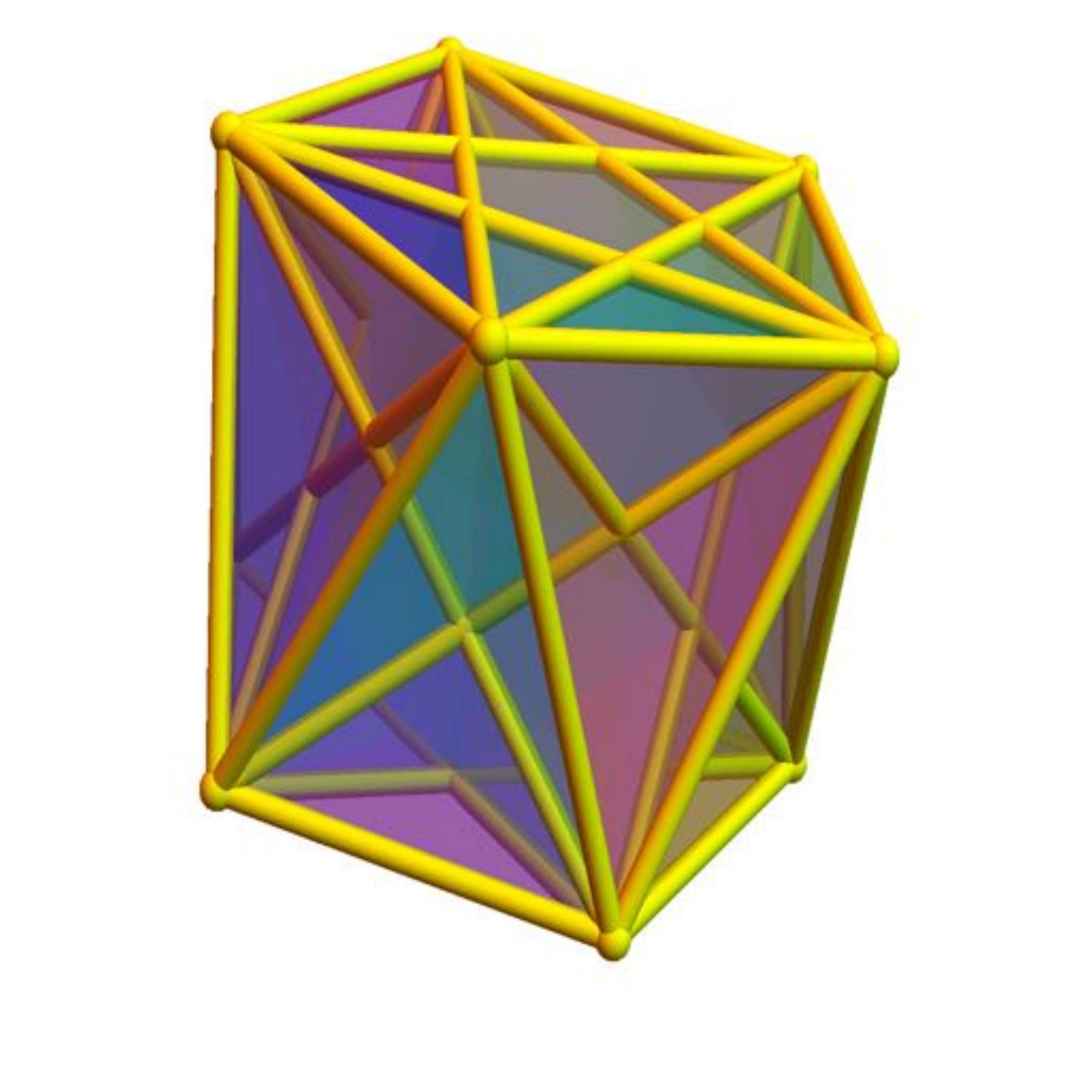}}
\caption{
\label{Spheres}
The smallest 1-sphere (cyclic graph $C_4$),
then the smallest 2-sphere (octahedron), 3-sphere and 4-sphere. 
The 3-sphere (16-cell) is the suspension of the 2-sphere, the 
4-sphere (32-chamber) is the suspension of the 3-sphere. They are all Hamiltonian
as they are all d-graphs. 
}
\end{figure}

\begin{figure}
\scalebox{0.12}{\includegraphics{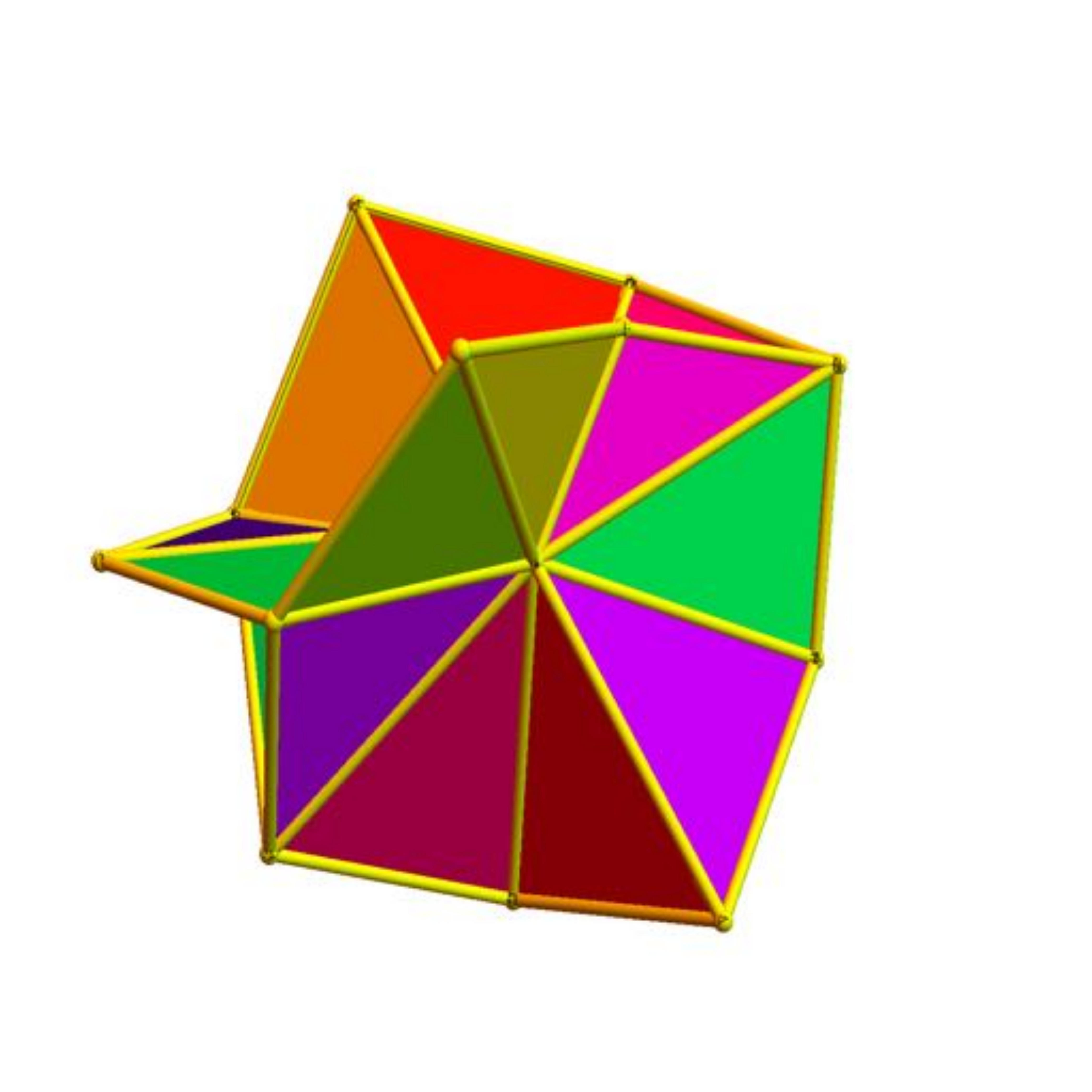}}
\scalebox{0.12}{\includegraphics{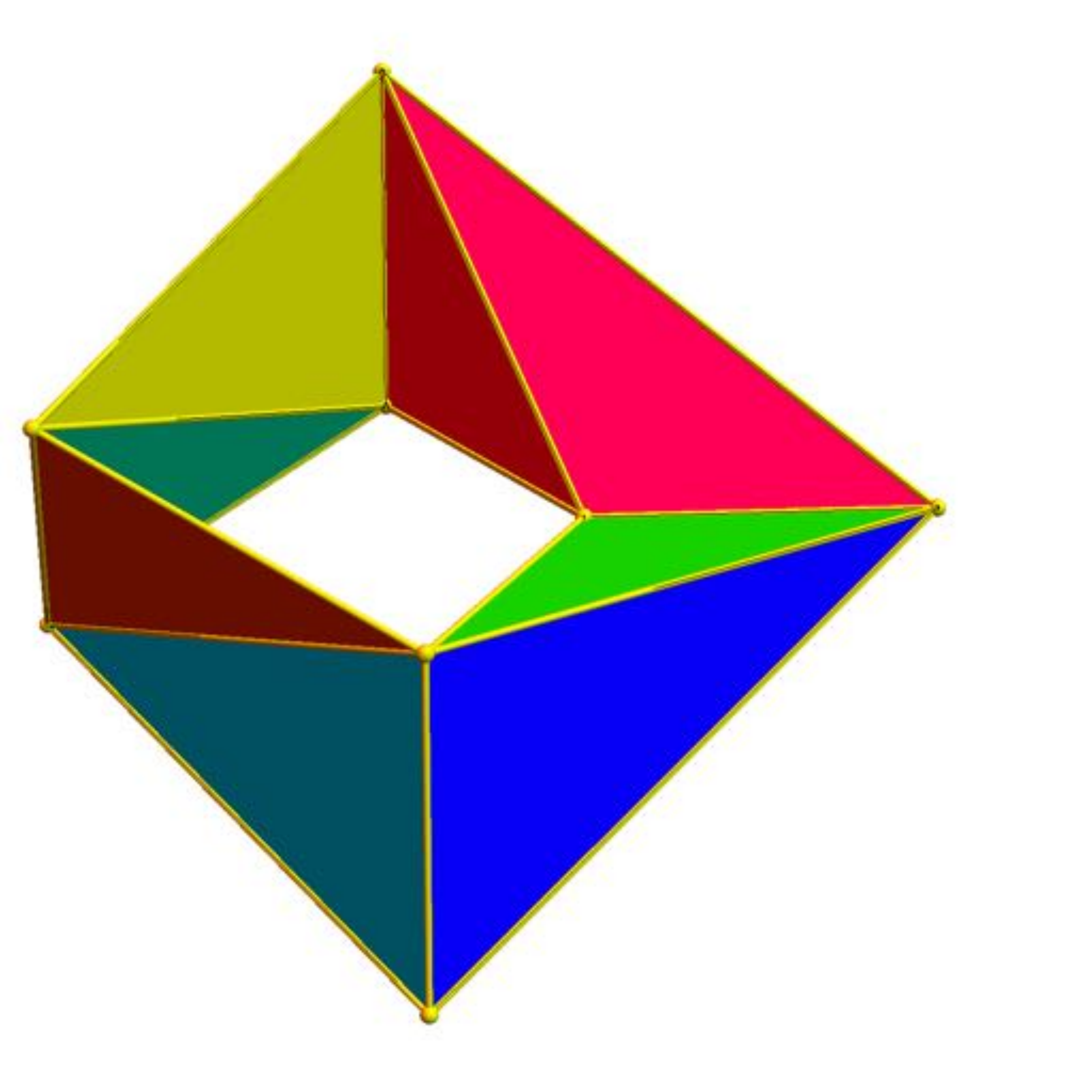}}
\caption{
\label{Duncehat}
The dunce hat is a non-shellable, non-contractible 2-complex
with $f$-vector $(17,52,36)$. It is not a 2-graph as some
unit spheres $S(x)$ have Euler characteristic $\chi(S(x))=-1$. 
It is Hamiltonian, even so the theorem proven here 
does not cover it.  To the right, we see the 
M\"obius strip, a non-shellable complex. It is a non-orientable
generalized 2-graph with boundary. It is Hamiltonian, the boundary is the 
Hamiltonian path. 
}
\end{figure}

\begin{figure}
\scalebox{0.14}{\includegraphics{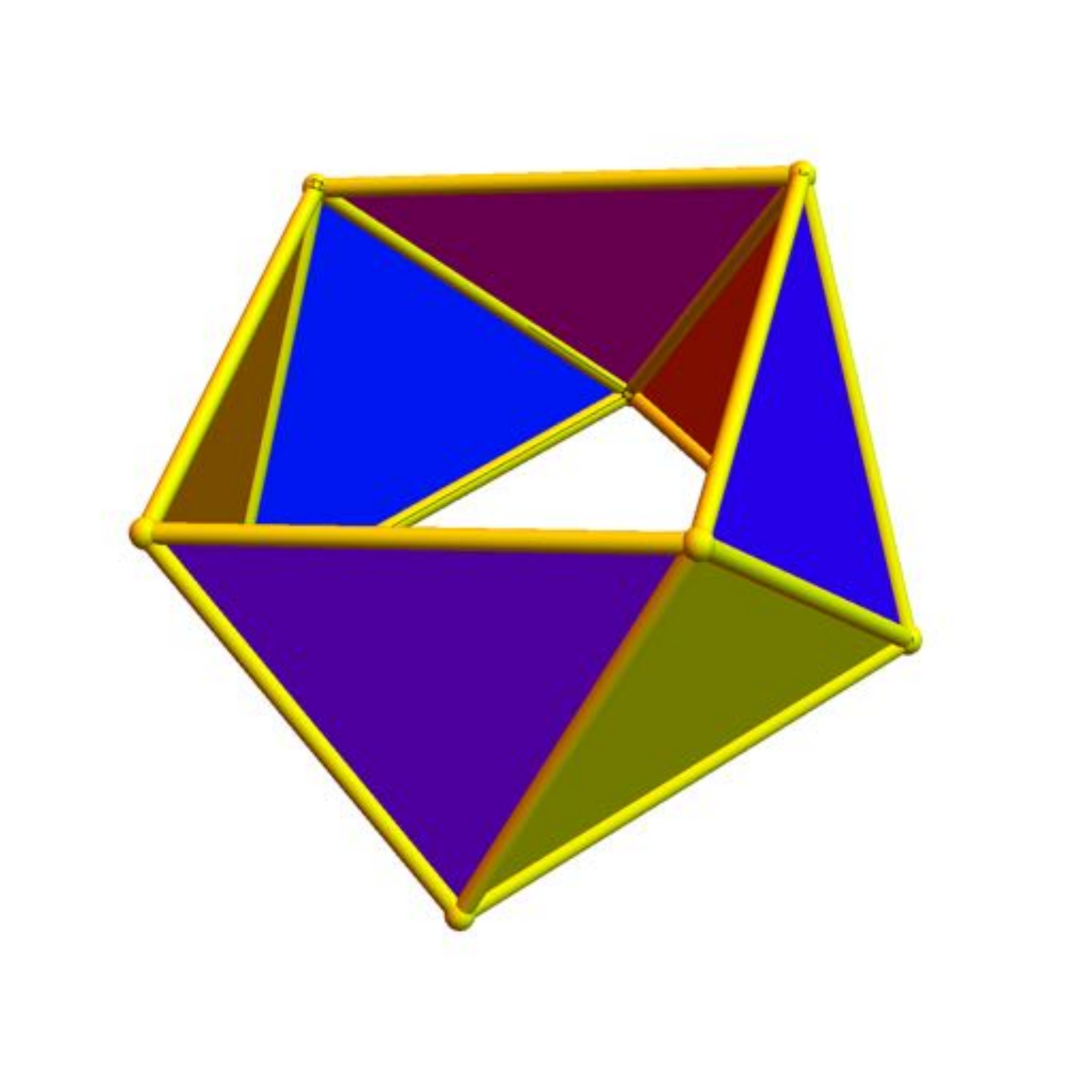}}
\scalebox{0.14}{\includegraphics{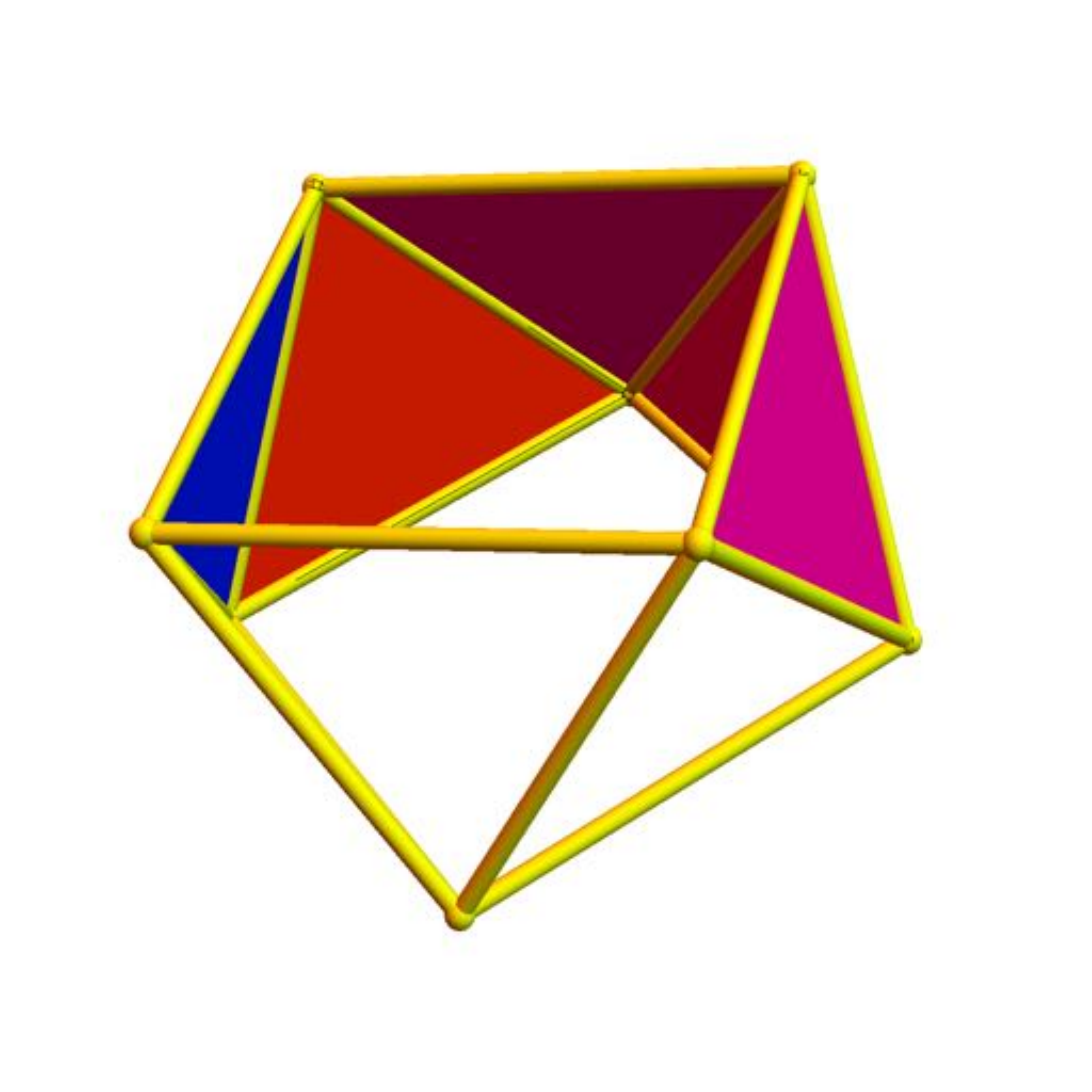}}
\caption{
\label{Cylinder}
A Hamiltonian path for a cylinder without interior point can be
constructed by building a bridge between the two boundary curves.
This is equivalent with rewiring using a prism bridge. 
}
\end{figure}

\begin{figure}
\scalebox{0.14}{\includegraphics{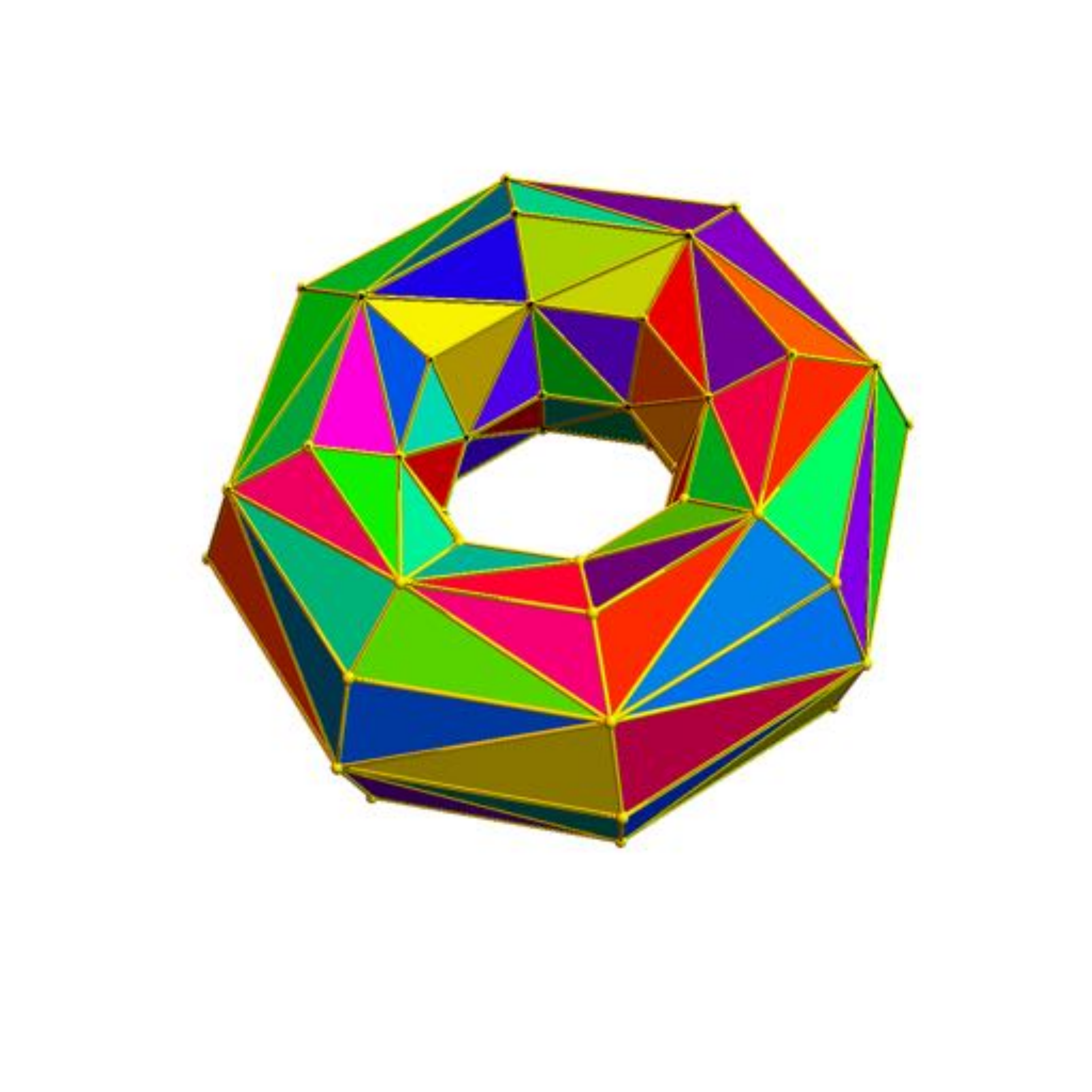}}
\scalebox{0.14}{\includegraphics{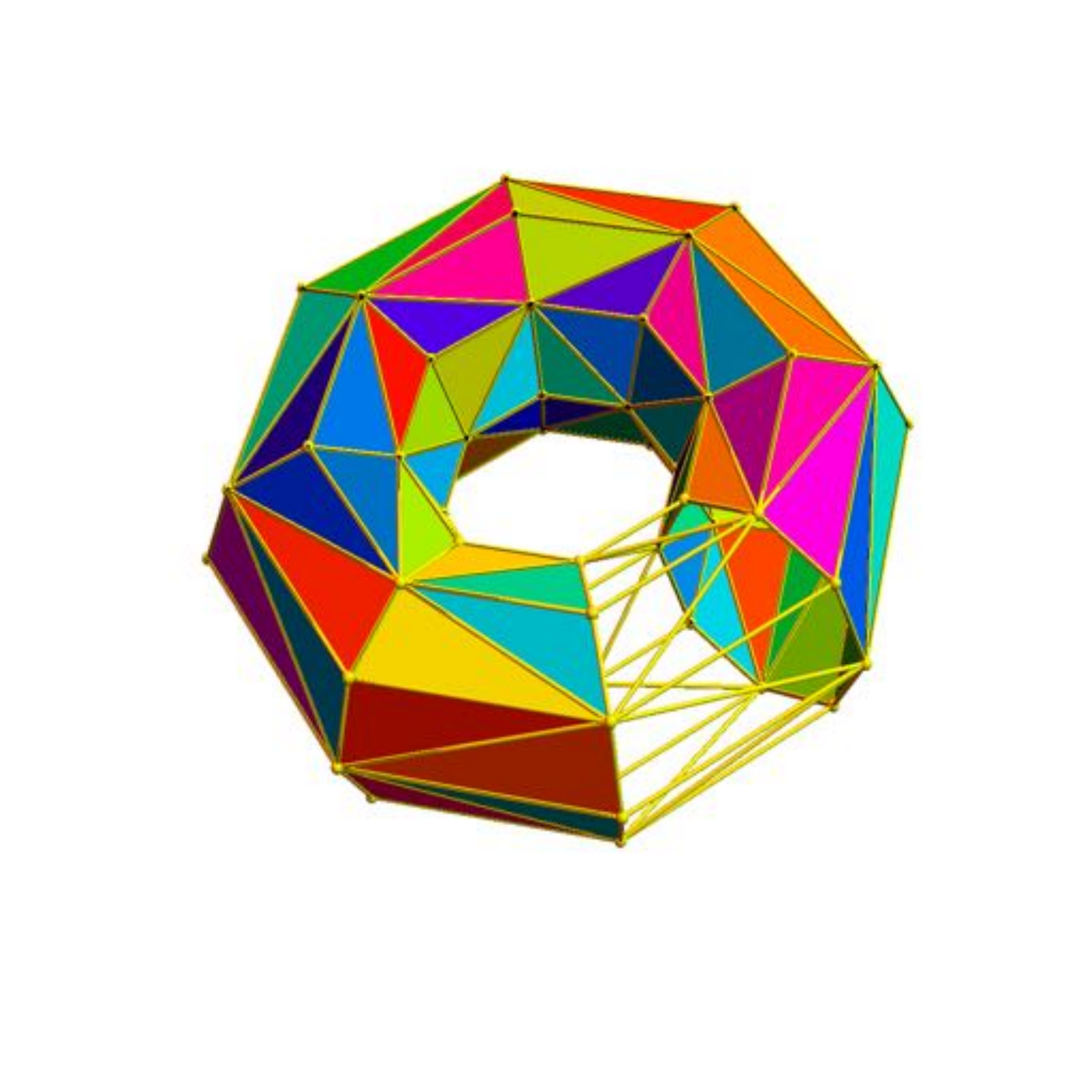}}
\caption{
\label{Torus}
A 2-torus is a 2-graph: every unit sphere is a circular graph. 
It is not shellable. To find a Hamiltonian path, find one first
in an annulus without interior, then use bridges to connect these
paths. 
}
\end{figure}

\begin{figure}
\scalebox{0.14}{\includegraphics{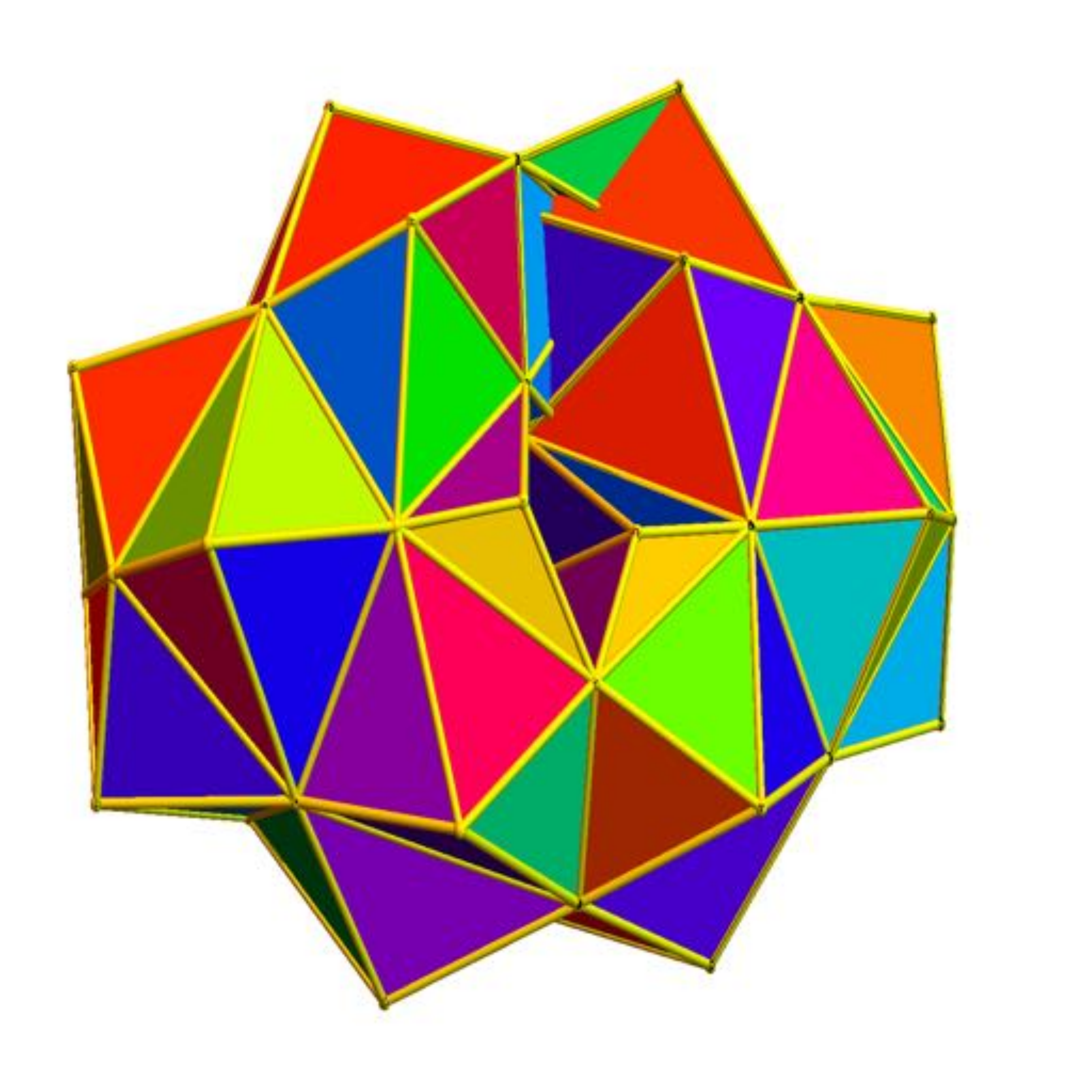}}
\scalebox{0.14}{\includegraphics{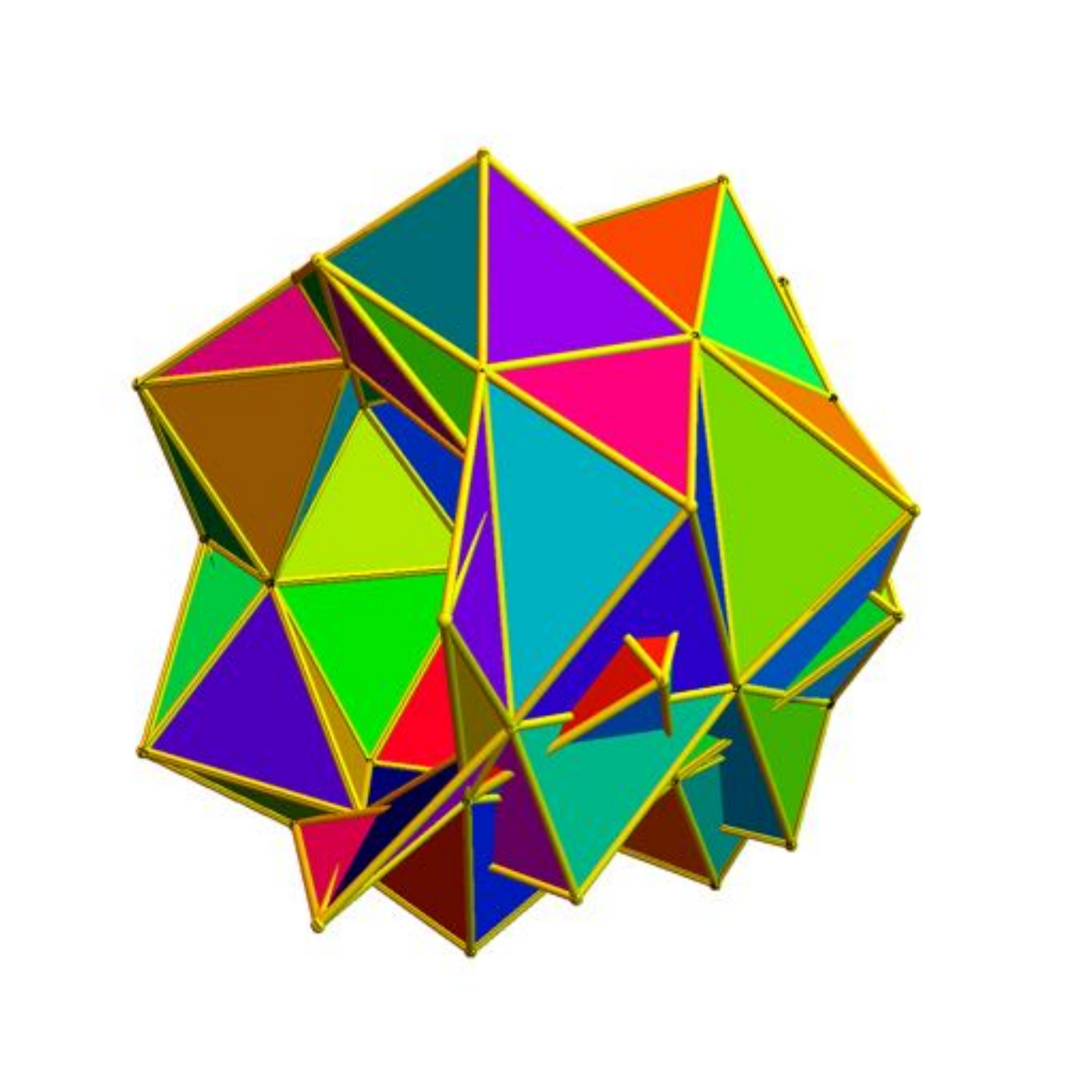}}
\caption{
\label{projective plane and kleinbottle }
The projective plane and the Klein bottle. 
Both are 2-graphs and Hamiltonian. 
}
\end{figure}

\begin{figure}
\scalebox{0.14}{\includegraphics{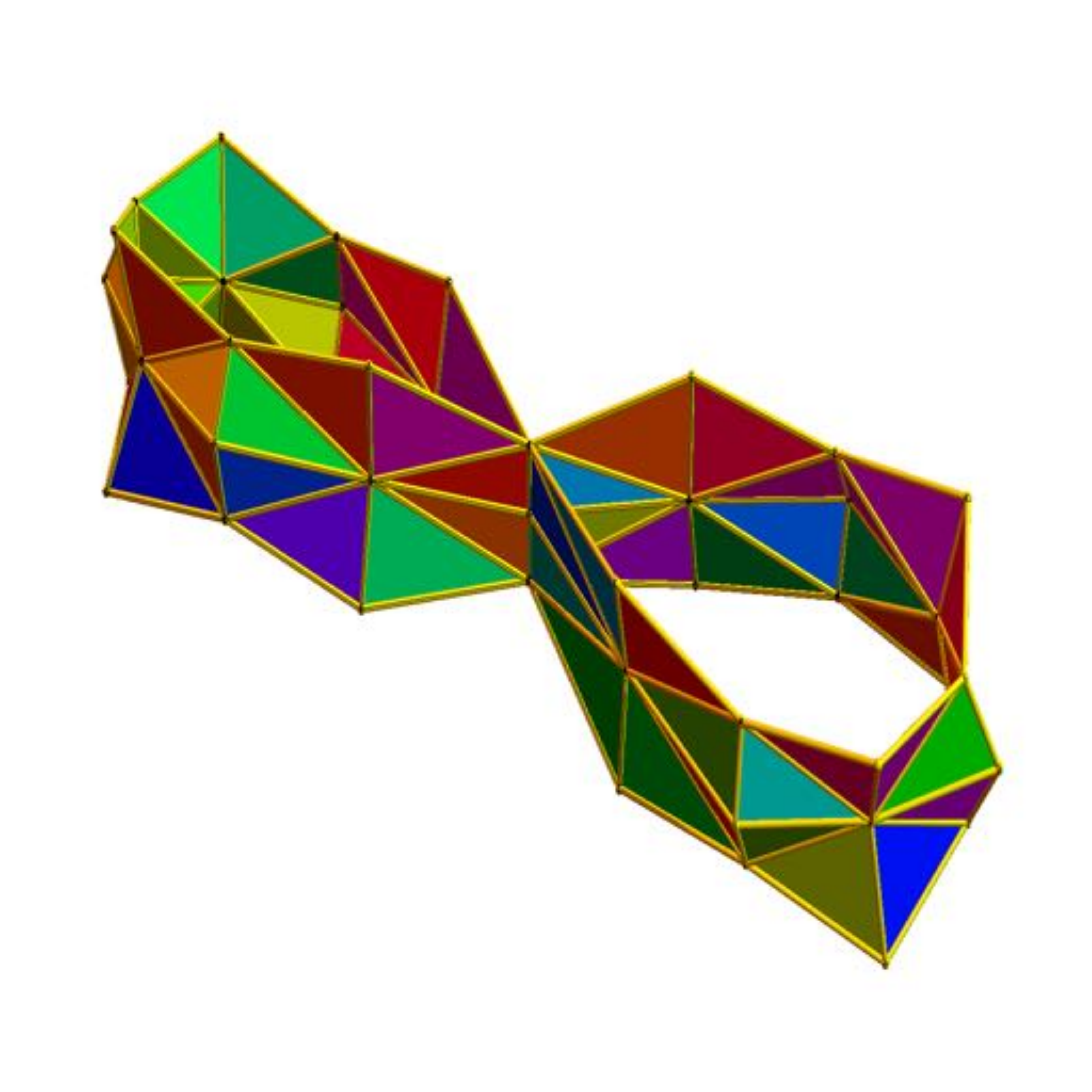}}
\scalebox{0.14}{\includegraphics{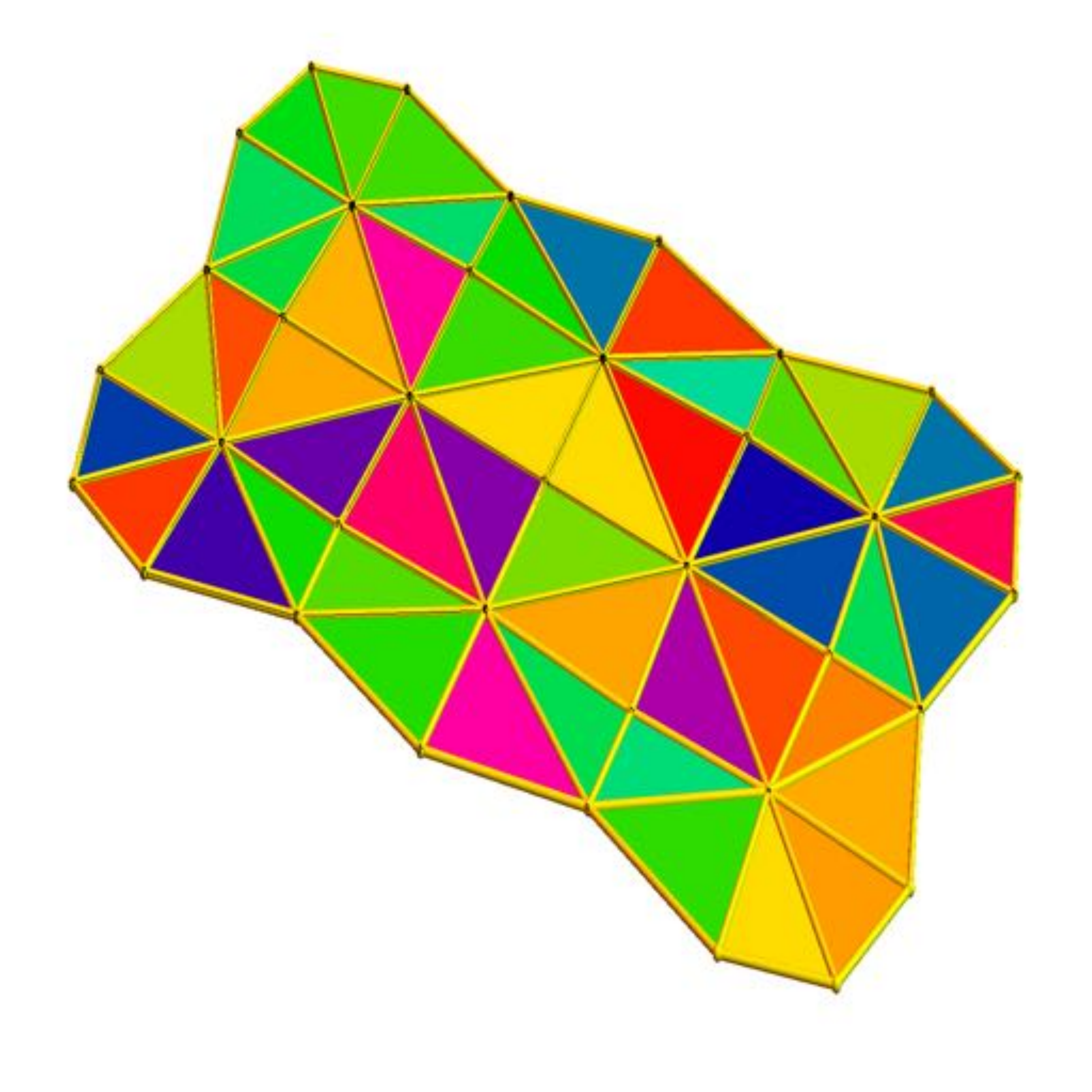}}
\caption{
\label{Product Graph2}
The product of a figure 8 graph and $K_2$ is a two dimensional
simplicial complex, but not a 2-graph with boundary. It is not
Hamiltonian. To the right, we see the product of two path graphs
which is Hamiltonian.
}
\end{figure}

\begin{figure}
\scalebox{0.14}{\includegraphics{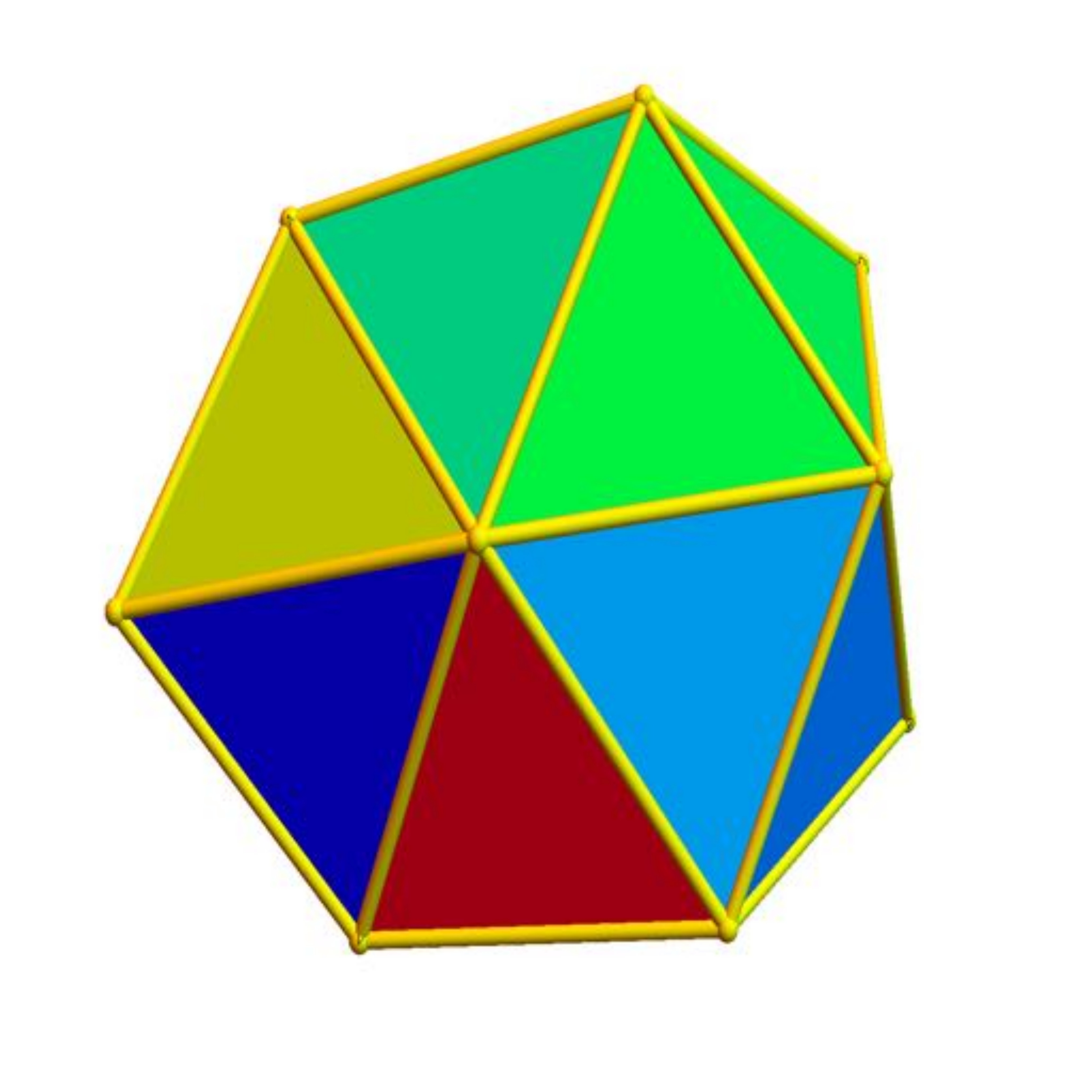}}
\scalebox{0.14}{\includegraphics{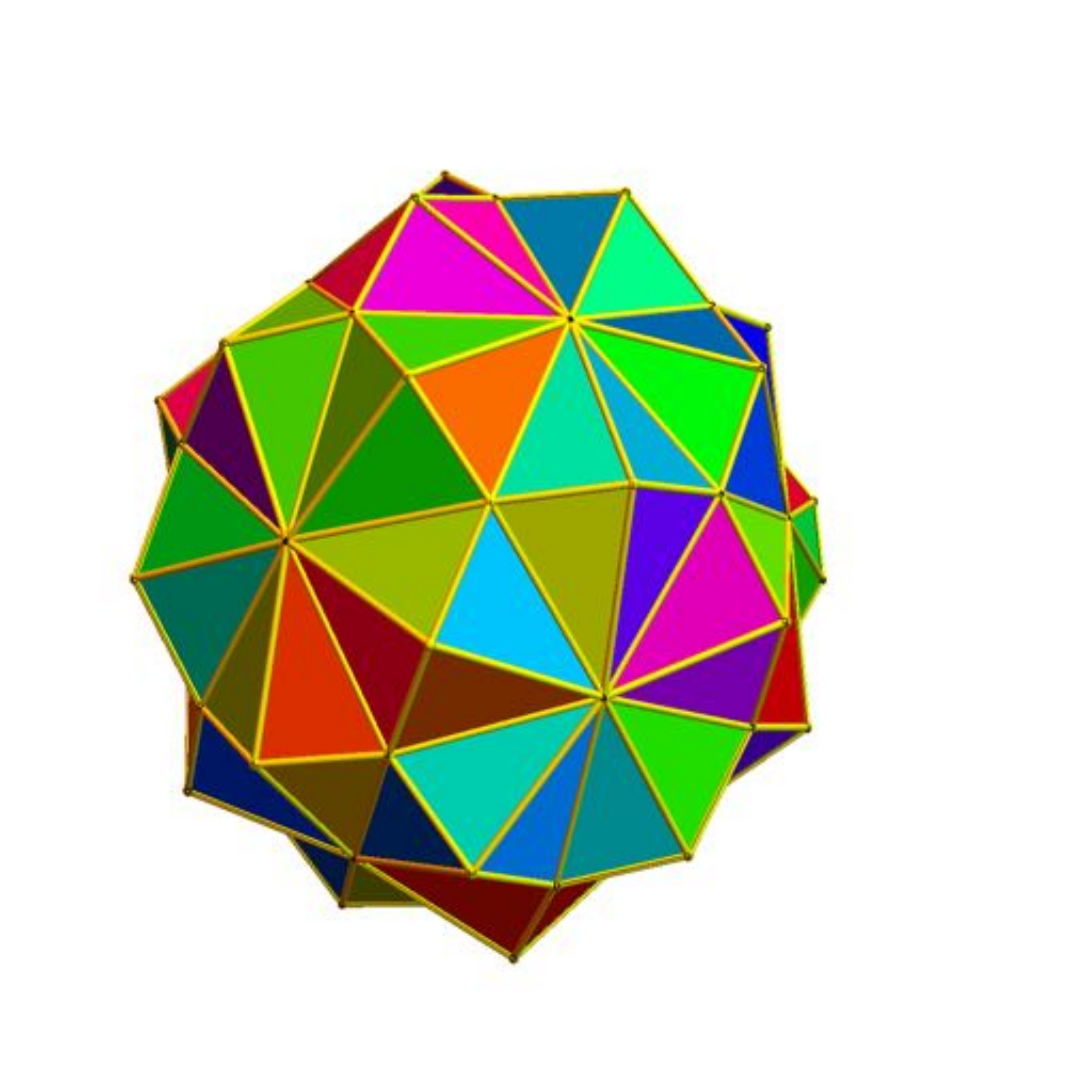}}
\scalebox{0.14}{\includegraphics{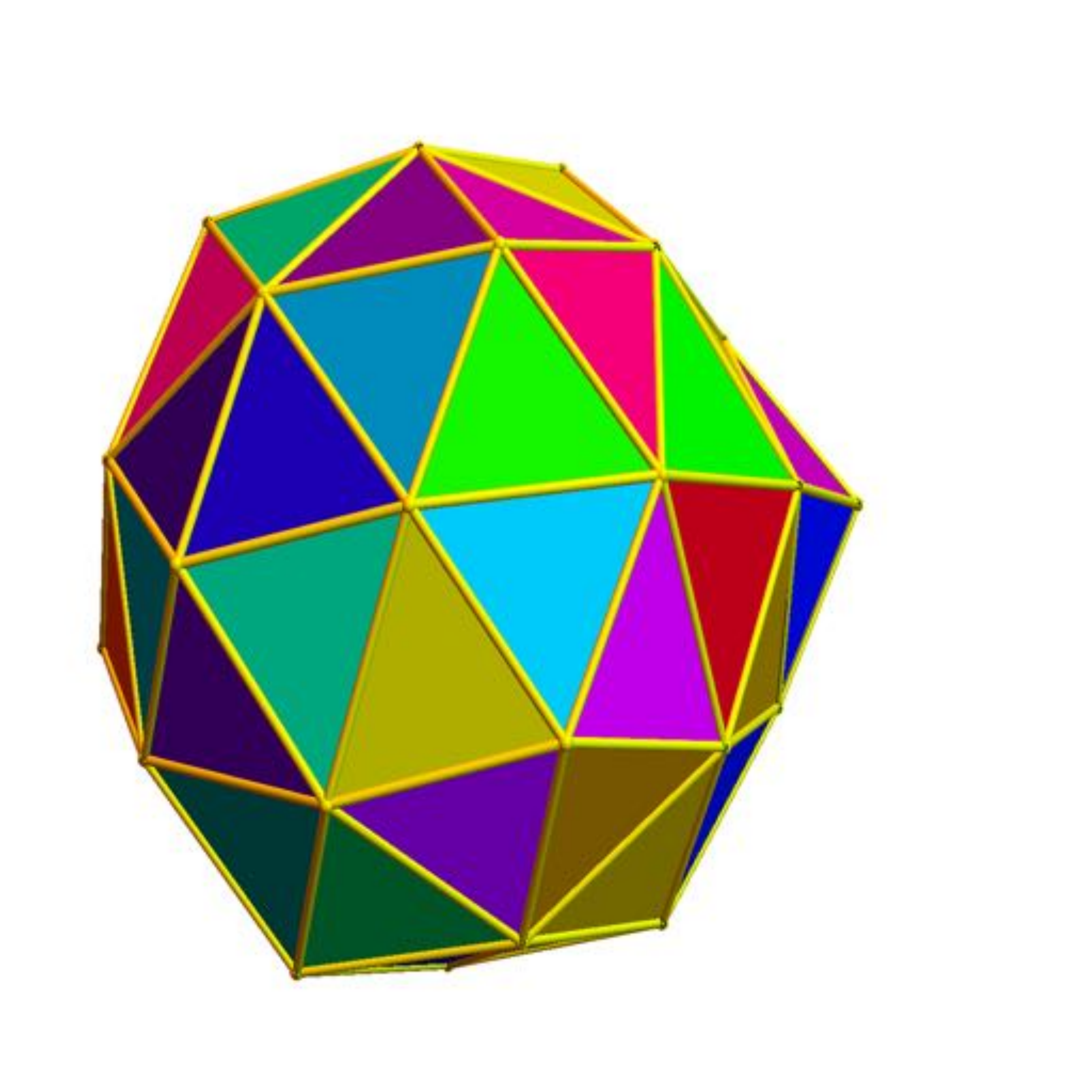}}
\scalebox{0.14}{\includegraphics{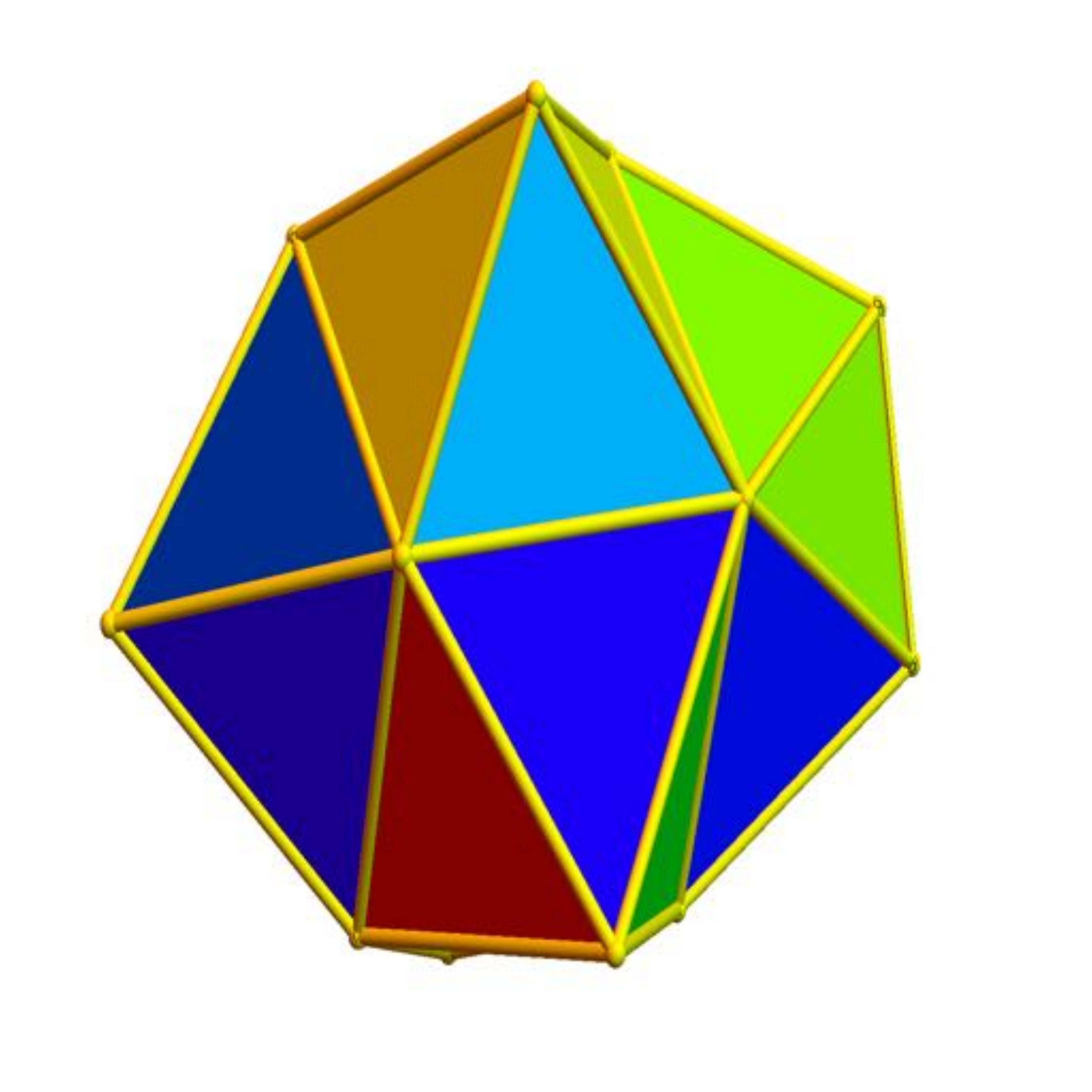}}
\caption{
\label{EdgeRefinements}
Edge refinements of the icosahedron,
the Disdyakis Triacontahedron, the Pentakis Dodecahedron and
Tetrakis Hexahedron. These are all 2-spheres. After adding a central 
point (building the join with a 1-point graph), we get in each case a 3-ball. 
As the spheres are Hamiltonian, also these unit balls are Hamiltonian (just
make a detour through the center once). 
}
\end{figure}

\begin{figure}
\scalebox{0.14}{\includegraphics{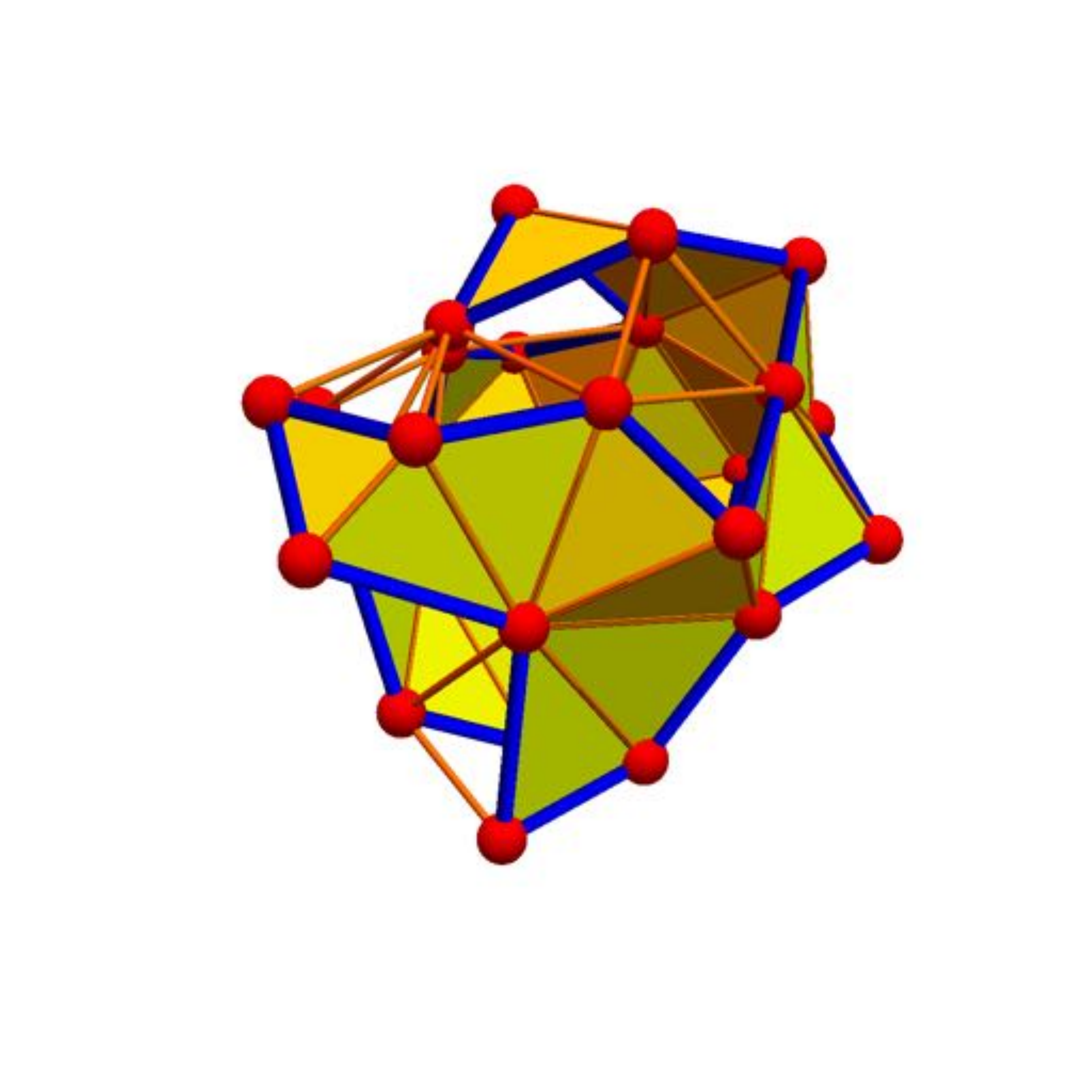}}
\scalebox{0.14}{\includegraphics{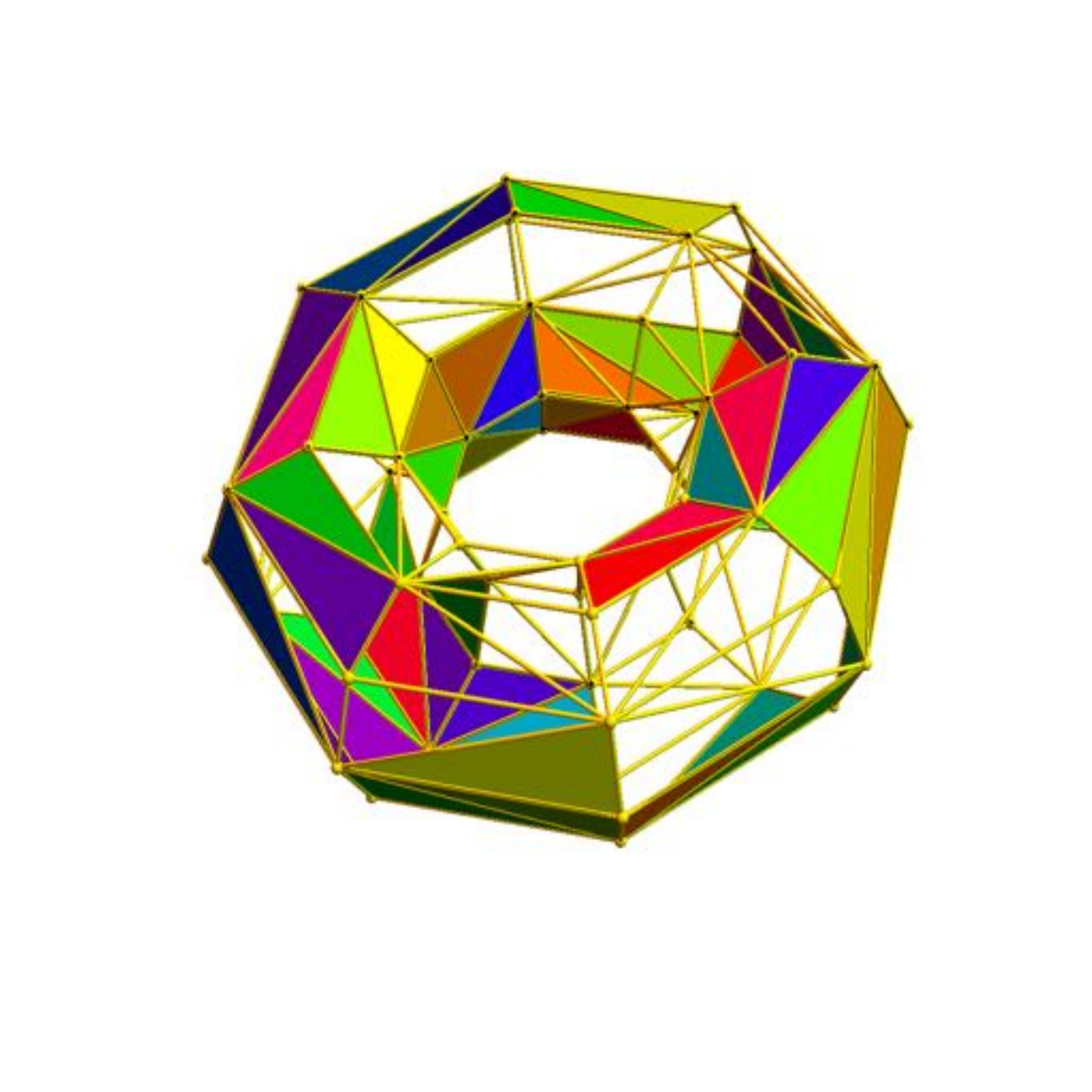}}
\caption{
\label{Automated find}
In a 2-graph, we can often find directly a Hamiltonian path by building up a random shellable 2-complex
without interior inside $G$. The cyclic boundary then is a Hamiltonian path. In general
this can be used to make the task of covering the rest smaller. This also works
in higher dimensions. We see the case of a sphere and a torus. 
}
\end{figure}

\begin{figure}
\scalebox{0.24}{\includegraphics{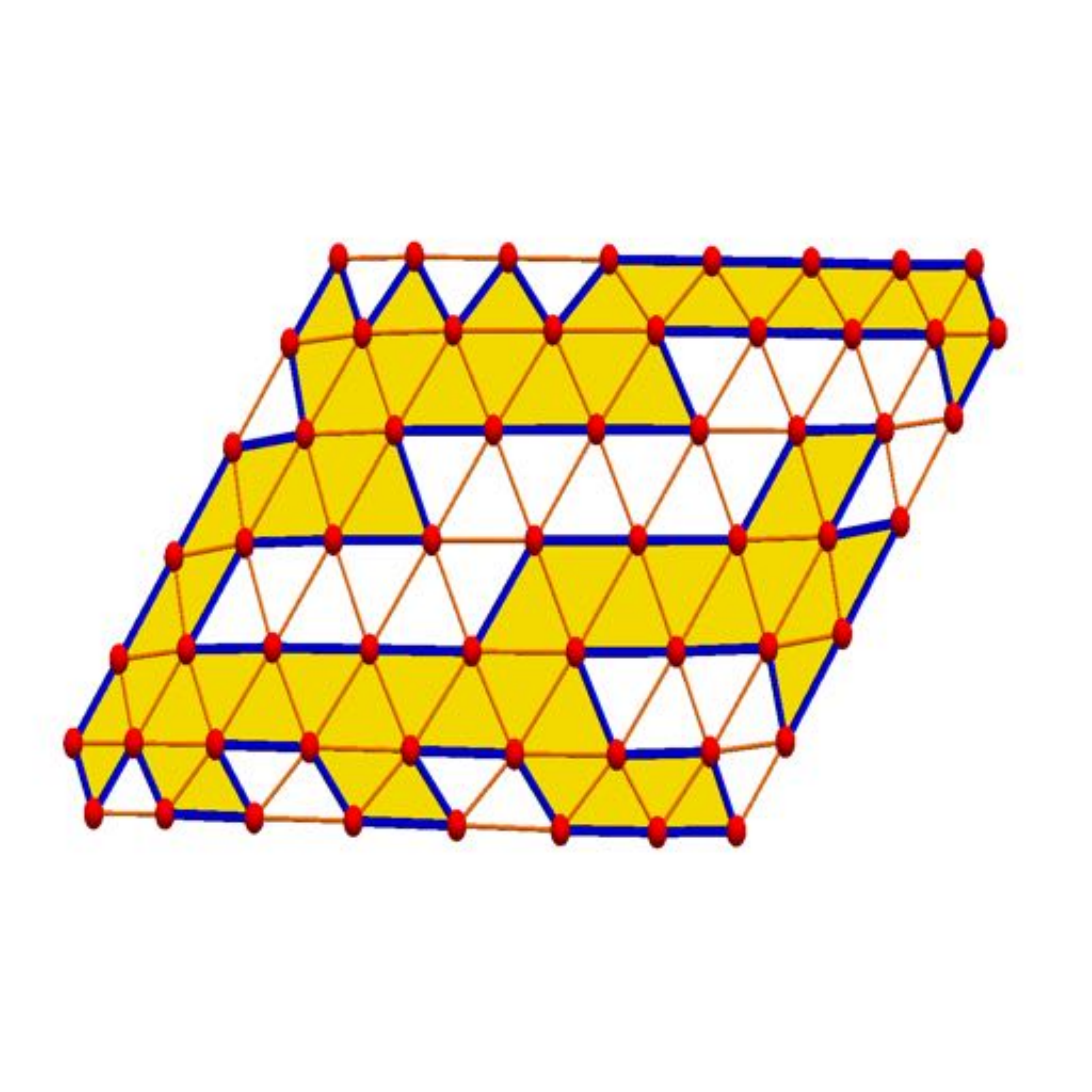}}
\caption{
\label{Hex}
A Hamiltonian path found by embedding a shellable simplicial complex without
interior points looks often like a Peano curve. It would also have been possible
to cut the region up into two, solve the problem in both and then merge them together. 
}
\end{figure}

\begin{figure}
\scalebox{0.34}{\includegraphics{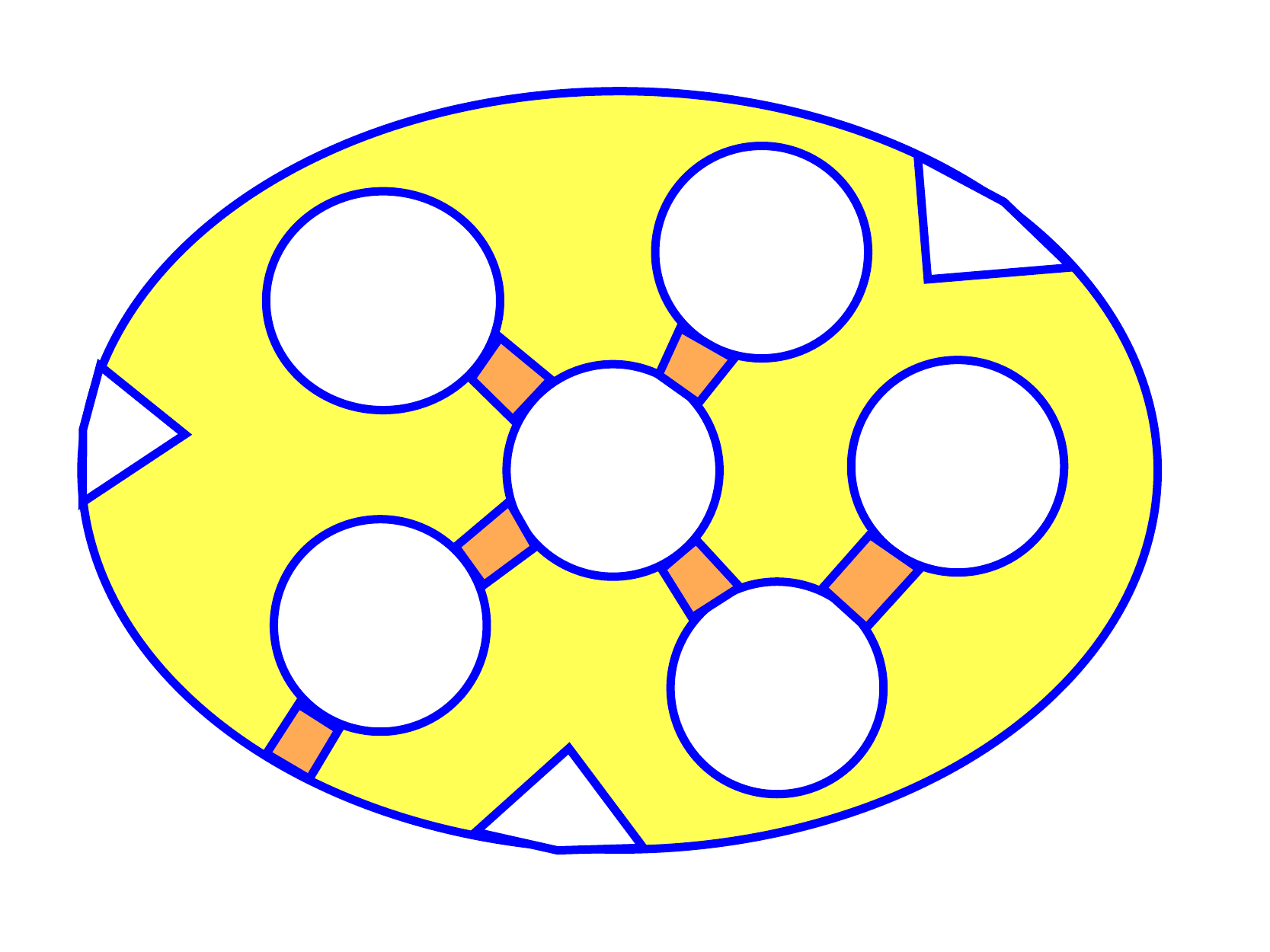}}
\caption{
\label{Cheese}
The Swiss cheese algorithm to construct a Hamiltonian path is the
idea to bild sphere caves around strongly interior points, building 
bridges between the different boundary components. Then access
other interior points from the boundary.
}
\end{figure}

\begin{figure}
\scalebox{0.07}{\includegraphics{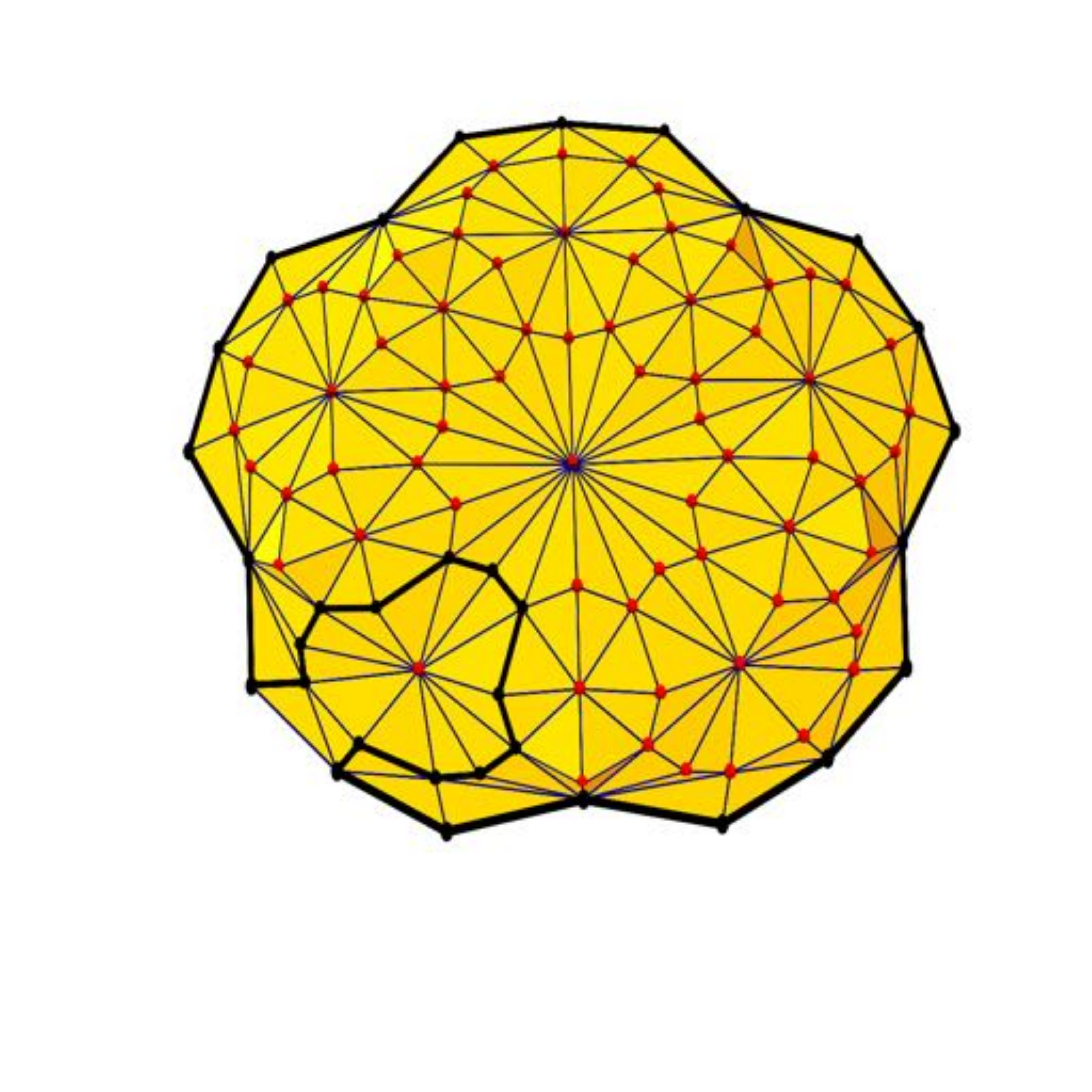}}
\scalebox{0.07}{\includegraphics{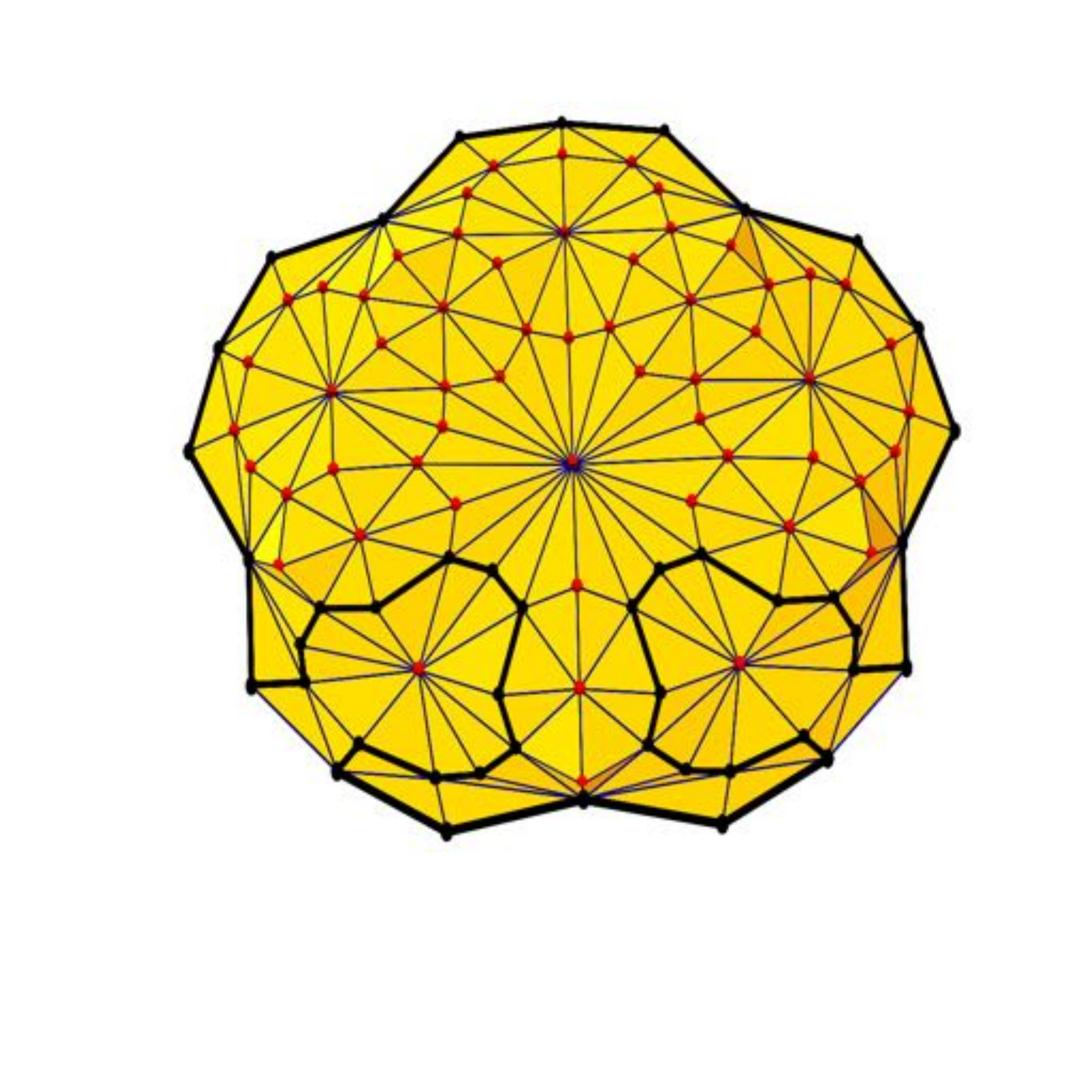}}
\scalebox{0.07}{\includegraphics{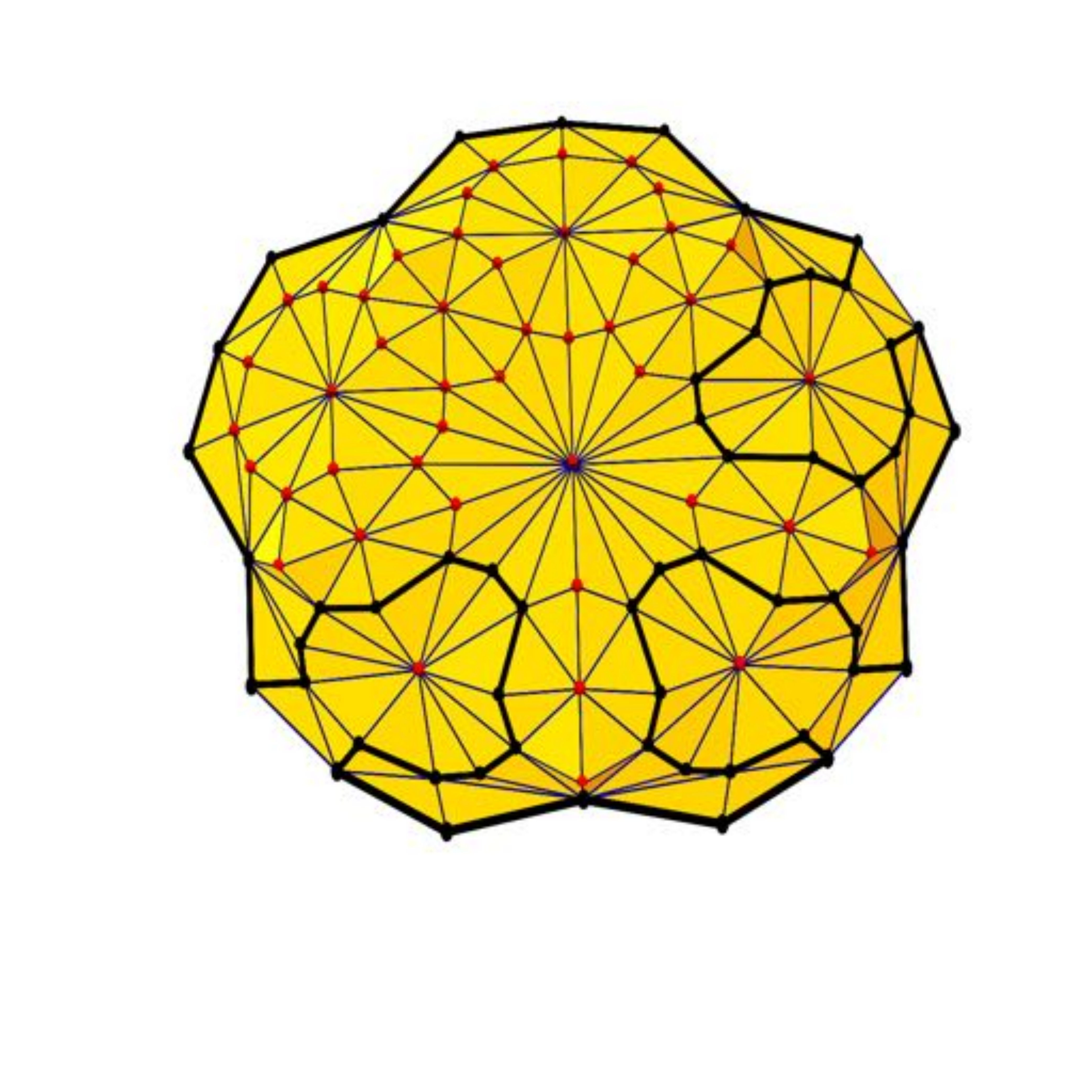}}
\scalebox{0.07}{\includegraphics{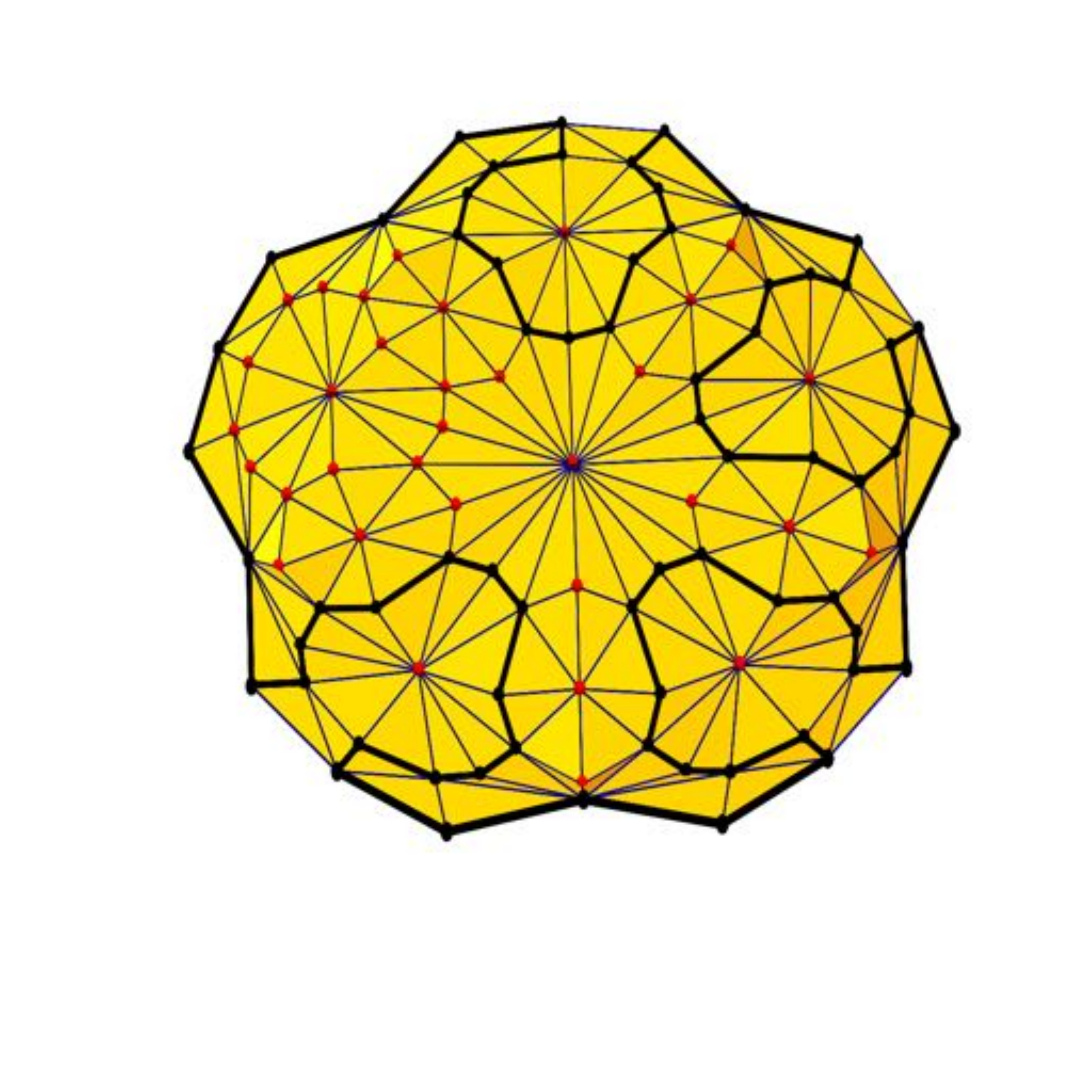}}
\scalebox{0.07}{\includegraphics{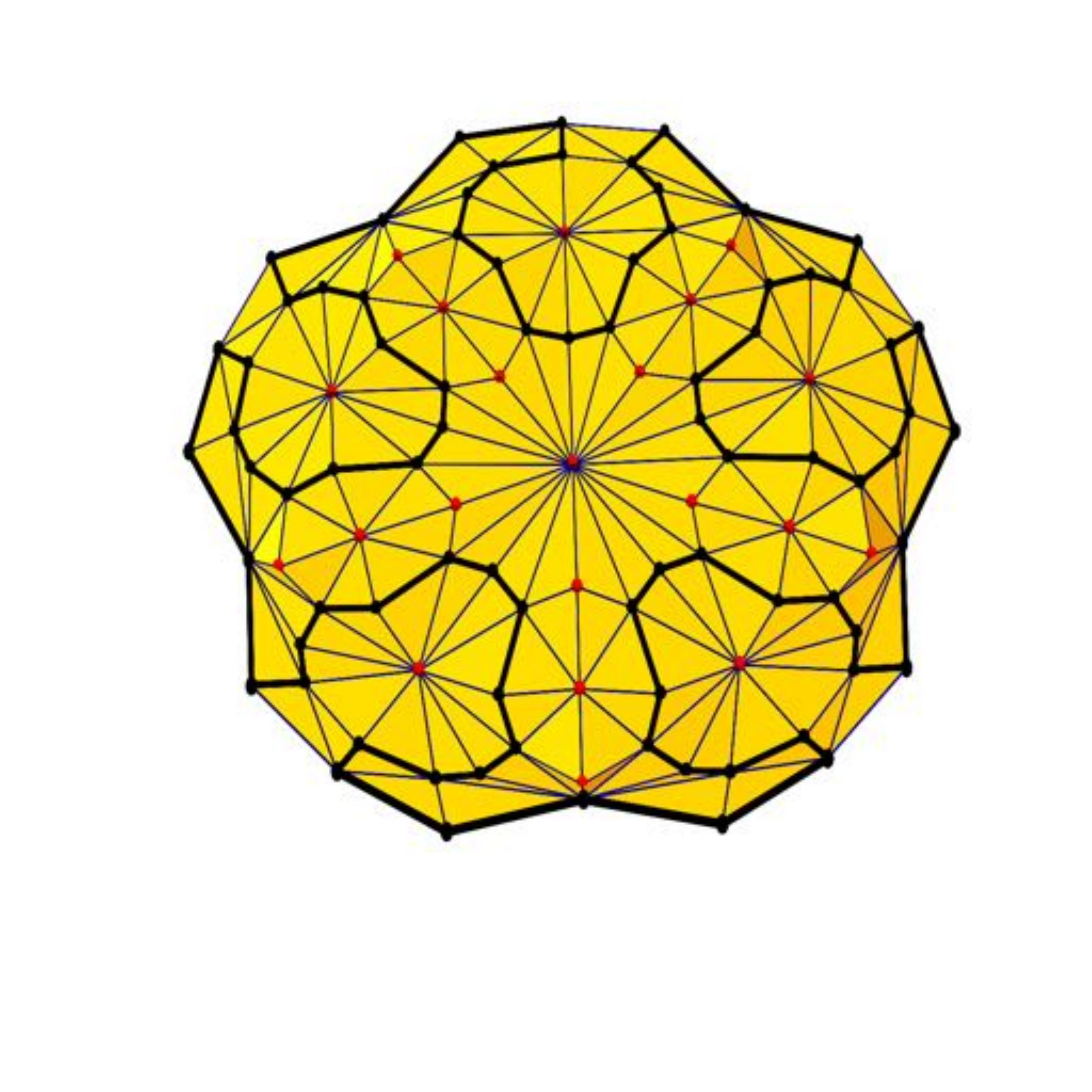}}
\scalebox{0.07}{\includegraphics{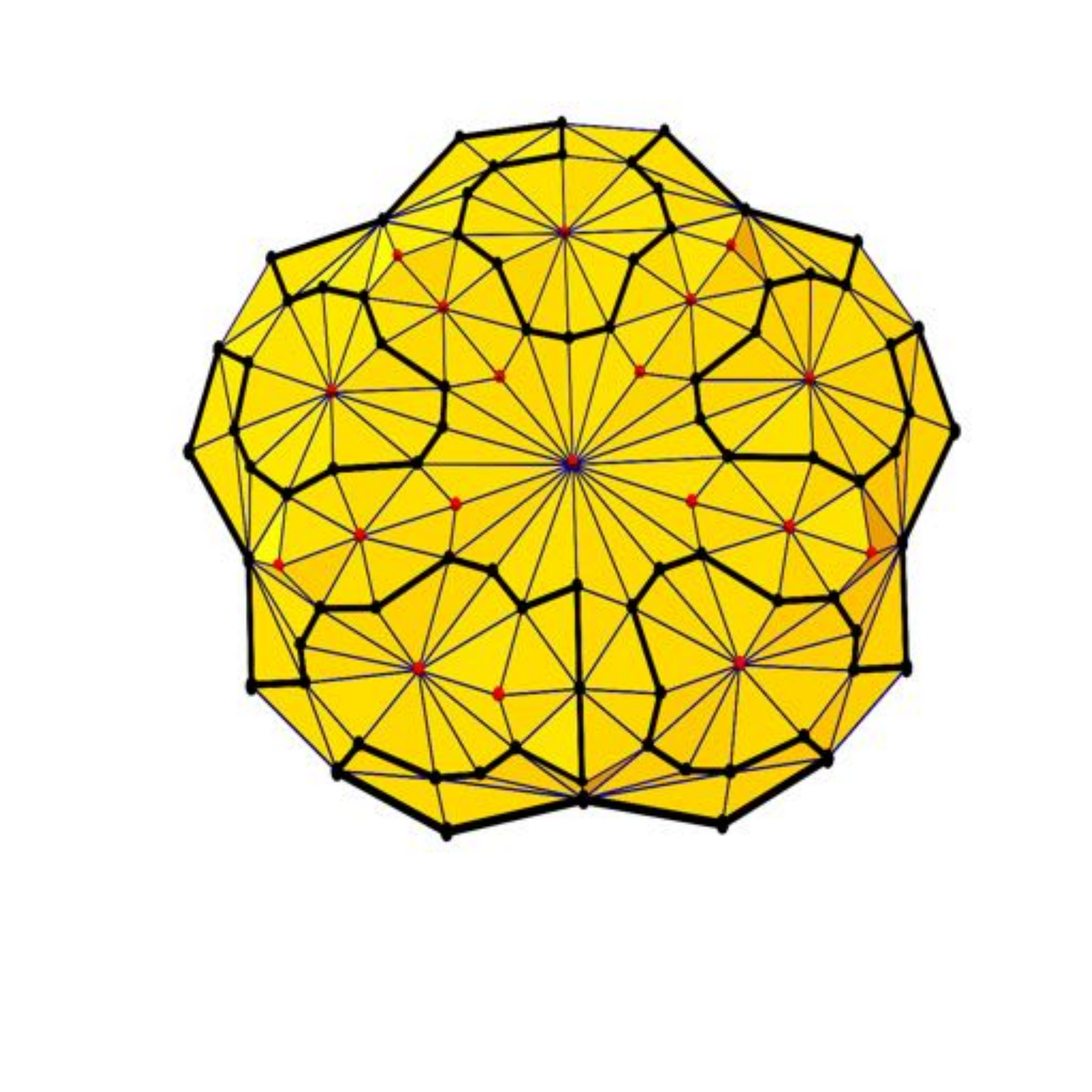}}
\scalebox{0.07}{\includegraphics{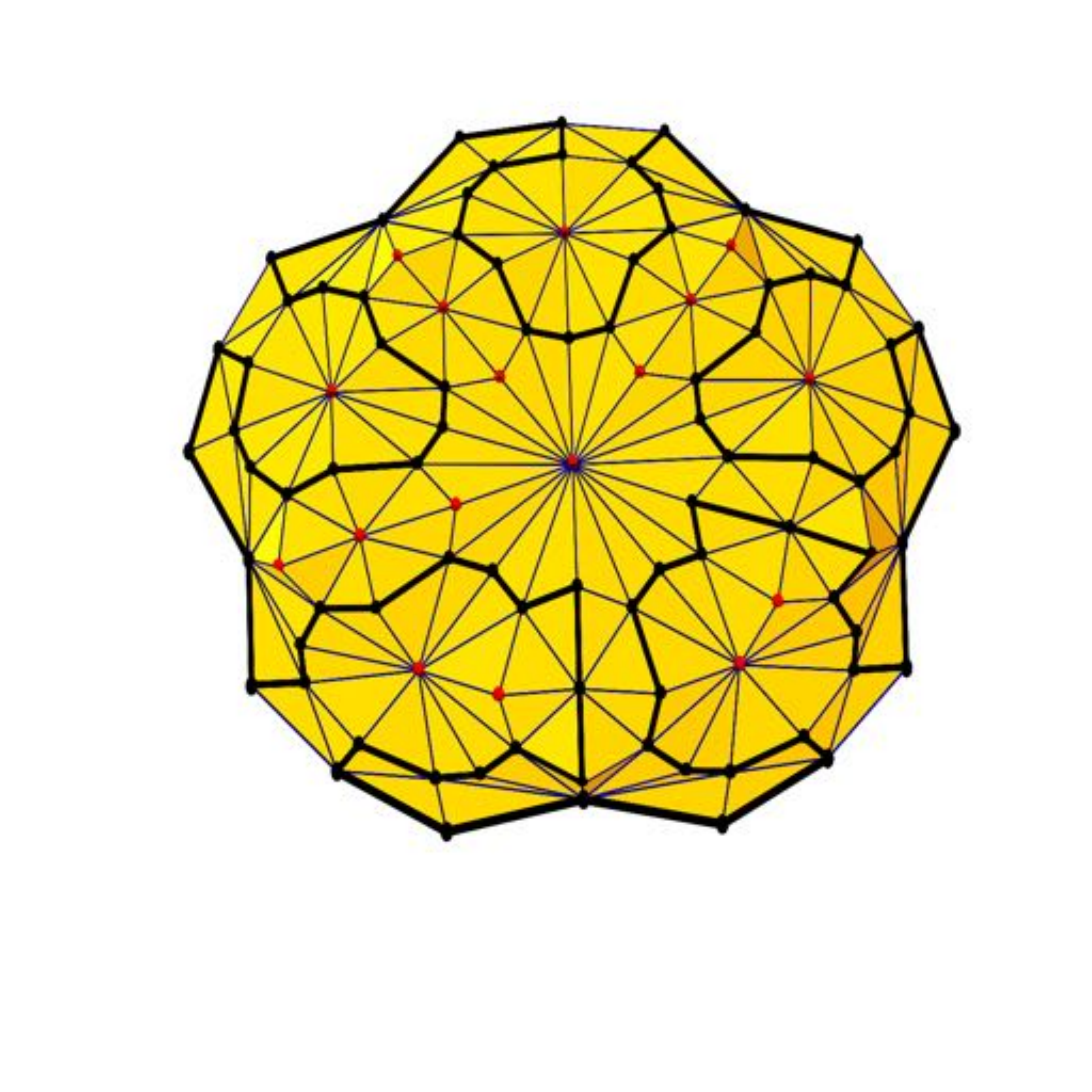}}
\scalebox{0.07}{\includegraphics{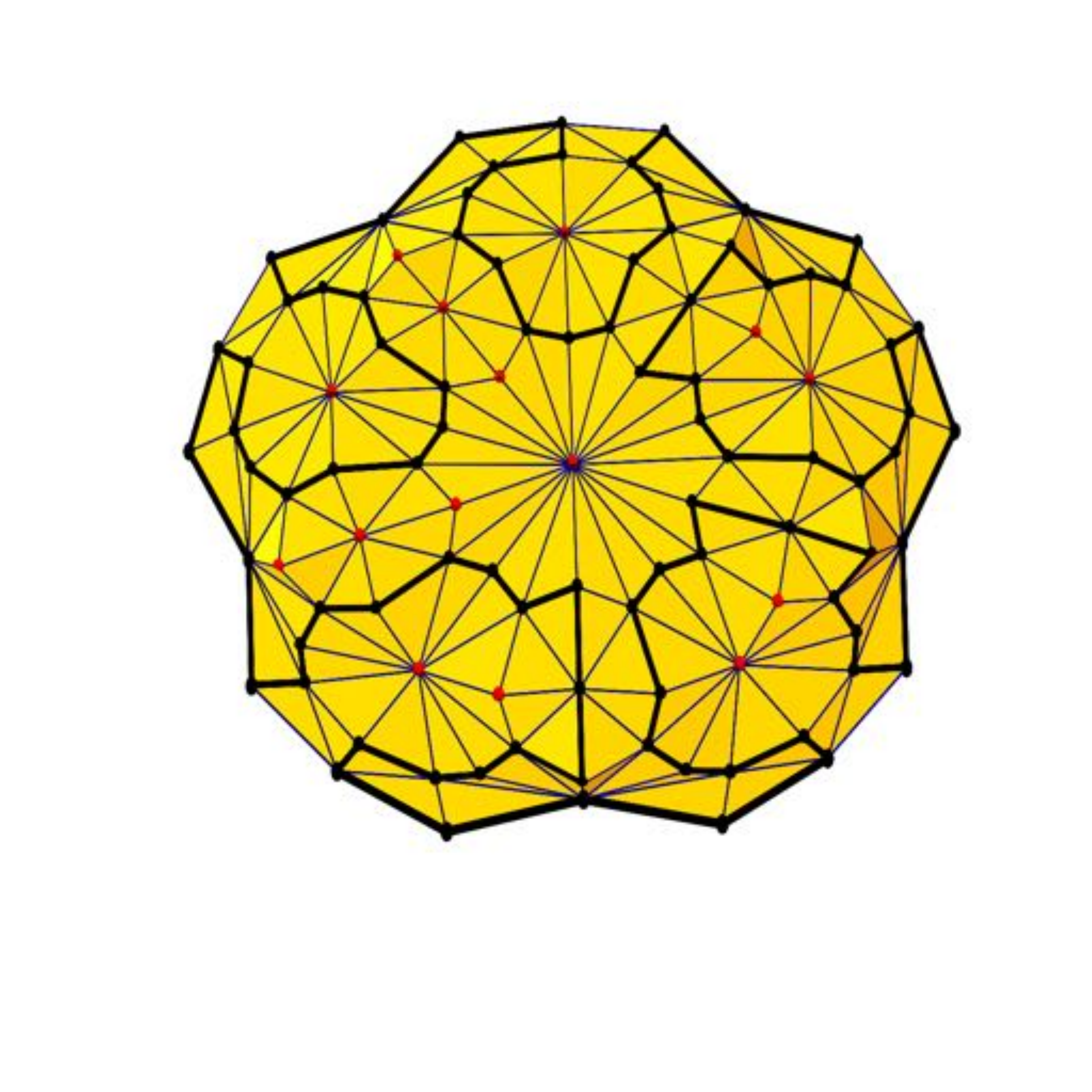}}
\scalebox{0.07}{\includegraphics{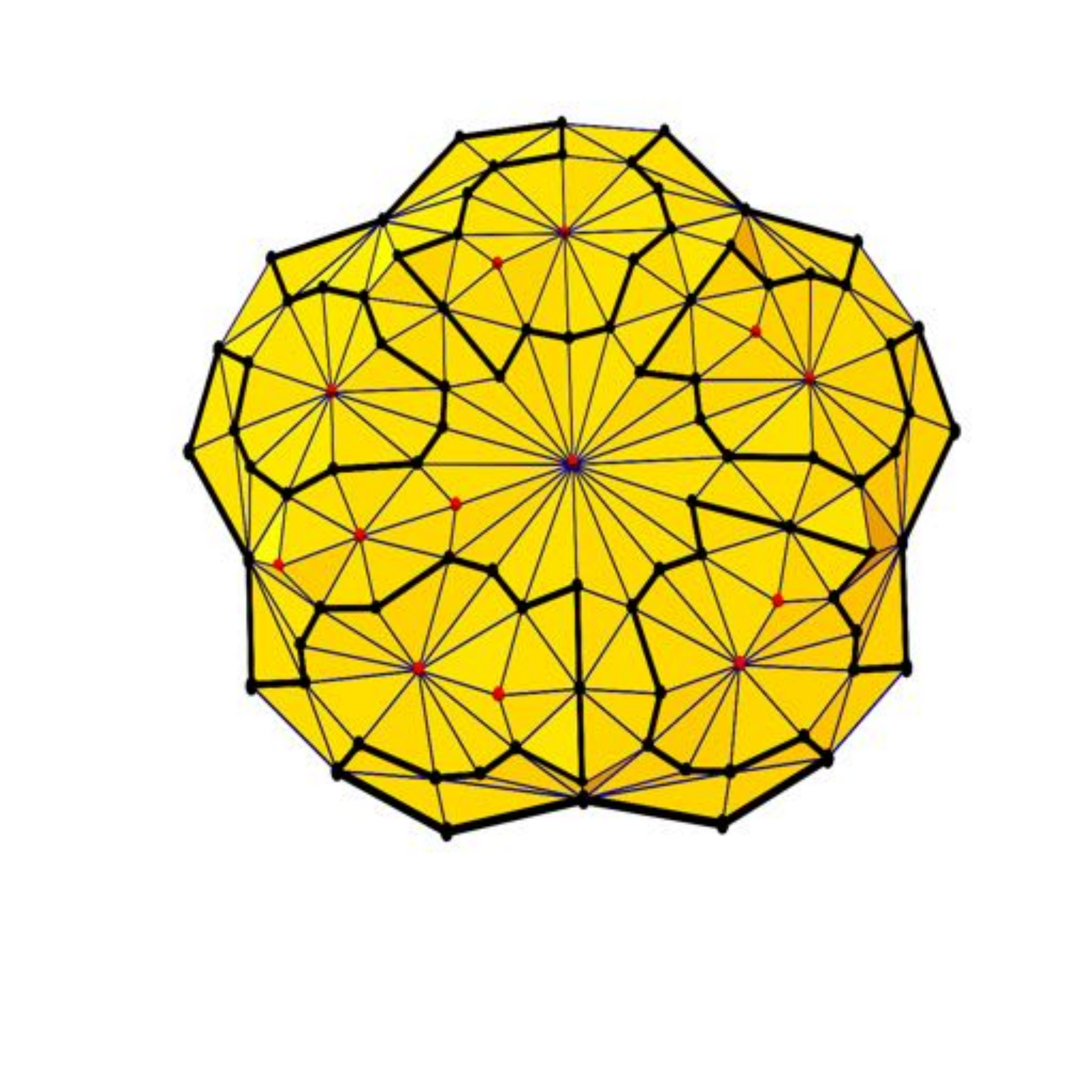}}
\scalebox{0.07}{\includegraphics{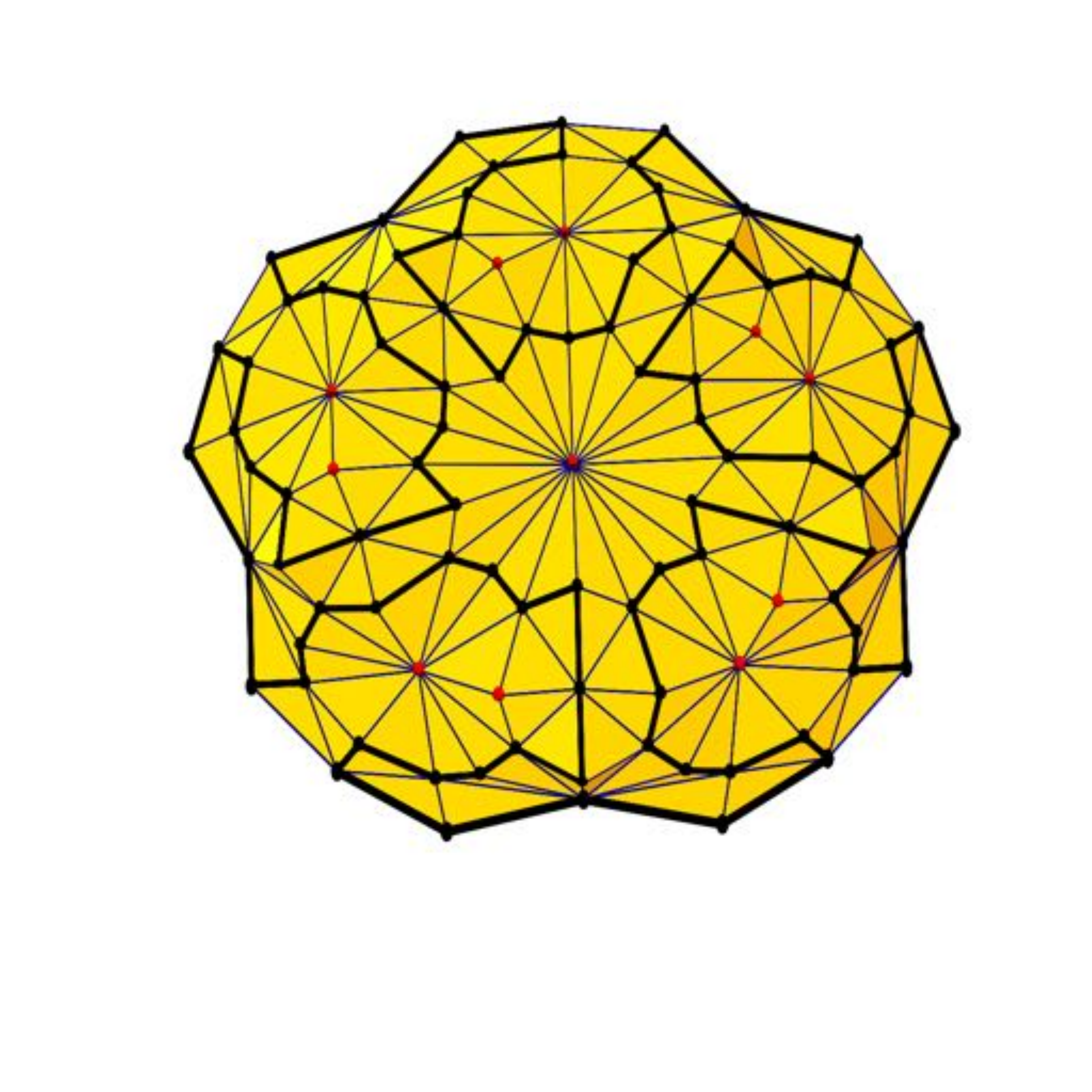}}
\scalebox{0.07}{\includegraphics{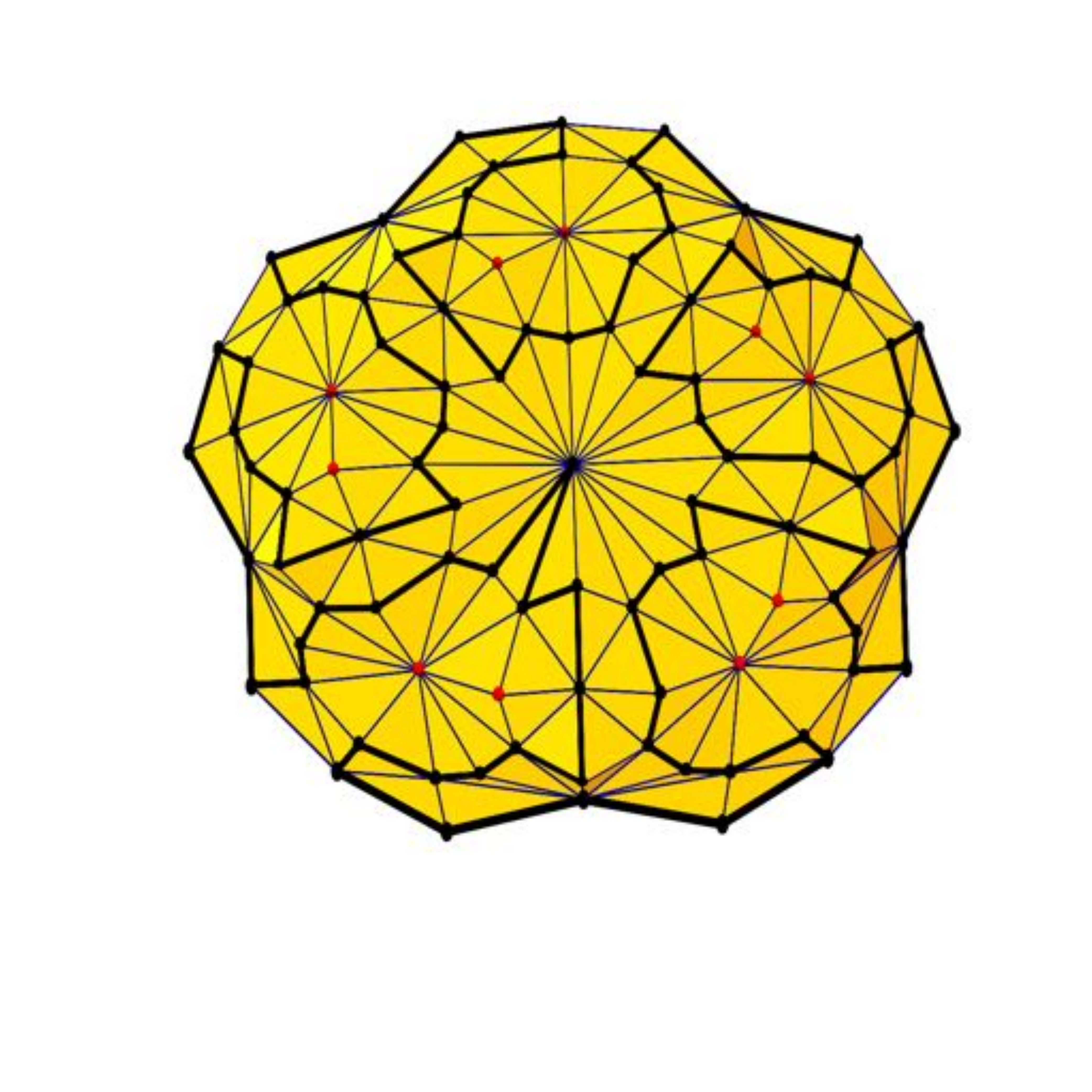}}
\scalebox{0.07}{\includegraphics{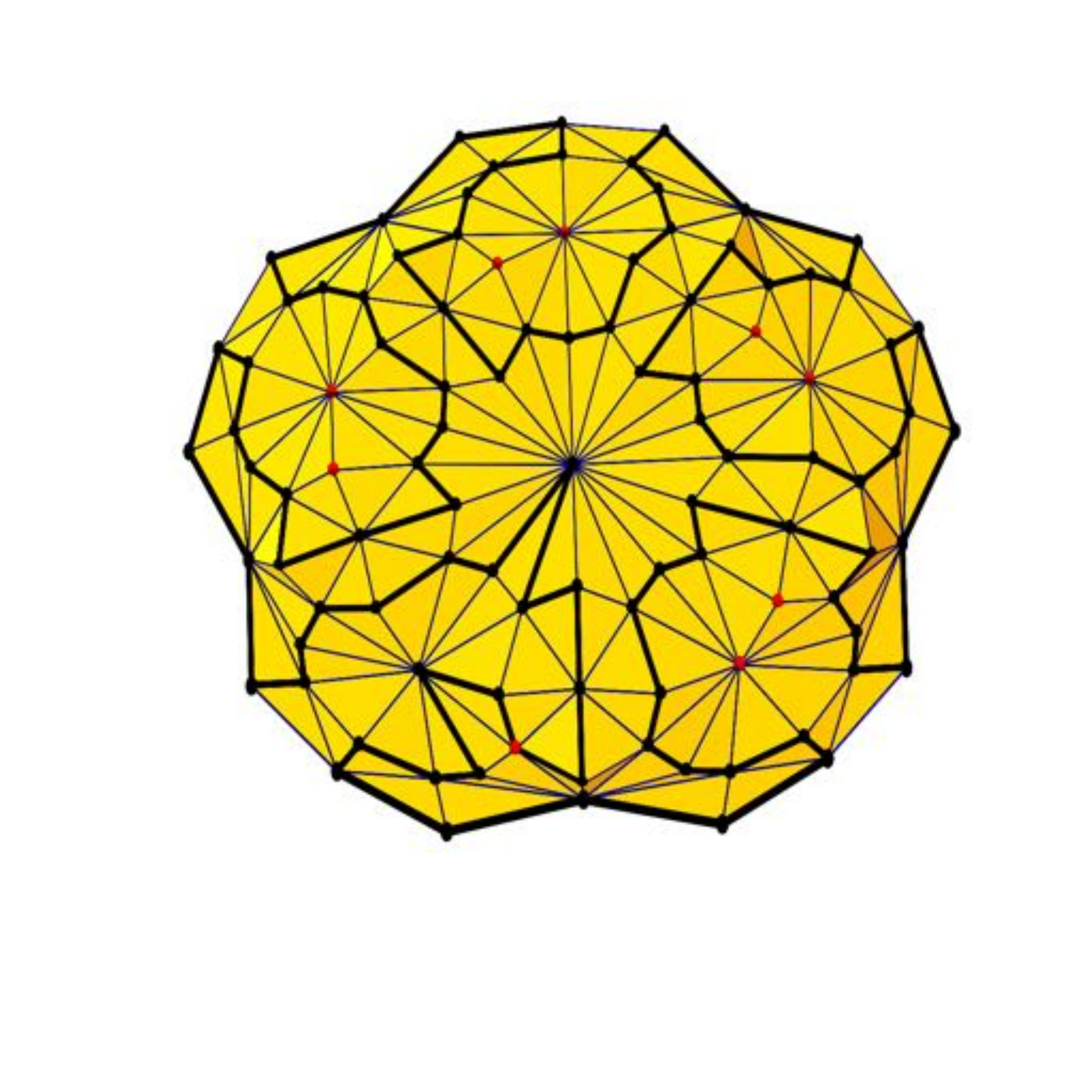}}
\scalebox{0.07}{\includegraphics{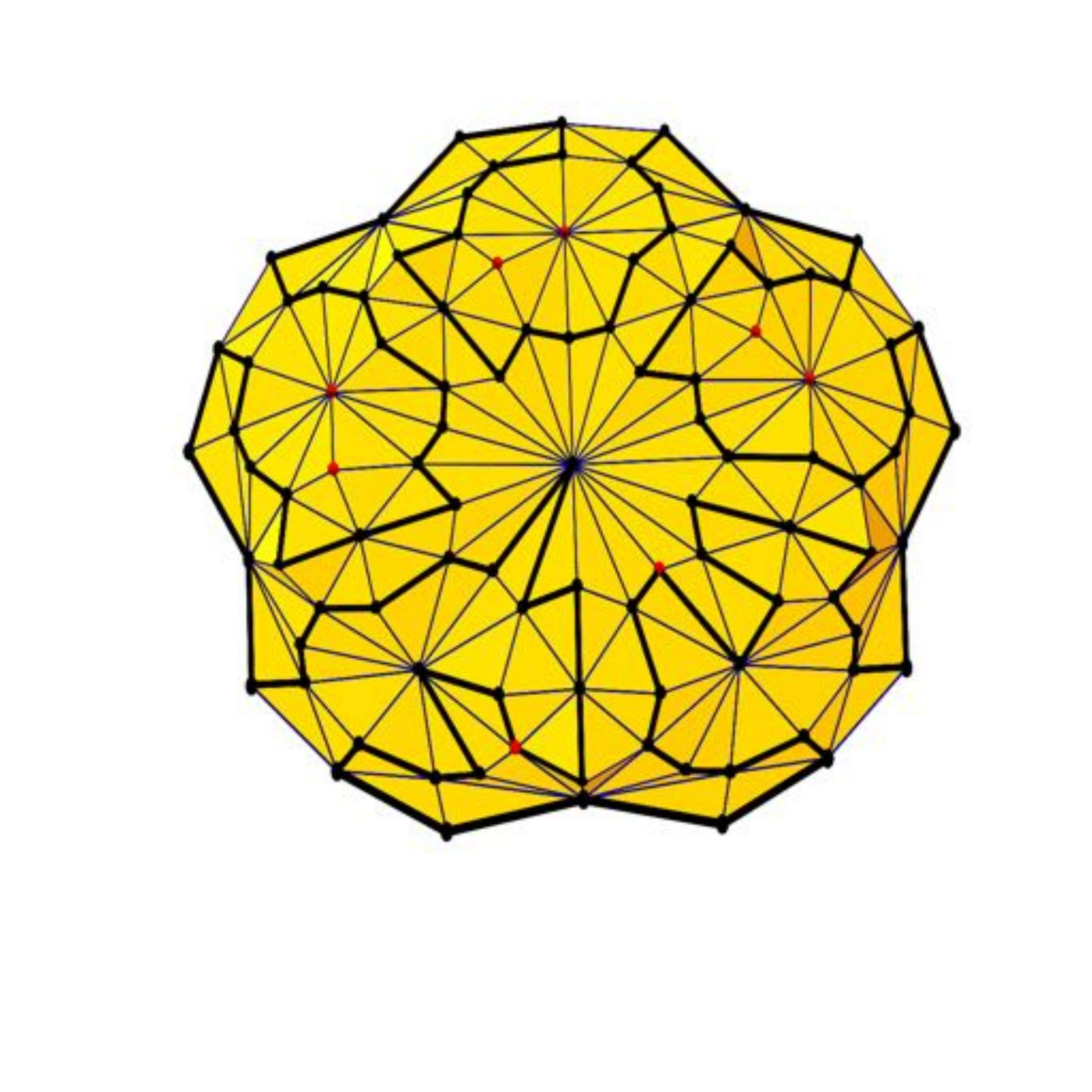}}
\scalebox{0.07}{\includegraphics{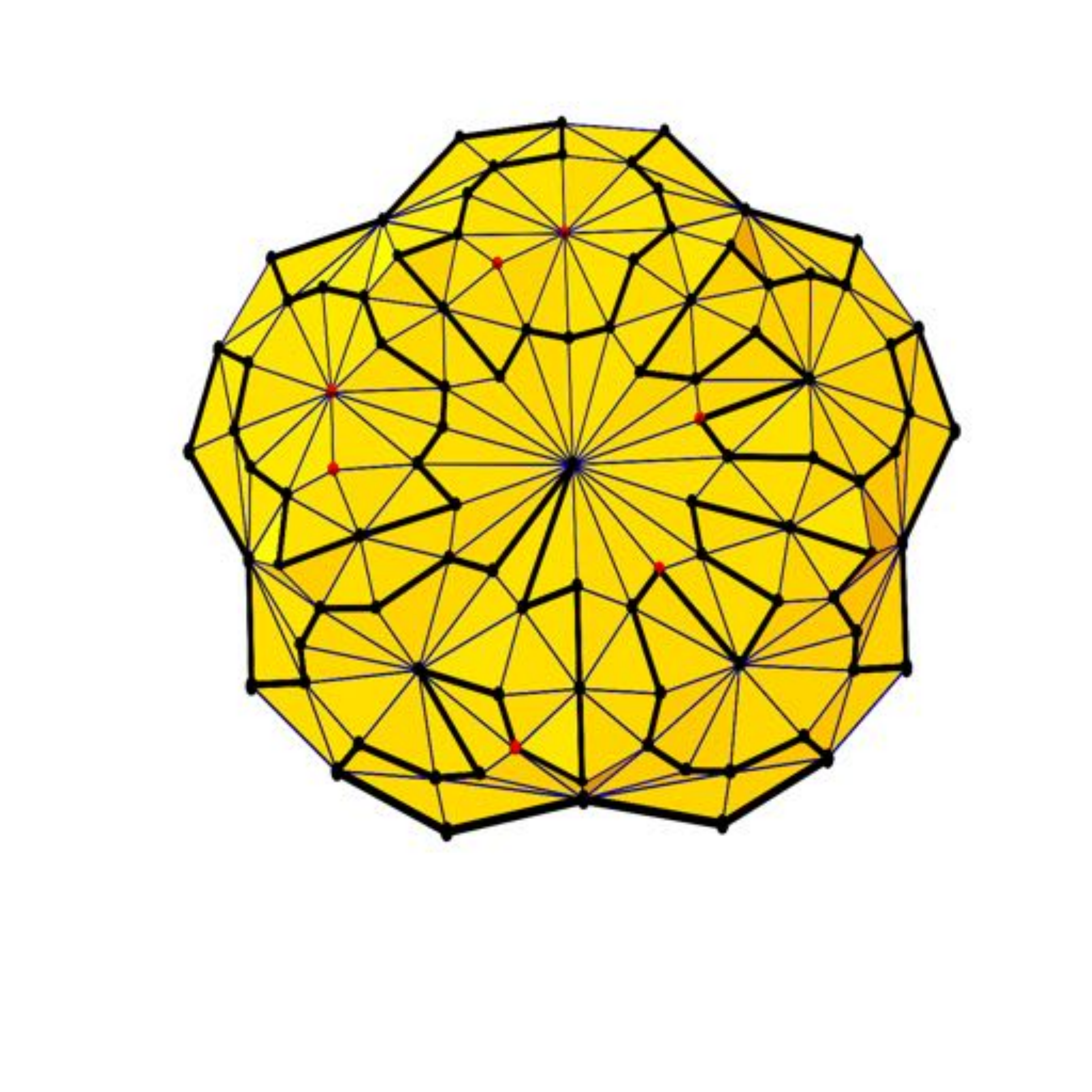}}
\scalebox{0.07}{\includegraphics{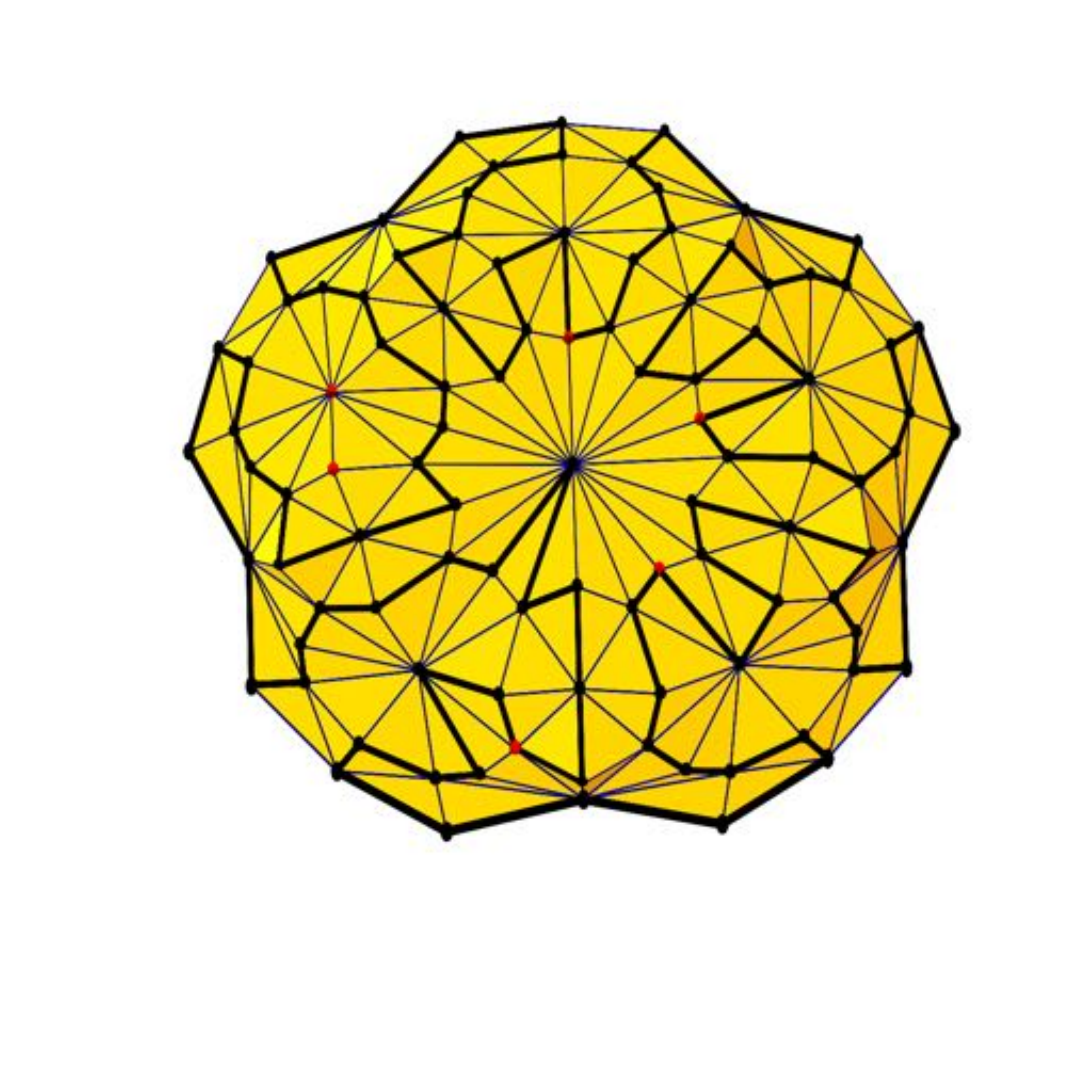}}
\scalebox{0.07}{\includegraphics{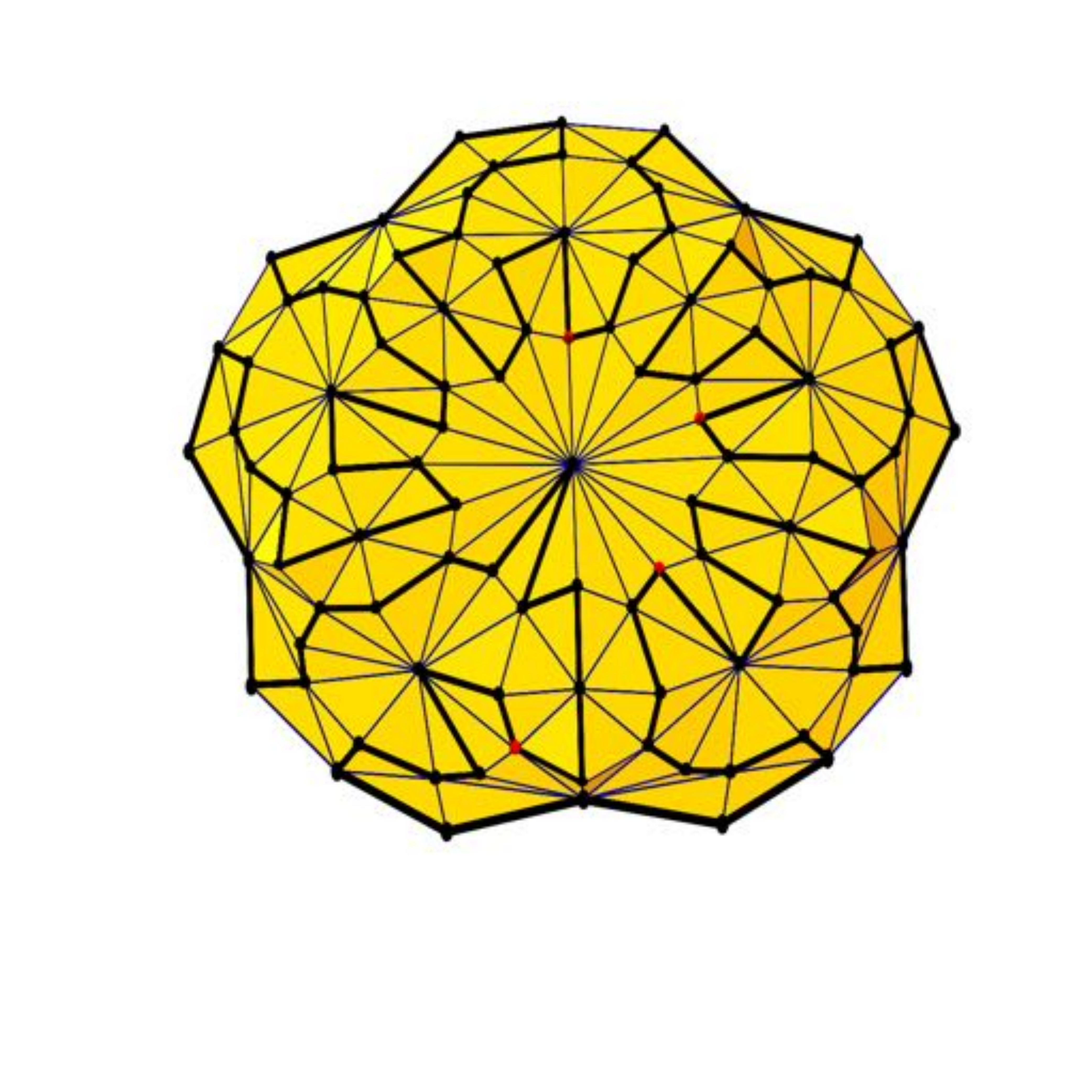}}
\caption{
\label{Strategy}
The Swiss Cheese strategy applied in an example: 
start with covering the boundary with a Hamiltonian path. Then cut out some
holes and connect them to each other. We chose an example which has the difficulty
that after cutting out the hole we have some interior points not reachable. They
are included by enlarging the hole slightly. Finally gather the inside the of the 
holes and make detours to the other remaining interior points. Now we have a 
Hamiltonian path. 
}
\end{figure}

\bibliographystyle{plain}

\end{document}